\documentclass[11pt]{article}

\usepackage[top=1in,bottom=1in,left=1in,right=1in]{geometry}
\usepackage{natbib}
\usepackage[unicode=true,bookmarks=true,bookmarksnumbered=false,bookmarksopen=false,breaklinks=false,pdfborder={0 0 1},backref=false,colorlinks=true]{hyperref}
\hypersetup{citecolor=blue}
\usepackage{url}

\usepackage{amsmath}
\usepackage{mathrsfs}
\usepackage{amssymb}
\usepackage{amsthm}
\usepackage[normalem]{ulem}

\newtheorem{theorem}{Theorem}
\newtheorem{lemma}{Lemma}

\newtheorem{proposition}{Proposition}

\usepackage{float}
\usepackage{rotfloat}

\floatstyle{ruled}
\newfloat{algorithm}{tpb}{loa}
\providecommand{\algorithmname}{Algorithm}
\floatname{algorithm}{\protect\algorithmname}

\usepackage{algpseudocode,algorithm}
\algnewcommand\algorithmicinput{\textbf{Input}:}
\algnewcommand\algorithmicoutput{\textbf{Output}:}
\algnewcommand\INPUT{\item[\algorithmicinput]}
\algnewcommand\OUTPUT{\item[\algorithmicoutput]}

\usepackage{graphicx}
\usepackage[caption=false,labelformat=simple]{subfig}

\usepackage{array}
\newcolumntype{L}[1]{>{\raggedright\let\newline\\\arraybackslash\hspace{0pt}}m{#1}}
\newcolumntype{C}[1]{>{\centering\let\newline\\\arraybackslash\hspace{0pt}}m{#1}}
\newcolumntype{R}[1]{>{\raggedleft\let\newline\\\arraybackslash\hspace{0pt}}m{#1}}

\usepackage{rotating}
\usepackage{multirow}

\usepackage{tikz}
\usetikzlibrary{shapes,arrows,backgrounds,calc,positioning,fit,petri,plotmarks}
\usetikzlibrary{arrows}

\usepackage{enumerate}
\usepackage{paralist}

\usepackage{chngcntr}         

\newcommand*{\affaddr}[1]{#1} 
\newcommand*{\affmark}[1][*]{\textsuperscript{#1}}

\global\long\def\bx{\mathbf{x}}
\global\long\def\bY{\mathbf{Y}}
\global\long\def\by{\mathbf{y}}
\global\long\def\bZ{\mathbf{Z}}
\global\long\def\bz{\mathbf{z}}

\global\long\def\bA{\mathbf{A}}

\global\long\def\bH{\mathbf{H}}

\global\long\def\bU{\mathbf{U}}

\global\long\def\bS{\mathbf{S}}

\global\long\def\bgamma{\boldsymbol{\gamma}}
\global\long\def\bbeta{\boldsymbol{\beta}}
\global\long\def\bpi{\boldsymbol{\pi}}

\title{Principal Regression for High Dimensional Covariance Matrices}

\author{%
    Yi Zhao\affmark[1], Brian S. Caffo\affmark[2], Xi Luo\affmark[3], for the Alzheimer's Disease Neuroimaging
Initiative\footnote{Data used in preparation of this article were obtained from the Alzheimer's Disease Neuroimaging Initiative (ADNI) database (\url{adni.loni.usc.edu}). As such, the investigators within the ADNI contributed to the design and implementation of ADNI and/or provided data but did not participate in analysis or writing of this report. A complete list of ADNI investigators can be found at: \url{http://adni.loni.usc.edu/wp-content/uploads/how_to_apply/ADNI_Acknowledgement_List.pdf}} \\
    \affaddr{\affmark[1]Department of Biostatistics, Indiana University School of Medicine} \\
    \affaddr{\affmark[2]Department of Biostatistics, Johns Hopkins Bloomberg School of Public Health} \\
    \affaddr{\affmark[3]Department of Biostatistics and Data Science,\\  The University of Texas 
Health Science Center at Houston} \\
}

\date{}

\providecommand{\keywords}[1]
{
  {\small	
  \textbf{Keywords:} #1}
}

\begin{document}

\maketitle

\thispagestyle{empty}

\begin{abstract}
This manuscript presents an approach to perform generalized linear regression with multiple high dimensional covariance matrices as the outcome. Model parameters are proposed to be estimated by maximizing a pseudo-likelihood. When the data are high dimensional, the normal likelihood function is ill-posed as the sample covariance matrix is rank-deficient. Thus, a well-conditioned linear shrinkage estimator of the covariance matrix is introduced. With multiple covariance matrices, the shrinkage coefficients are proposed to be common across matrices. Theoretical studies demonstrate that the proposed covariance matrix estimator is optimal achieving the uniformly minimum quadratic loss asymptotically among all linear combinations of the identity matrix and the sample covariance matrix. Under regularity conditions, the proposed estimator of the model parameters is consistent. The superior performance of the proposed approach over existing methods is illustrated through simulation studies. Implemented to a resting-state functional magnetic resonance imaging study acquired from the Alzheimer's Disease Neuroimaging Initiative, the proposed approach identified a brain network within which functional connectivity is significantly associated with Apolipoprotein E $\varepsilon$4, a strong genetic marker for Alzheimer's disease.
\end{abstract}



\clearpage
\setcounter{page}{1}

\section{Introduction}

In this manuscript, we study a regression problem with covariance matrices as the outcome under a high dimensional setting. Suppose $\by_{it}\in\mathbb{R}^{p}$ is a $p$-dimensional random vector, which is the $t$th acquisition from subject $i$, for $t=1,\dots,T_{i}$ and $i=1,\dots,n$, where $T_{i}$ is the number of observations of subject $i$ and $n$ is the number of subjects. Let $T_{\max}=\max_{i}T_{i}$. The high dimensionality refers to the scenario when $T_{\max}\ll p$. The data, $\by_{it}$, are assumed to follow a normal distribution with covariance matrix $\Sigma_{i}$. Here, without loss of generality, it is assumed that the distribution mean is zero as the study interest focuses on the covariance matrices. Let $\bx_{i}\in\mathbb{R}^{q}$ denote the $q$-dimensional covariates of interest acquired from subject $i$. For the covariance matrices, we assume the following regression model, which is considered in \citet{zhao2019covariate}. For $i=1,\dots,n$,
\begin{equation}\label{eq:model}
	\log(\bgamma^\top\Sigma_{i}\bgamma)=\bx_{i}^\top\bbeta,
\end{equation}
where $\bgamma\in\mathbb{R}^{p}$ is a linear projection, and $\bbeta\in\mathbb{R}^{q}$ is the model coefficient. In $\bx_{i}$, the first element is set to one to include the intercept term. The goal is to estimate $\bgamma$ and $\bbeta$ using the observed data $\left\{(\by_{i1},\dots,\by_{iT_{i}}),\bx_{i}\right\}_{i=1}^{n}$.

One application of such a regression problem is to analyze covariate associated variations in brain coactivation in a functional magnetic resonance imaging (fMRI) study, where covariance/correlation matrices of the fMRI signals are generally utilized to reveal the coactivation patterns. Characterizing these patterns with population/individual covariates is of great interest in neuroimaging studies~\citep{seiler2017multivariate,zhao2019covariate}. Another example is the study of financial stock market data. Considering a pool of stock returns, covariance matrices over a period of time capture the comovement or synchronicity of the stocks. Firm and market-level information, such as industry type, firm's cash flow, stock size, and book-to-market ratio, plays an essential role in determining the synchronicity. Quantifying such association is an important topic in financial theory~\citep{zou2017covariance}.

To estimate $\bgamma$ and $\bbeta$, \citet{zhao2019covariate} proposed a likelihood-based approach, that is to minimize the negative log-likelihood function in the projection space.
One sufficient condition to solve the likelihood-based criterion is that the sample covariance matrices are positive definite. Thus, the likelihood estimator is ill-posed when $T_{\max}<p$ as the sample covariance matrices are rank-deficient. Additionally, it has been shown that when $p$ increases, the sample covariance matrix performs poorly and can lead to invalid conclusions. For example, the largest eigenvalue of the sample covariance matrix is not a consistent estimator, and the eigenvectors can be nearly orthogonal to the truth~\citep{johnstone2009consistency}. To circumvent difficulties raised by the high dimensionality, one solution is to impose structural assumptions, such as bandable covariance matrices, sparse covariance matrices, spiked covariance matrices, covariances with a tensor product structure, and latent graphical models \cite[see a review of][and references therein]{cai2016structured}. Another class of high-dimensional covariance matrix estimator is the shrinkage estimator. \citet{daniels2001shrinkage} considered two shrinkage estimators of the covariance matrix, a correlation shrinkage and a rotation shrinkage, offering a compromise between completely unstructured and structured estimators to improve the robustness. 
\citet{ledoit2004well} introduced a well-conditioned estimator of the covariance matrix, which is an optimal linear combination of the identity matrix and the sample covariance matrix under squared error loss. This is equivalent to the optimal linear shrinkage of the eigenvalues while retaining the eigenvectors. Instead of a linear combination, \citet{ledoit2012nonlinear} extended this work to nonlinear transformations of the sample eigenvalues and presented a way of finding the transformation that is asymptotically equivalent to the oracle linear combination. Based on Tyler's robust $M$-estimator~\citep{tyler1987distribution} and the linear shrinkage estimator~\citep{ledoit2004well}, \citet{chen2011robust} and \citet{pascal2014generalized}, in parallel, introduced robust estimators  of covariance matrices for elliptical distributed samples.

To model multiple covariance matrices, procedures include regression-type approaches introduced by \citet{anderson1973asymptotically}, \citet{chiu1996matrix}, \citet{hoff2012covariance}, \citet{fox2015bayesian}, and \citet{zou2017covariance}; (common) principal component analysis related methods by \citet{flury1984common}, \citet{boik2002spectral}, \citet{hoff2009hierarchical}, and \citet{franks2019shared}; and methods based on other types of matrix decomposition, such as the Cholesky decomposition~\citep{pourahmadi2007simultaneous}. Among these, \citet{fox2015bayesian} introduced a scalable nonparametric covariance regression model applying low-rank approximation. \citet{franks2019shared} generalized a Bayesian hierarchical model studying the heterogeneity in the covariance matrices to high dimensional settings. Compared to the above-mentioned approaches, Model~\eqref{eq:model} offers higher flexibility in modeling the relationship with the covariates. For example, $x$ can be either continuous or categorical, and one can easily include interactions and/or polynomials of the covariates.

In the high dimensional setting considered in this study, $\bgamma$ and $\bbeta$, as well as $n$ covariance matrices will be estimated under Model~\eqref{eq:model}. Because of its computational efficiency and explicit formulations of the tuning parameters, the linear shrinkage approach will be generalized to multiple covariance estimates. Interestingly, it will be shown that estimating each covariance matrix separately, such as using the shrinkage estimator proposed in \citet{ledoit2004well}, leads to suboptimal estimation accuracy for $\bgamma$, $\bbeta$ and $\Sigma_{i}$'s.
Thus, a linear shrinkage estimator of the covariance matrix is proposed, of which the shrinkage coefficients are shared across matrices. With the shrinkage estimator, it is proposed to estimate $(\bgamma,\bbeta)$ through maximizing a pseudo-likelihood.

The framework proposed in this manuscript has three major contributions. 
\begin{enumerate}[(1)]
	\item This is probably the first attempt to analyze a large number of high-dimensional covariance matrices varying with covariates in a regression setting.
	\item The proposed shrinkage estimator of the covariance matrices is well-conditioned and has uniformly minimum quadratic risk asymptotically among all linear combinations.
	\item Under regularity conditions, the proposed approach achieves consistent estimators of the parameters.
\end{enumerate}

The rest of the paper is organized as the following. Section~\ref{sec:method} introduces the proposed shrinkage estimator of the covariance matrices and the pseudo-likelihood based method of estimating $\bgamma$ and $\bbeta$. Section~\ref{sec:asmp} studies the asymptotic properties. In Section~\ref{sec:sim}, the superior performance of the proposed approach over existing methods is demonstrated through simulation studies. Section~\ref{sec:real} articulates an application to a resting-state fMRI data set acquired from the Alzheimer's Disease Neuroimaging Initiative (ADNI). Section~\ref{sec:discussion} concludes this paper with discussions. Technical proofs are collected in the supplementary materials.


\section{Method}
\label{sec:method}

Considering the regression model~\eqref{eq:model}, it is proposed to estimate the parameters by solving the following optimization problem.
\begin{eqnarray}\label{eq:opt_Sstar}
	\underset{(\bbeta,\bgamma)}{\text{minimize}} && \ell(\bbeta,\bgamma)=\frac{1}{2}\sum_{i=1}^{n}T_{i}\left\{\bx_{i}^\top\bbeta+\bgamma^\top\hat{\Sigma}_{i}\bgamma\cdot\exp(-\bx_{i}^\top\bbeta)\right\}, \nonumber \\
	\text{such that} && \bgamma^\top\bH\bgamma=1,
\end{eqnarray}
where $\hat{\Sigma}_{i}$ is an estimator of the covariance matrix $\Sigma_{i}$ to be discussed later, which is positive definite, for $i=1,\dots,n$; and $\bH$ is a positive definite matrix in $\mathbb{R}^{p\times p}$, which is set to be the average of $\hat{\Sigma}_{i}$'s, that is $\bH=\sum_{i=1}^{n}T_{i}\hat{\Sigma}_{i}/\sum_{i=1}^{n}T_{i}$. It is essential to impose a constraint on $\bgamma$, otherwise the objective function of \eqref{eq:opt_Sstar} is minimized at $\bgamma=\boldsymbol{\mathrm{0}}$ with fixed $\bbeta$. When $\hat{\Sigma}_{i}=\bS_{i}=\sum_{t=1}^{T_{i}}\by_{it}\by_{it}^\top/T_{i}$ (i.e., the sample covariance matrix), which is the proposal in \citet{zhao2019covariate}, it is equivalent to minimize the negative log-likelihood function of $\{\bgamma^\top\by_{it}\}_{i,t}$ assuming the data are normally distributed. However, when $T_{\max}=\max_{i}T_{i}<p$, problem~\eqref{eq:opt_Sstar} is ill-posed as $\bS_{i}$'s are rank-deficient. Thus, the goal of this manuscript is to propose a well-conditioned estimator of $\Sigma_{i}$ that yields optimal properties. To achieve this, a covariate-dependent linear shrinkage estimator, denoted as $\Sigma_{i}^{*}$, is proposed, which yields the minimum expected squared loss under regression model~\eqref{eq:model}, where the expectation is taken over the sample covariance matrix $\bS_{i}$.
\begin{eqnarray}\label{eq:opt_cov}
	\underset{(\mu,\rho)}{\text{minimize}} && \frac{1}{n}\sum_{i=1}^{n}\mathbb{E}\left\{\bgamma^\top\Sigma_{i}^{*}\bgamma-\exp(\bx_{i}^\top\bbeta)\right\}^{2}, \nonumber \\
	\text{such that} && \Sigma_{i}^{*}=\rho\mu\boldsymbol{\mathrm{I}}+(1-\rho)\bS_{i}, \quad \text{for } i=1,\dots,n.
\end{eqnarray}
The following theorem gives the solution to~\eqref{eq:opt_cov}.
\begin{theorem}\label{thm:Sigma_solution}
	For given $(\bgamma,\bbeta)$, the solution to optimization problem~\eqref{eq:opt_cov} is
	\begin{equation}\label{eq:sol_SigmaS}
		\Sigma_{i}^{*}=\frac{\psi^{2}}{\delta^{2}}\mu\boldsymbol{\mathrm{I}}+\frac{\phi^{2}}{\delta^{2}}\bS_{i}, \quad \text{for } i=1,\dots,n,
	\end{equation}
	and the minimum value is
	\begin{equation}
		\frac{1}{n}\sum_{i=1}^{n}\mathbb{E}\left\{\bgamma^\top\Sigma_{i}^{*}\bgamma-\exp(\bx_{i}^\top\bbeta)\right\}^{2}=\frac{\phi^{2}\psi^{2}}{\delta^{2}},
	\end{equation}
	where
	\[
		\mu=\frac{1}{n(\bgamma^\top\bgamma)}\sum_{i=1}^{n}\exp(\bx_{i}^\top\bbeta), \quad \phi^{2}=\frac{1}{n}\sum_{i=1}^{n}\phi_{i}^{2}, \quad \psi^{2}=\frac{1}{n}\sum_{i=1}^{n}\psi_{i}^{2}, \quad \delta^{2}=\frac{1}{n}\sum_{i=1}^{n}\delta_{i}^{2},
	\]
	\[
		\phi_{i}^{2}=\left\{\mu(\bgamma^\top\bgamma)-\exp(\bx_{i}^\top\bbeta)\right\}^{2}, \quad \psi_{i}^{2}=\mathbb{E}\left\{\bgamma^\top\bS_{i}\bgamma-\exp(\bx_{i}^\top\bbeta)\right\}^{2}, \quad \delta_{i}^{2}=\mathbb{E}\left\{\bgamma^\top\bS_{i}\bgamma-\mu(\bgamma^\top\bgamma)\right\}^{2};
	\]
	and Lemma~\ref{lemma:par_relation} shows that $\psi^{2}/\delta^{2}+\phi^{2}/\delta^{2}=1$.
\end{theorem}
\begin{lemma}\label{lemma:par_relation}
	For $\forall~i\in\{1,\dots,n\}$, $\delta_{i}^{2}=\phi_{i}^{2}+\psi_{i}^{2}$, and thus $\delta^{2}=\phi^{2}+\psi^{2}$.
\end{lemma}
\noindent According to Theorem~\ref{thm:Sigma_solution}, parameters $\phi_{i}^{2}$, $\psi_{i}^{2}$ and $\delta_{i}^{2}$ are expected values as the objective is to minimize the expected squared loss. Thus, one cannot replace $\hat{\Sigma}_{i}$ with $\Sigma_{i}^{*}$ in \eqref{eq:opt_Sstar} and solve for solution using the data. For implementation in practice, the following sample counterparts are used to compute~\eqref{eq:sol_SigmaS} and thus $\hat{\Sigma}_{i}$ in~\eqref{eq:opt_Sstar}. Let
	\[
		\hat{\delta}_{i}^{2}=\left\{\bgamma^\top\bS_{i}\bgamma-\mu(\bgamma^\top\bgamma)\right\}^{2}, \quad \hat{\psi}_{i}^{2}=\frac{1}{T_{i}}\left\{\bgamma^\top\bS_{i}\bgamma-\exp(\bx_{i}^\top\bbeta)\right\}^{2}, \quad \hat{\phi}_{i}^{2}=\hat{\delta}_{i}^{2}-\hat{\psi}_{i}^{2},
	\]
	\[
		\hat{\delta}^{2}=\frac{1}{n}\sum_{i=1}^{n}\hat{\delta}_{i}^{2}, \quad \hat{\psi}^{2}=\frac{1}{n}\sum_{i=1}^{n}\min(\hat{\psi}_{i}^{2},\hat{\delta}_{i}^{2}), \quad \hat{\phi}^{2}=\frac{1}{n}\sum_{i=1}^{n}\hat{\phi}_{i}^{2},
	\]
and
	\begin{equation}\label{eq:S_gbshrink}
		\bS_{i}^{*}=\frac{\hat{\psi}^{2}}{\hat{\delta}^{2}}\mu\boldsymbol{\mathrm{I}}+\frac{\hat{\phi}^{2}}{\hat{\delta}^{2}}\bS_{i}, \quad \text{for }i=1,\dots,n.
	\end{equation}
In Section~\ref{sec:asmp}, we show that $\bS_{i}^{*}$ is a consistent estimator of $\Sigma_{i}^{*}$ and is uniformly optimal asymptotically among all the linear combinations of the sample covariance matrices and the identity matrix regarding the quadratic risk. The objective function $\ell(\bbeta,\bgamma)$ is an approximation of the negative log-likelihood function if replacing $\hat{\Sigma}_{i}$ with the proposed shrinkage estimator $\bS_{i}^{*}$. Thus, optimizing~\eqref{eq:opt_Sstar} can be considered as a pseudo-likelihood approach under the normality assumption.

The proof of Theorem~\ref{thm:Sigma_solution} and Lemma~\ref{lemma:par_relation} is presented in Section~\ref{appendix:sub:proof_SigmaSolution} of the supplementary materials. Formulation~\eqref{eq:opt_cov} introduces a shrinkage estimator of the covariance matrix, where the shrinkage is shared across subjects and is optimal under the squared error loss. For each subject, $\Sigma_{i}^{*}$ is a linear combination of the sample covariance matrix $\bS_{i}$ and the identity matrix. The weighting parameters, $\rho$ and $\mu$, are population level parameters that are shared across subjects. This is equivalent to imposing a linear shrinkage on the sample eigenvalues. Assuming $\bgamma$ is a common eigenvector of all the covariance matrices, $\mu$ is the average eigenvalue corresponding to $\bgamma$. The level of shrinkage is determined by the leverage between the accuracy of $\bS_{i}$'s and the variation in the eigenvalues. If $\bS_{i}$'s are accurate or the errors are small relative to the variation in the eigenvalues, less shrinkage will be imposed; otherwise, if $\bS_{i}$'s are inaccurate and the errors are comparable or even higher than the eigenvalue variability, the sample covariance matrices will be shrank more. 

Algorithm~\ref{alg:covreg} summarizes the optimization procedure. As problem~\eqref{eq:opt_Sstar} is nonconvex, a series of random initializations of $(\bgamma,\bbeta)$ is considered and the one that achieves the minimum value of the objective function is the estimate. The initial values of $\bgamma$ can be set as the eigenvectors of the average sample covariance matrices, $\bar{\bS}=\sum_{i=1}^{n}T_{i}\bS_{i}/\sum_{i=1}^{n}T_{i}$; and the initial values of $\bbeta$ is the corresponding solution to~\eqref{eq:opt_Sstar} by replacing $\hat{\Sigma}_{i}$ with a well-conditioned estimator, such as the estimator proposed in \citet{ledoit2004well}. When $p<\sum_{i=1}^{n}T_{i}$, $\bar{\bS}$ is of full rank, and the sample eigenvectors are consistent estimators assuming all the covariance matrices have the same eigendecomposition. Step 3 in the algorithm updates the covariance matrix estimators with a global shrinkage parameter. In Section~\ref{sec:sim}, through simulation studies, we show that it improves the performance in estimating the covariance matrices and $\bbeta$ with lower bias and higher stability. The details of updating $\bgamma$ and $\bbeta$ in Step 4 can be found in Algorithm 1 in \citet{zhao2019covariate}. 

For higher-order components, one can first remove the identified components and use the new data to estimate the next with an additional orthogonality constraint, that is, the new component is orthogonal to the identified ones. Different from Algorithm 2 in \citet{zhao2019covariate}, there is no need to include a rank-completion step as $\bS_{i}^{*}$ is introduced to render the rank-deficiency issue. To determine the number of components, the metric of average deviation from diagonality proposed in \citet{zhao2019covariate} is adopted. Let $\Gamma^{(k)}\in\mathbb{R}^{p\times k}$ denote the first $k$ estimated components, the average deviation from diagonality is defined as
\begin{equation}
	\mathrm{DfD}(\Gamma^{(k)})=\prod_{i=1}^{n}\left(\frac{\det\{\mathrm{diag}(\Gamma^{(k)\top}\bS_{i}^{*}\Gamma^{(k)})\}}{\det(\Gamma^{(k)\top}\bS_{i}^{*}\Gamma^{(k)})}\right)^{T_{i}/\sum_{i}T_{i}},
\end{equation}
where $\mathrm{diag}(\bA)$ is a diagonal matrix of the diagonal elements in a square matrix $\bA$, and $\det(\bA)$ is the determinant of $\bA$. If $\Gamma^{(k)}$ is a common diagonalization of $\bS_{i}^{*}$'s, that is, $\Gamma^{(k)\top}\bS_{i}^{*}\Gamma^{(k)}$ is a diagonal matrix, for $\forall~i=1,\dots,n$, then $\mathrm{DfD}(\Gamma^{(k)})=1$. In practice, $k$ can be chosen before $\mathrm{DfD}$ increases far away from one or before a sudden jump occurs.
\begin{algorithm}
	\caption{\label{alg:covreg}The optimization algorithm for problems~\eqref{eq:opt_Sstar} and~\eqref{eq:opt_cov}.}
	\begin{algorithmic}[1]
		\INPUT $\{(\by_{i1},\dots,\by_{iT_{i}}),\bx_{i}\}_{i=1}^{n}$

		\State \textbf{initialization}: $(\bgamma^{(0)},\bbeta^{(0)})$

		\Repeat \; for iteration $s=0,1,2,\dots$

		\State  \; for $i=1,\dots,n$, update
			\[
				\bS_{i}^{*(s+1)}=\frac{\hat{\psi}^{2(s)}}{\hat{\delta}^{2(s)}}\mu^{(s)}\boldsymbol{\mathrm{I}}+\frac{\hat{\phi}^{2(s)}}{\hat{\delta}^{2(s)}}\bS_{i},
			\]
			\; \; \; \; where $(\hat{\psi}^{2},\hat{\phi}^{2},\hat{\delta}^{2},\mu)$ are set to the value with $\bgamma=\bgamma^{(s)}$ and $\bbeta=\bbeta^{(s)}$,

		\State \; update $\bgamma$ and $\bbeta$ by solving~\eqref{eq:opt_Sstar} with $\hat{\Sigma}_{i}=\bS_{i}^{*(s+1)}$, denoted as $\bgamma^{(s+1)}$ and $\bbeta^{(s+1)}$, respectively,

		\Until{the objective function in~\eqref{eq:opt_Sstar} converges;}

		\State consider a random series of initializations, repeat Steps 1--5, and choose the results with the minimum objective value.

		\OUTPUT $(\hat{\bgamma},\hat{\bbeta})$
	\end{algorithmic}
\end{algorithm}




\section{Asymptotic Properties}
\label{sec:asmp}

In this section, we study the asymptotic properties of the proposed estimators. For $i=1,\dots,n$, it is assumed that $\Sigma_{i}$ has the eigendecomposition of $\Sigma_{i}=\Pi_{i}\Lambda_{i}\Pi_{i}^\top$, where $\Lambda_{i}=\mathrm{diag}\{\lambda_{i1},\dots,\lambda_{ip}\}$ is a diagonal matrix and $\Pi_{i}=(\bpi_{i1},\dots,\bpi_{ip})$ is an orthonormal rotation matrix; $\{\lambda_{i1},\dots,\lambda_{ip}\}$ are the eigenvalues and the columns of $\Pi_{i}$ are the corresponding eigenvectors. Let $\bZ_{i}=\bY_{i}\Pi_{i}$, where $\bY_{i}=(\by_{i1},\dots,\by_{iT_{i}})^\top\in\mathbb{R}^{T_{i}\times p}$ is the data matrix of subject $i$. Under the normality assumption, the columns of $\bZ_{i}=(z_{itj})_{t,j}$ are uncorrelated, and the rows, $\bz_{it}=(z_{i1},\dots,z_{ip})\in\mathbb{R}^{p}$ for $t=1,\dots,T_{i}$, are normally distributed with mean zero and covariance matrix $\Lambda_{i}$. The following assumptions are imposed.
\begin{description}
	\item[Assumption A1] There exists a constant $C_{1}$ independent of $T_{\max}$ such that $p/T_{\max}\leq C_{1}$, where $T_{\max}=\max_{i}T_{i}$.
	\item[Assumption A2] Let $N=\sum_{i=1}^{n}T_{i}$, $p/N\rightarrow0$ as $n,T_{\min}\rightarrow\infty$, where $T_{\min}=\min_{i}T_{i}$.
	\item[Assumption A3] There exists a constant $C_{2}$ independent of $T_{\min}$ and $T_{\max}$ such that $\sum_{j=1}^{p}\mathbb{E}(z_{i1j}^{8})/p\leq C_{2}$, for $\forall ~i\in\{1,\dots,n\}$.
	\item[Assumption A4] Let $\mathcal{Q}$ denote the set of all the quadruples that are made of four distinct integers between $1$ and $p$, for $\forall~i\in\{1,\dots,n\}$,
		\begin{equation}
			\lim_{T_{i}\rightarrow\infty}\frac{p^{2}}{T_{i}^{2}}\frac{\sum_{(j,k,l,m)\in\mathcal{Q}}\left\{\mathrm{Cov}(z_{i1j}z_{i1k},z_{i1l}z_{i1m})\right\}^{2}}{|\mathcal{Q}|}=0,
		\end{equation}
	where $|\mathcal{Q}|$ is the cardinality of set $\mathcal{Q}$.
	\item[Assumption A5] All the covariance matrices share the same set of eigenvectors, i.e., $\Pi_{i}=\Pi$, for $i=1,\dots,n$. For each $\Sigma_{i}$, there exists (at least) a column, indexed by $j_{i}$, such that $\bgamma=\bpi_{ij_{i}}$ and Model~\eqref{eq:model} is satisfied.
\end{description}
Assumption A1 allows the data dimension, $p$, to be greater than the (maximum) number of observations, $T_{\max}$, and to grow at the same rate as $T_{\max}$ does. This is a common regularity condition for shrinkage estimators~\citep{ledoit2004well}. Assumption A2 guarantees that the average sample covariance matrix $\bar{\bS}=\sum_{i=1}^{n}T_{i}\bS_{i}/N$ utilized in the initial step of Algorithm~\ref{alg:covreg} is positive definite. Together with Assumption A5, the eigenvectors of $\bar{\bS}$ are consistent estimators of $\Pi$~\citep{anderson1963asymptotic}. Assumptions A3 and A4 regulate $\bz_{it}$ on higher-order moments, which is equivalent to imposing restrictions on the higher-order moments of $\by_{it}$. When the data are assumed to be normally distributed, both A3 and A4 are satisfied. Assumption A5 assumes that all the covariance matrices share the same eigenspace, though the ordering of the eigenvectors may differ. When $p/T_{\min}\rightarrow0$, \citet{zhao2019covariate} relaxed this assumption to partial common diagonalization and demonstrated the method robustness through numerical examples. Studying the asymptotic properties under the relaxation is difficult and not available in existing literature, especially when $p> T_{\max}$. 

Taking the eigenvectors of $\bar{\bS}$ as the initial values of $\bgamma$, the following proposition demonstrates the consistency of the proposed estimator.
\begin{proposition}\label{prop:asymp_alg1}
	Under Assumptions A1--A5, the estimator of $\bgamma$ and $\bbeta$ obtained by Algorithm~\ref{alg:covreg} are asymptotically consistent.
\end{proposition}
\noindent To proof Proposition~\ref{prop:asymp_alg1}, we first study the asymptotic properties of $\bS_{i}^{*}$ and show that $\bS_{i}^{*}$ is the optimal linear shrinkage estimator of the covariance matrix under the squared loss. This is accomplished under the assumption that $\bgamma$ is given. As the initialization of $\bgamma$ is already a consistent estimator, the consistency of the solution after iteration follows. For $\bbeta$, it is firstly shown that the association between the shrinkage estimator, $\Sigma_{i}^{*}$, and the covariates is the same as the covariance matrix, $\Sigma_{i}$, does (Lemma~\ref{lemma:beta_equiv}). Thus, it is equivalent to optimize problems~\eqref{eq:opt_Sstar} and \eqref{eq:opt_cov} to solve for $\bbeta$, and the solution is a consistent estimator of $\bbeta$ based on the pseudo-likelihood theory~\citep{gong1981pseudo}. In the iteration step of Algorithm~\ref{alg:covreg}, $\bS_{i}^{*}$ improves the estimation of the covariance matrices with lower squared loss, and in consequence, improves the estimation of $\bgamma$ and $\bbeta$. In Section~\ref{sec:sim}, the improvement is demonstrated through simulation studies.


In Section~\ref{sec:method}, the optimization problem~\eqref{eq:opt_cov} introduces a linear combination of the sample covariance matrix and the identity matrix, $\Sigma_{i}^{*}$, that achieves the minimum expected squared error. From Theorem~\ref{thm:Sigma_solution}, the solution has population-level parameters. Thus, the sample counterpart, $\bS_{i}^{*}$, is introduced. The following Lemma~\ref{lemma:shpar_bound} first shows that asymptotically, the weighting parameters in $\Sigma_{i}^{*}$ are well-behaved. Lemma~\ref{lemma:shpar_consist} demonstrates that the corresponding sample counterpart of the weighting parameters are consistent estimators. Theorem~\ref{thm:Si_consistent} demonstrates that $\bS_{i}^{*}$ performs as well as $\Sigma_{i}^{*}$ does asymptotically.
\begin{lemma}\label{lemma:shpar_bound}
	For given $(\bgamma,\bbeta)$, let $T_{\min}=\min_{i}T_{i}$, as $T_{\min}\rightarrow\infty$, $\mu$, $\phi^{2}$, $\psi^{2}$ and $\delta^{2}$ are bounded.
\end{lemma}
\begin{lemma}\label{lemma:shpar_consist}
	For given $(\bgamma,\bbeta)$, as $T_{\min}\rightarrow\infty$,
	\begin{enumerate}[(i)]
		\item $\mathbb{E}(\hat{\delta}_{i}^{2}-\delta_{i}^{2})^{2}\rightarrow0$, for $i=1,\dots,n$, and thus $\mathbb{E}(\hat{\delta}^{2}-\delta^{2})^{2}\rightarrow 0$;
		\item $\mathbb{E}(\hat{\psi}_{i}^{2}-\psi_{i}^{2})^{2}\rightarrow0$, for $i=1,\dots,n$, and thus $\mathbb{E}(\hat{\psi}^{2}-\psi^{2})^{2}\rightarrow 0$;
		\item $\mathbb{E}(\hat{\phi}_{i}^{2}-\phi_{i}^{2})^{2}\rightarrow0$, for $i=1,\dots,n$, and thus $\mathbb{E}(\hat{\phi}^{2}-\phi^{2})^{2}\rightarrow 0$.
	\end{enumerate}
\end{lemma}
\begin{theorem}\label{thm:Si_consistent}
	For $\forall~i\in\{1,\dots,n\}$, $\bS_{i}^{*}$ is a consistent estimator of $\Sigma_{i}^{*}$, that is, as $T_{\min}=\min_{i}T_{i}\rightarrow\infty$,
	\begin{equation}
		\mathbb{E}\|\bS_{i}^{*}-\Sigma_{i}^{*}\|^{2}\rightarrow 0.
	\end{equation}
	Thus, the asymptotic expected loss of $\bS_{i}^{*}$ and $\Sigma_{i}^{*}$ are identical, that is,
	\begin{equation}
		\mathbb{E}\left\{\bgamma^\top\bS_{i}^{*}\bgamma-\exp(\bx_{i}^\top\bbeta)\right\}^{2}-\mathbb{E}\left\{\bgamma^\top\Sigma_{i}^{*}\bgamma-\exp(\bx_{i}^\top\bbeta)\right\}^{2}\rightarrow 0.
	\end{equation}
\end{theorem}

Next, we show that $\bS_{i}^{*}$ uniformly achieves the minimum quadratic risk asymptotically over all linear combinations of the sample covariance matrix and the identity matrix. For given $(\bgamma,\bbeta)$, let $\Sigma_{i}^{**}$ denote the solution to the following optimization problem,
\begin{eqnarray}\label{eq:opt_sample}
	\underset{\rho_{1},\rho_{2}}{\text{minimize}} && \frac{1}{n}\sum_{i=1}^{n}\left\{\bgamma^\top\Sigma_{i}^{**}\bgamma-\exp(\bx_{i}^\top\bbeta)\right\}^{2}, \nonumber \\
	\text{such that} && \Sigma_{i}^{**}=\rho_{1}\boldsymbol{\mathrm{I}}+\rho_{2}\bS_{i}, \quad \text{for } i=1,\dots,n.
\end{eqnarray}
\begin{theorem}\label{thm:optimal1}
	$\bS_{i}^{*}$ is a consistent estimator of $\Sigma_{i}^{**}$, that is, as $T_{\min}=\min_{i}T_{i}\rightarrow\infty$, for $i=1,\dots,n$,
	\begin{equation}
		\mathbb{E}\|\bS_{i}^{*}-\Sigma_{i}^{**}\|^{2}\rightarrow 0.
	\end{equation}
	Then, $\bS_{i}^{*}$ has the same asymptotic expected loss as $\Sigma_{i}^{**}$ does, that is,
	\begin{equation}
		\mathbb{E}\left\{\bgamma^\top\bS_{i}^{*}\bgamma-\exp(\bx_{i}^\top\bbeta)\right\}^{2}-\mathbb{E}\left\{\bgamma^\top\Sigma_{i}^{**}\bgamma-\exp(\bx_{i}^\top\bbeta)\right\}^{2}\rightarrow 0.
	\end{equation}
\end{theorem}
\begin{theorem}\label{thm:optimal2}
	Assume $(\bgamma,\bbeta)$ is given. With a fixed $n\in\mathbb{N}^{+}$, for any sequence of linear combinations $\{\hat{\Sigma}_{i}\}_{i=1}^{n}$ of the identity matrix and the sample covariance matrix, where the combination coefficients are constant over $i\in\{1,\dots,n\}$, the estimator $\bS_{i}^{*}$ verifies:
	\begin{equation}
  		\lim_{T\rightarrow\infty}\inf_{T_{i}\geq T}\left[\frac{1}{n}\sum_{i=1}^{n}\mathbb{E}\left\{\bgamma^\top\hat{\Sigma}_{i}\bgamma-\exp(\bx_{i}^\top\bbeta)\right\}^{2}-\frac{1}{n}\sum_{i=1}^{n}\mathbb{E}\left\{\bgamma^\top\bS_{i}^{*}\bgamma-\exp(\bx_{i}^\top\bbeta)\right\}^{2}\right]\geq 0.
  	\end{equation}
	In addition, every sequence of $\{\hat{\Sigma}_{i}\}_{i=1}^{n}$ that performs as well as $\{\bS_{i}^{*}\}_{i=1}^{n}$ is identical to $\{\bS_{i}^{*}\}_{i=1}^{n}$ in the limit:
  	\begin{equation}
  		\lim_{T\rightarrow\infty}\left[\frac{1}{n}\sum_{i=1}^{n}\mathbb{E}\left\{\bgamma^\top\hat{\Sigma}_{i}\bgamma-\exp(\bx_{i}^\top\bbeta)\right\}^{2}-\frac{1}{n}\sum_{i=1}^{n}\mathbb{E}\left\{\bgamma^\top\bS_{i}^{*}\bgamma-\exp(\bx_{i}^\top\bbeta)\right\}^{2}\right]=0 
  	\end{equation}
  	\begin{equation}
  	\Leftrightarrow \quad \mathbb{E}\|\hat{\Sigma}_{i}-\bS_{i}^{*}\|^{2} \rightarrow 0, \quad \text{for } i=1,\dots,n.
  	\end{equation}
\end{theorem}

The difference between $\Sigma_{i}^{**}$ and $\Sigma_{i}^{*}$ is that $\Sigma_{i}^{**}$ minimizes the squared loss instead of the expected loss, while asymptotically they are equivalent (Theorems~\ref{thm:Si_consistent} and \ref{thm:optimal1}). Theorem~\ref{thm:optimal2} presents the main result that, with a fixed sample size $n$, the proposed shrinkage estimator $\{\bS_{i}^{*}\}_{i=1}^{n}$ achieves the uniformly minimum (average) quadratic risk asymptotically among all linear combinations of the identity matrix and the sample covariance matrix. Here, ``average'' implies an average over the subjects, and ``asymptotically'' refers to that the number of observations within each subject increases to infinity. Therefore, $\bS_{i}^{*}$ is asymptotically optimal. In addition, it is guaranteed that $\bS_{i}^{*}$ is positive definite (see a discussion in Section~\ref{appendix:sub:Si_willcondition} of the supplementary materials). Thus, there exits unique solution to the optimization problem~\eqref{eq:opt_Sstar}.


Next, we study the asymptotic properties of the model coefficient estimator. Let $\hat{\bbeta}$ denote the solution to the optimization problem~\eqref{eq:opt_Sstar}.
\begin{lemma}\label{lemma:beta_equiv}
	For given $\bgamma$, assume the linear shrinkage estimator, $\Sigma_{i}^{*}$, satisfies
	\begin{equation}
		\mathbb{E}(\bgamma^\top\Sigma_{i}^{*}\bgamma)=\exp(\bx_{i}^\top\bbeta^{*}), \quad \text{for } i=1,\dots,n,
	\end{equation}
	then
	\begin{equation}
		\bbeta^{*}=\bbeta.
	\end{equation}
\end{lemma}
\begin{theorem}\label{thm:beta_consist}
	For given $\bgamma$, assume Assumptions A1--A5 are satisfied, $\hat{\bbeta}$ is a consistent estimator of $\bbeta$ as $n,T_{\min}\rightarrow\infty$, where $T_{\min}=\min_{i}T_{i}$.
\end{theorem}
\noindent Lemma~\ref{lemma:beta_equiv} implies that under the rotation $\bgamma$, the expectation of the shrinkage estimator, $\Sigma_{i}^{*}$, has the same association with the covariates as the true covariance matrix, $\Sigma_{i}$, does. $\bS_{i}^{*}$ is a consistent estimator of $\Sigma_{i}^{*}$ and is positive definite. This substantiates the choice of $\bS_{i}^{*}$ replacing the sample covariance matrix $\bS_{i}$ in the optimization problem. Theorem~\ref{thm:beta_consist} states the consistency of $\hat{\bbeta}$.


\section{Simulation Study}
\label{sec:sim}

\subsection{\texorpdfstring{$\bgamma$}{} is known}
\label{sub:sim_gamma_known}

In this section, we focus on examining the performance of the proposed method in estimating the covariance matrices and model coefficients by assuming the projection $\bgamma$ is known. Three methods are compared. (i) Estimate each individual covariance matrix using the estimator proposed in \citet{ledoit2004well} and replace $\hat{\Sigma}_{i}$ with it in the optimization problem~\eqref{eq:opt_Sstar}. We denote this approach as LW-CAP (\textbf{L}edoit and \textbf{W}olf based \textbf{C}ovariate \textbf{A}ssisted \textbf{P}rincipal regression), where the shrinkage is estimated on each individual covariance matrix. (ii) Estimate the covariance matrices using the proposed shrinkage estimator $\bS_{i}^{*}$ in~\eqref{eq:S_gbshrink}. We denote this approach as CS-CAP (\textbf{C}ovariate dependent \textbf{S}hrinkage CAP), where the shrinkage parameters are assumed to be shared across subjects. (iii) Estimate each individual covariance matrix using the sample covariance matrix and plug into the optimization problem~\eqref{eq:opt_Sstar}. This is the CAP approach proposed in \citet{zhao2019covariate}, which is only applicable when $T_{\min}=\min_{i}T_{i}>p$.

The covariance matrices are generated using the eigendecomposition $\Sigma_{i}=\Pi\Lambda_{i}\Pi^\top$, where $\Pi=(\bpi_{1},\dots,\bpi_{p})$ is an orthonormal matrix in $\mathbb{R}^{p\times p}$ and $\Lambda_{i}=\mathrm{diag}\{\lambda_{i1},\dots,\lambda_{ip}\}$ is a diagonal matrix with the diagonal elements to be the eigenvalues, for $i=1,\dots,n$. In $\Lambda_{i}$, the diagonal elements are exponentially decaying, where eigenvalues of the second and the fourth dimension (D2 and D4) satisfy the log-linear model in~\eqref{eq:model}. We consider a case with a single predictor $X$ (thus $q=2$), which is generated from a Bernoulli distribution with probability $0.5$ to be one. For D2, the coefficient $\beta_{1}=-1$; and for D4, $\beta_{1}=1$. For the rest dimensions, $\lambda_{ij}$, for $i=1,\dots,n$, is generated from a log-normal distribution, where the mean of the corresponding normal distribution decreases from $5$ to $-1$ over $j$. Cases when $p=20,50,100$ are considered.

We first compare the three approaches, LW-CAP, CS-CAP and CAP, under sample sizes $n=50$ and $T_{i}=T=50$ for all $i$ and present the result in Table~\ref{table:sim_gammaknown_comp}. In the estimation, for dimension $j$, $\bgamma$ is set to be $\bpi_{j}$. In Table~\ref{table:sim_gammaknown_comp}, we present the bias and the mean squared error (MSE) in estimating the eigenvalues and the model coefficient in D2 and D4. From the table, for both the eigenvalues and $\beta_{1}$, CS-CAP yields lower estimation bias and MSE than LW-CAP does. When $p<T$, CS-CAP achieves a similar estimation bias as the CAP approach does in estimating the covariance matrices, while the MSE is slightly lower. For the estimation of $\beta_{1}$, CS-CAP yields slightly lower bias. As the dimension $p$ increases, the bias and MSE of eigenvalue estimates from LW-CAP increase; while the bias and MSE of the estimates from CS-CAP are similar at all $p$ settings. This demonstrates the superiority of the proposed estimator in estimating the covariance matrices. Figure~\ref{fig:sim_gammaknown_asmp} presents the estimation bias and MSE of CS-CAP estimator at various levels of $T$ when fixing $n=50$ when $p=20$. From the figure, as the number of observations within each subject increases, the estimates converge to the truth.

\begin{table}
	\caption{\label{table:sim_gammaknown_comp}Bias and mean squared error (MSE) in estimating the eigenvalues of the covariance matrices and the $\beta_{1}$ coefficient with sample sizes $n=50$ and $T_{i}=T=50$, for $i=1,\dots,n$, when $\bgamma$ is known.}
	\begin{center}
		\begin{tabular}{l l l r r r c r r r}
			\hline
			& & & \multicolumn{3}{c}{Eigenvalue} && \multicolumn{3}{c}{$\beta_{1}$} \\
			\cline{4-6}\cline{8-10}
			& & & \multicolumn{1}{c}{LW-CAP} & \multicolumn{1}{c}{CS-CAP} & \multicolumn{1}{c}{CAP} && \multicolumn{1}{c}{LW-CAP} & \multicolumn{1}{c}{CS-CAP} & \multicolumn{1}{c}{CAP} \\
			\hline
			& & Bias & -6.520 & -1.175 & -1.175 && -0.053 & 0.001 & -0.003 \\
			& \multirow{-2}{*}{D2} & MSE & 225.360 & 204.686 & 206.117 && 0.006 & 0.004 & 0.004 \\
			\cline{2-10}
			& & Bias & -7.422 & -1.223 & -1.223 && -0.040 & 0.001 & 0.005 \\
			\multirow{-4}{*}{$p=20$} & \multirow{-2}{*}{D4} & MSE & 277.888 & 249.881 & 251.595 && 0.005 & 0.004 & 0.004 \\
			\hline
			& & Bias & -7.975 & -1.428 & \multicolumn{1}{c}{-} && 0.028 & 0.008 & \multicolumn{1}{c}{-} \\
			& \multirow{-2}{*}{D2} & MSE & 224.326 & 202.141 & \multicolumn{1}{c}{-} && 0.004 & 0.003 & \multicolumn{1}{c}{-} \\
			\cline{2-10}
			& & Bias & -8.641 & -1.242 & \multicolumn{1}{c}{-} && -0.012 & 0.001 & \multicolumn{1}{c}{-} \\
			\multirow{-4}{*}{$p=50$} & \multirow{-2}{*}{D4} & MSE & 295.221 & 248.254 & \multicolumn{1}{c}{-} && 0.004 & 0.004 & \multicolumn{1}{c}{-} \\
			\hline
			& & Bias & -8.923 & -0.973 & \multicolumn{1}{c}{-} && 0.010 & -0.001 & \multicolumn{1}{c}{-} \\
			& \multirow{-2}{*}{D2} & MSE & 260.268 & 203.151 & \multicolumn{1}{c}{-} && 0.004 & 0.003 & \multicolumn{1}{c}{-} \\
			\cline{2-10}
			& & Bias & -10.487 & -1.705 & \multicolumn{1}{c}{-} && -0.011 & -0.007 & \multicolumn{1}{c}{-} \\
			\multirow{-4}{*}{$p=100$} & \multirow{-2}{*}{D4} & MSE & 331.864 & 245.754 & \multicolumn{1}{c}{-} && 0.003 & 0.003 & \multicolumn{1}{c}{-} \\
			\hline
		\end{tabular}
	\end{center}
\end{table}
\begin{figure}
	\begin{center}
		\subfloat[Bias of $\hat{\lambda}_{ij}$]{\includegraphics[width=0.25\textwidth]{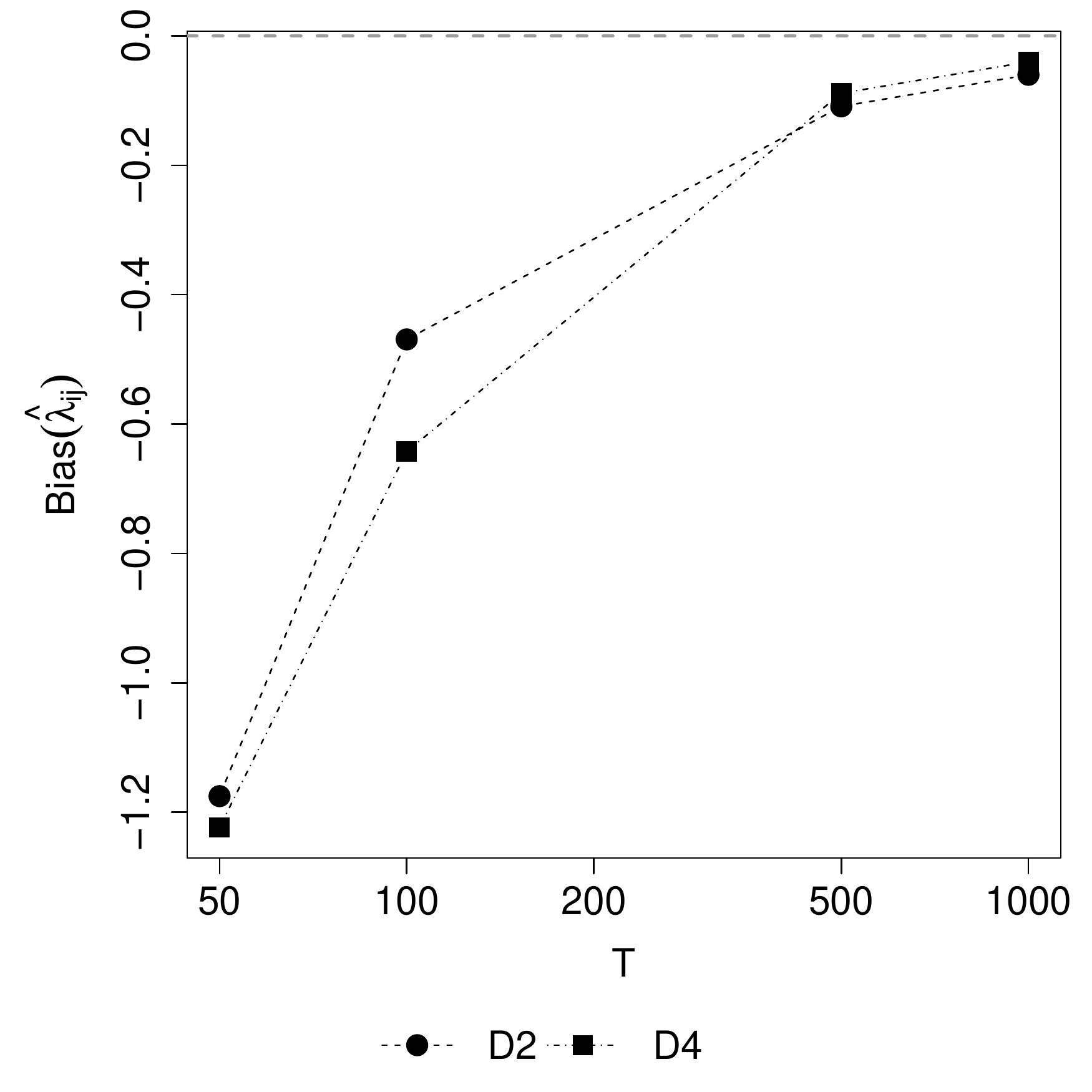}}
		\subfloat[MSE of $\hat{\lambda}_{ij}$]{\includegraphics[width=0.25\textwidth]{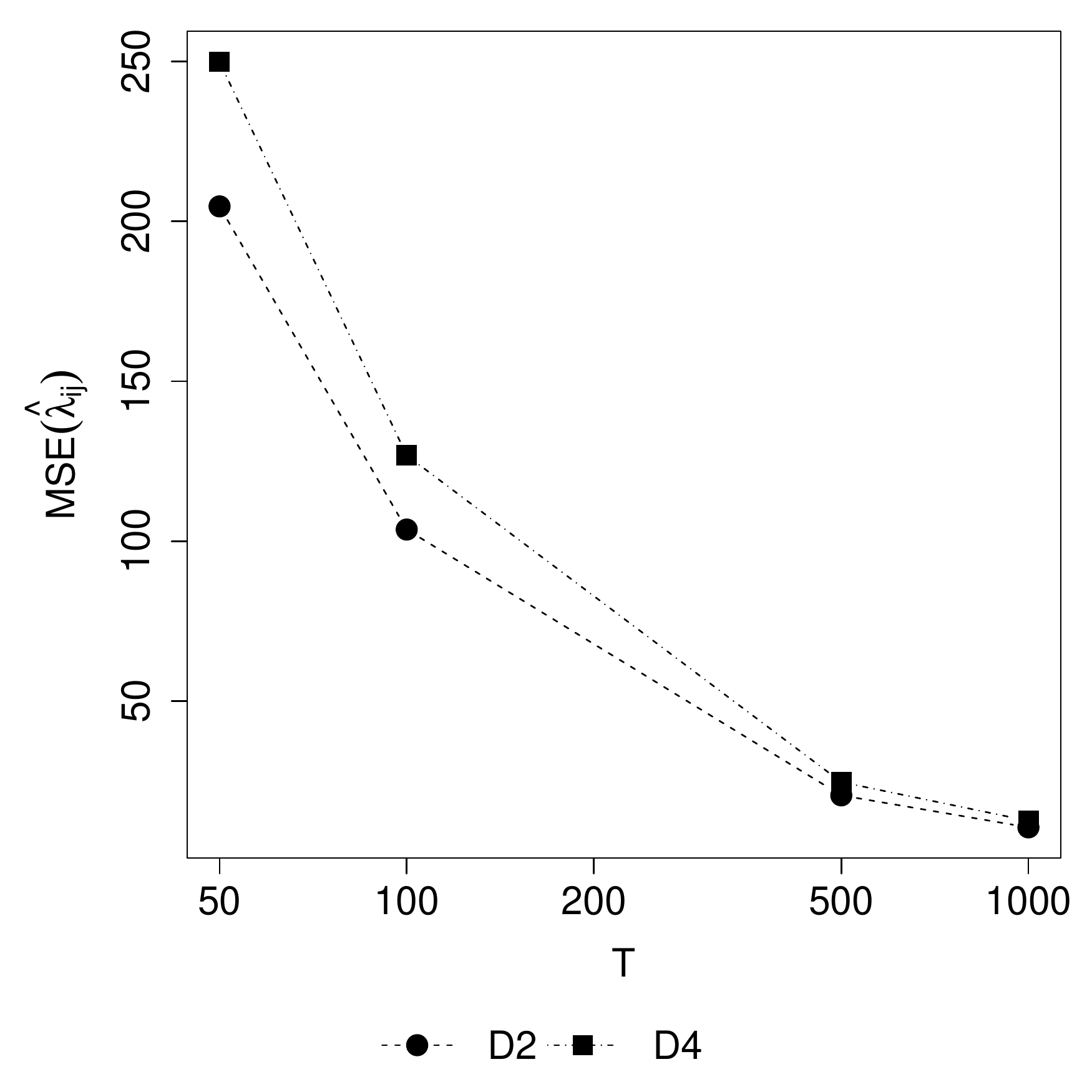}}
		\subfloat[Bias of $\hat{\beta}_{1}$]{\includegraphics[width=0.25\textwidth]{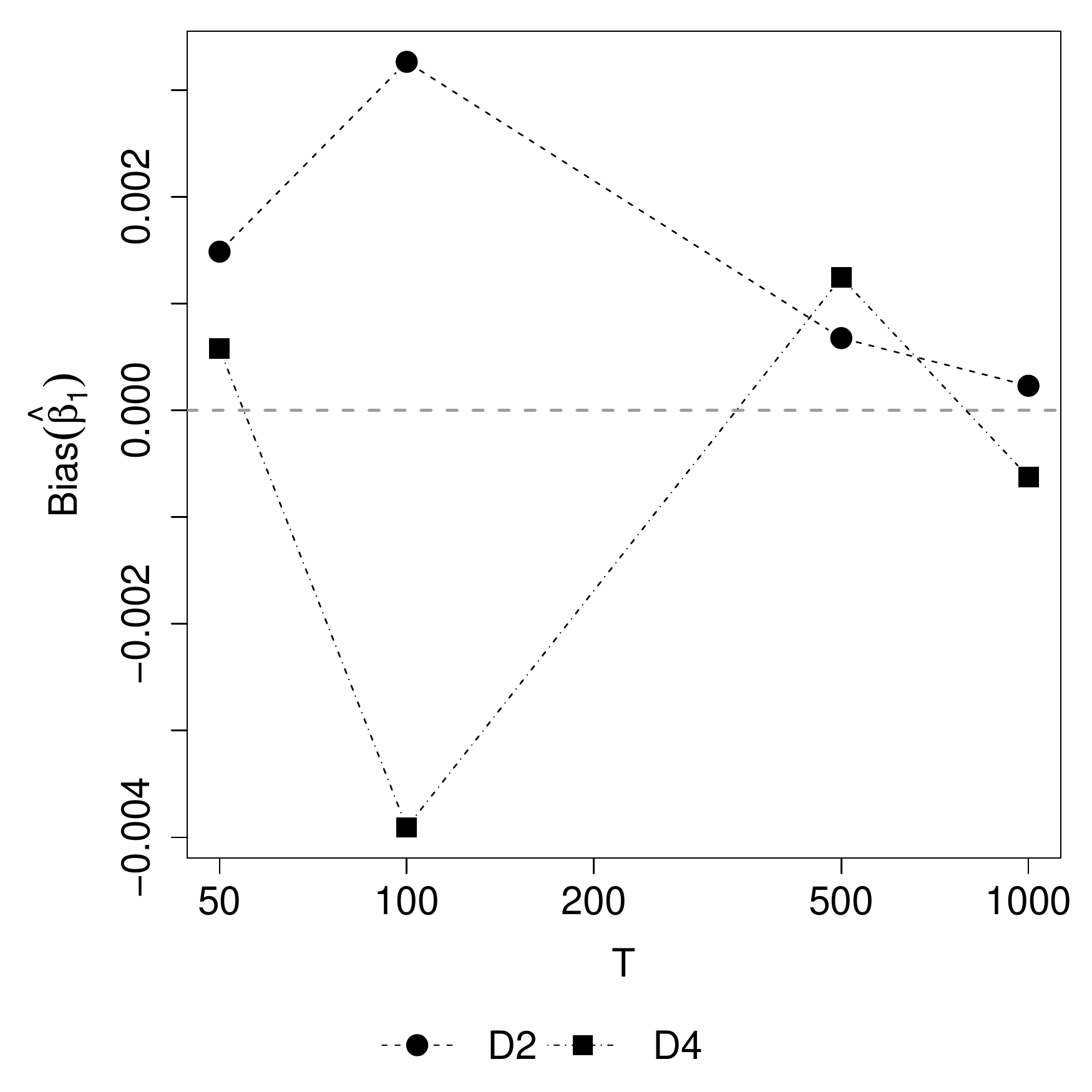}}
		\subfloat[MSE of $\hat{\beta}_{1}$]{\includegraphics[width=0.25\textwidth]{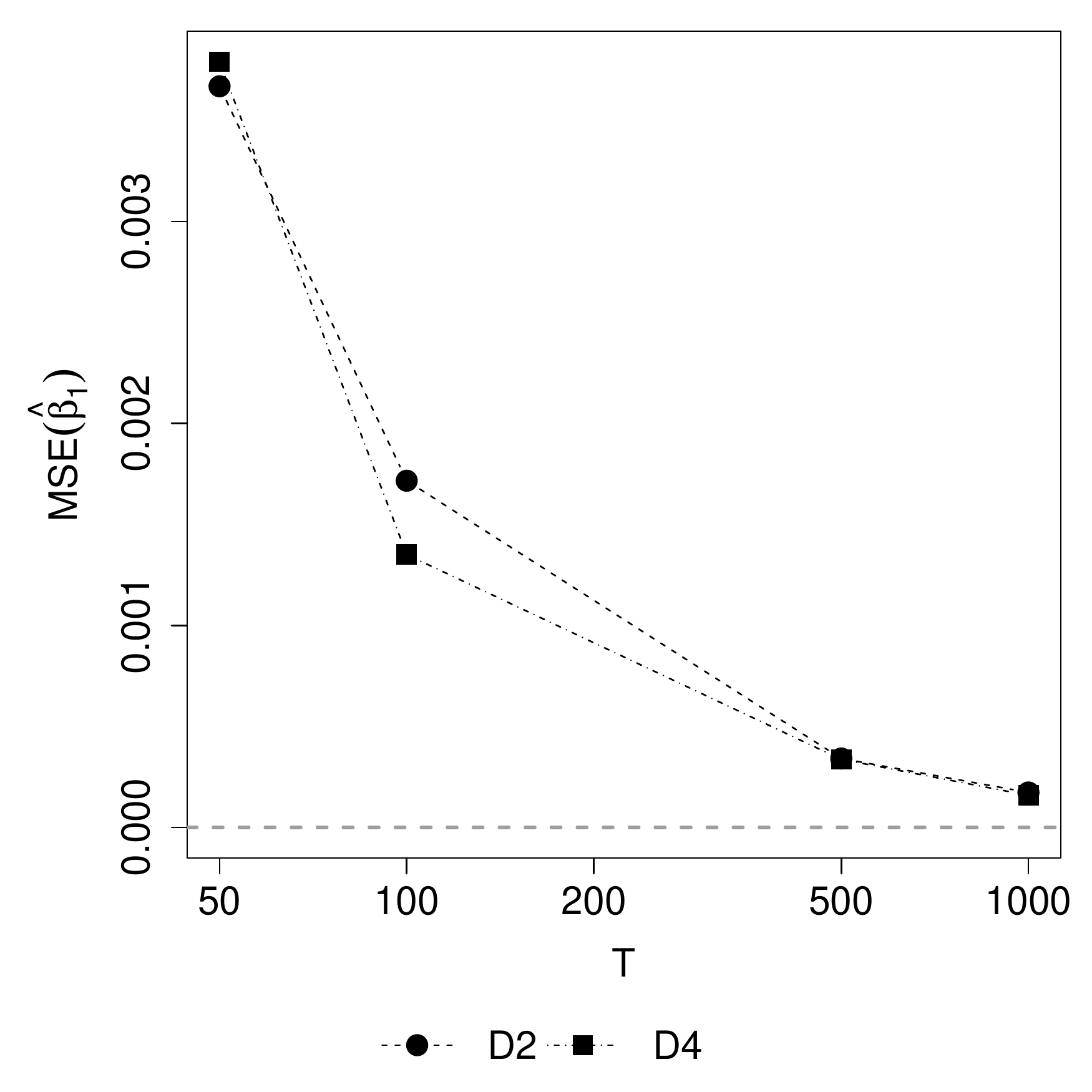}}
	\end{center}
	\caption{\label{fig:sim_gammaknown_asmp}Bias and mean squared error (MSE) in estimating the eigenvalues of the covariance matrices and the $\beta_{1}$ coefficient using CS-CAP with the number of subjects $n=50$ at various numbers of observations from each subject with $p=20$ when $\bgamma$ is known.}
\end{figure}

\subsection{\texorpdfstring{$\bgamma$}{} is unknown}
\label{sub:sim_gamma_unknown}

In this section, we evaluate the performance of the CS-CAP approach when $\bgamma$ is unknown and estimated by solving optimization problem~\eqref{eq:opt_Sstar} using Algorithm~\ref{alg:covreg}. The data are generated following the same procedure as in Section~\ref{sub:sim_gamma_known}. To evaluate the performance in estimating the projection $\bgamma$, we consider a similarity metric measured by $|\langle \hat{\bgamma},\bgamma\rangle|$, where $\langle\cdot,\cdot\rangle$ is the inner product of two vectors and $\hat{\bgamma}$ denotes the estimate of $\bgamma$. When this metric is one, the two vectors are identical (up to sign flipping); and when this metric is zero, the two vectors are orthogonal. Case where $p=100$ is studied. The performance of the CS-CAP approach is firstly compared to the LW-CAP approach with sample sizes $n=100$ and $T_{i}=T=100$. The results are presented in Table~\ref{table:sim_gammaunknown_comp}. From the table, the CS-CAP approach improves the performance with much lower MSE in estimating the eigenvalues, and lower MSE and higher coverage probability (CP) in estimating the $\beta$ coefficient. After iterations, the CS-CAP approach yields an estimate of the projection with much higher similarity to the truth. To further examine the performance of the CS-CAP approach under finite sample size, combinations of sample sizes $n=50,100,500,1000$ and $T_{i}=T=50,100,500,1000$ are considered. Figure~\ref{fig:sim_gammaunknown_asmp} presents the performance in estimating the second dimension (D2), including the bias, the MSE and the CP of $\hat{\beta}_{1}$, the MSE of $\hat{\lambda}_{ij}$, and the similarity of $\hat{\bgamma}$ to the eigenvector of D2 (Section~\ref{appendix:sub:sim_gamma_unknown} of the supplementary materials presents the results of the fourth dimension, D4). From the figure, as $n,T\rightarrow\infty$, all estimates converge to the truth.

\begin{table}
	\caption{\label{table:sim_gammaunknown_comp}Bias, mean squared error (MSE), and coverage probability (CP) from 500 bootstrap samples in estimating the $\beta_{1}$ coefficient, the similarity of $\hat{\bgamma}$ to $\bpi_{j}$ and the standard error (SE), and the MSE in estimating the eigenvalues $\hat{\lambda}_{ij}$, for $j=2,4$. Data dimension $p=100$, sample size $n=100$ and $T_{i}=T=100$.}
	\begin{center}
		\begin{tabular}{l l r r r c r c r}
			\hline
			& & \multicolumn{3}{c}{$\hat{\beta}_{1}$} && \multicolumn{1}{c}{$\hat{\bgamma}$} && \multicolumn{1}{c}{$\hat{\lambda}_{ij}$} \\
			\cline{3-5}\cline{7-7}\cline{9-9}
			& \multicolumn{1}{c}{\multirow{-2}{*}{Method}} & \multicolumn{1}{c}{Bias} & \multicolumn{1}{c}{MSE} & \multicolumn{1}{c}{CP} && \multicolumn{1}{c}{$|\langle\hat{\bgamma},\bpi_{2}\rangle|$ (SE)} && \multicolumn{1}{c}{MSE} \\
			\hline
			& LW-CAP & -0.027 & 0.002 & 0.782 && 0.653 (0.033) && 1812.091 \\
			\multirow{-2}{*}{D2} & CS-CAP & -0.023 & 0.001 & 0.855 && 0.931 (0.012) && 173.225 \\
			\hline
			& LW-CAP & 0.018 & 0.002 & 0.770 && 0.666 (0.027) && 2186.265 \\
			\multirow{-2}{*}{D4} & CS-CAP & 0.019 & 0.001 & 0.845 && 0.926 (0.011) && 231.856 \\
			\hline
		\end{tabular}
	\end{center}
\end{table}
\begin{figure}
	\begin{center}
		\subfloat[Bias of $\hat{\beta}_{1}$]{\includegraphics[width=0.3\textwidth]{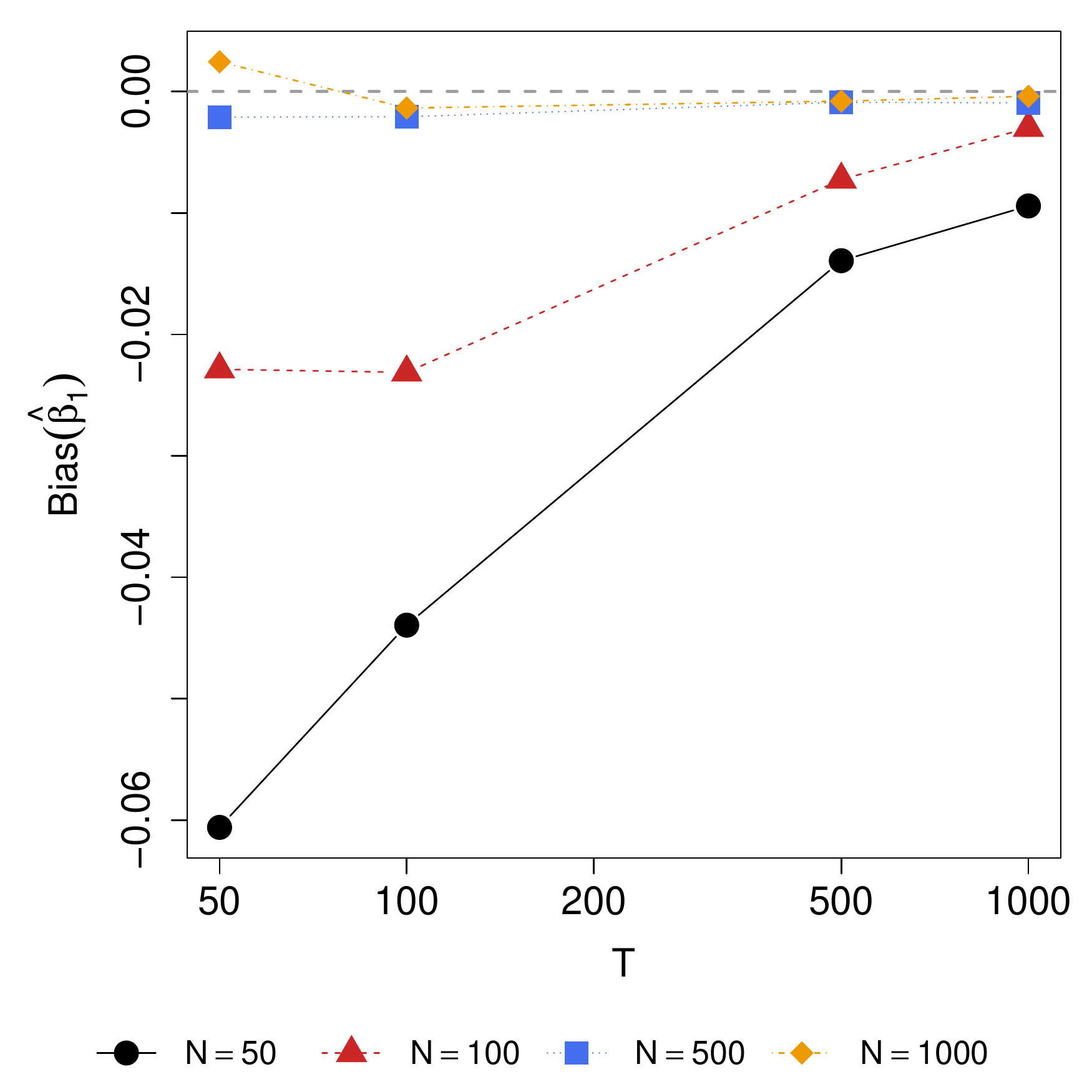}}
		\enskip{}
		\subfloat[MSE of $\hat{\beta}_{1}$]{\includegraphics[width=0.3\textwidth]{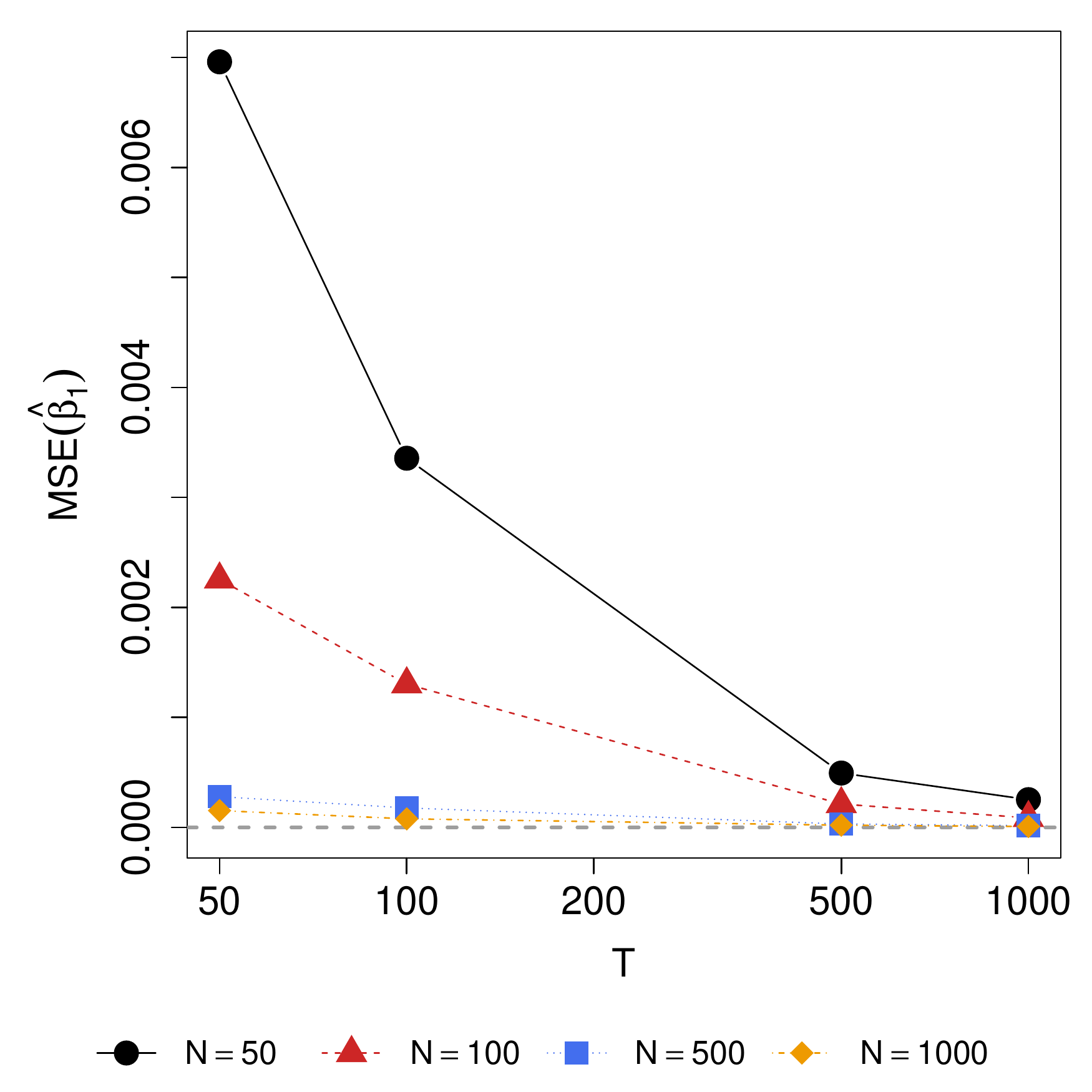}}
		\enskip{}
		\subfloat[Coverage probability of $\hat{\beta}_{1}$]{\includegraphics[width=0.3\textwidth]{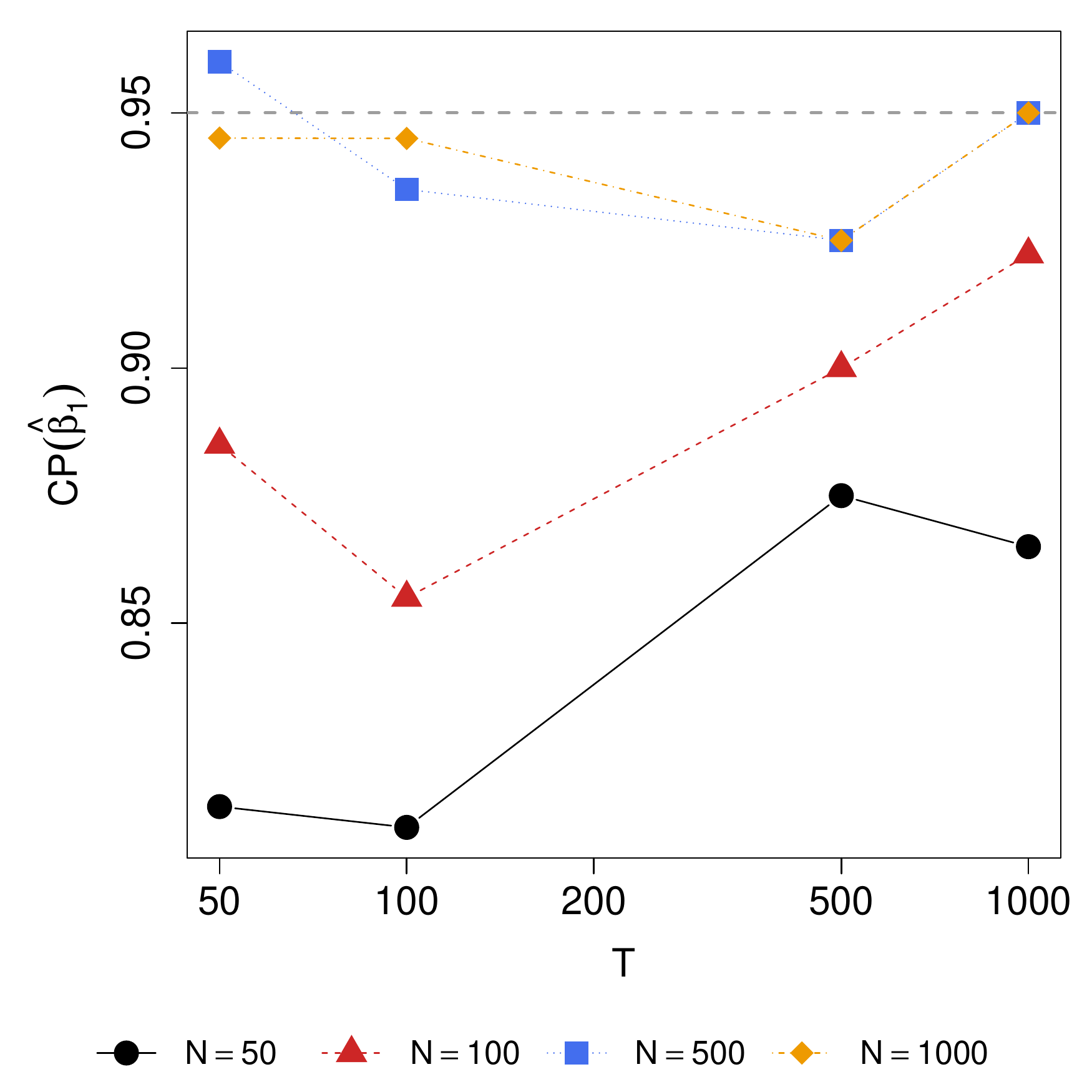}}

		\subfloat[MSE of $\hat{\lambda}_{ij}$]{\includegraphics[width=0.3\textwidth]{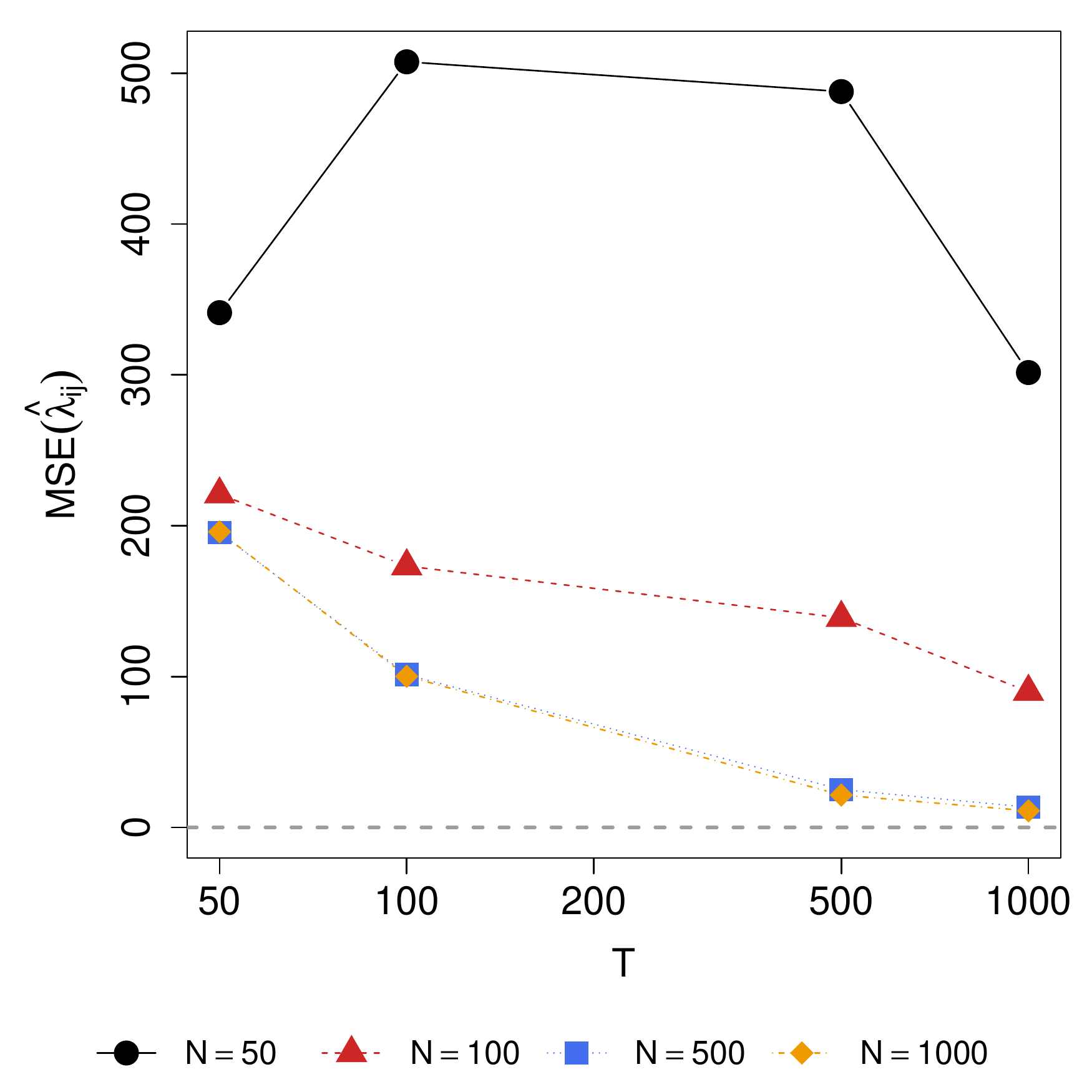}}
		\enskip{}
		\subfloat[Similarity metric $|\langle\hat{\bgamma},\bpi_{2}\rangle|$]{\includegraphics[width=0.3\textwidth]{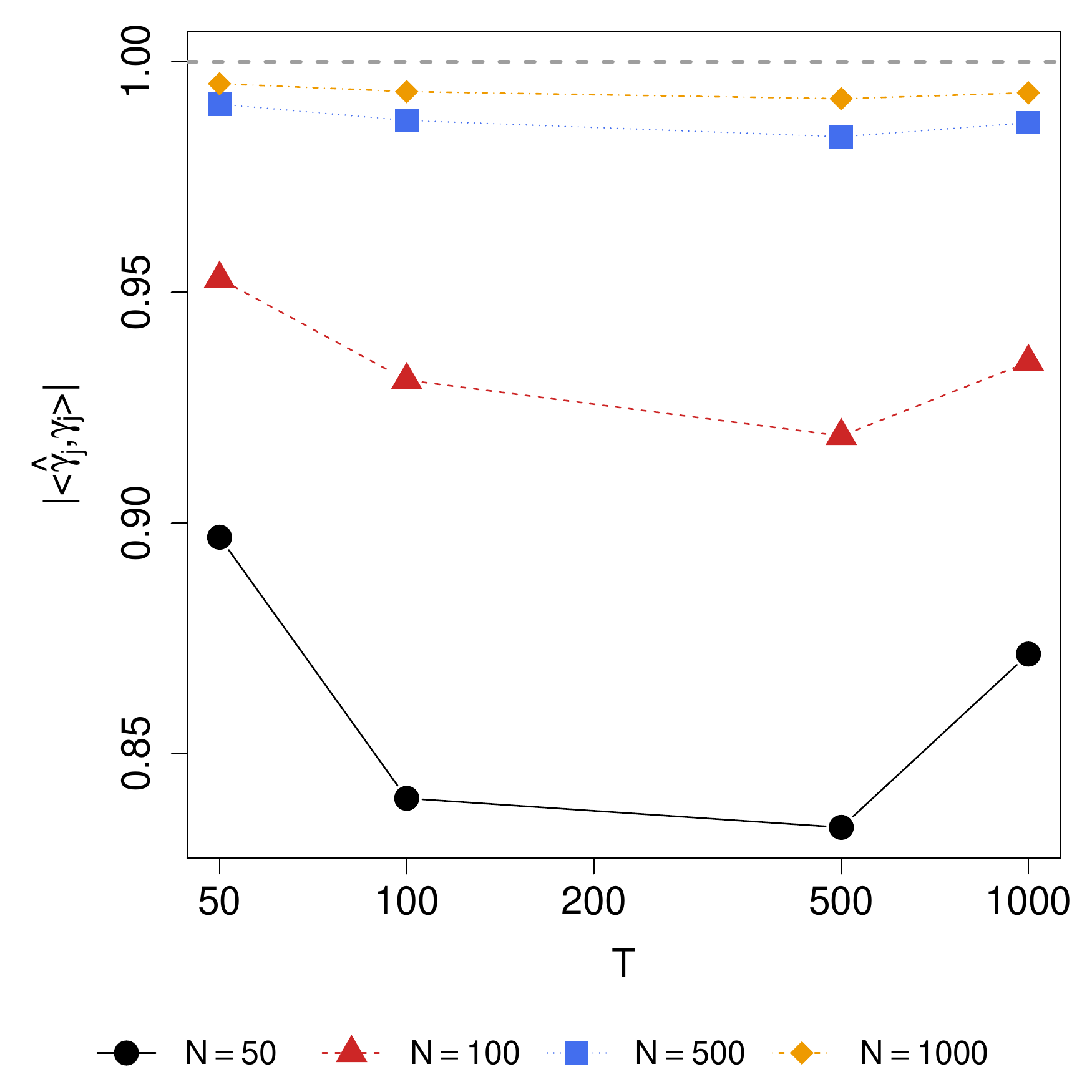}}
	\end{center}
	\caption{\label{fig:sim_gammaunknown_asmp}Estimation performance of CS-CAP in estimating the second dimension (D2) when $\bgamma$ is unknown. For $\hat{\beta}_{1}$, (a) bias, (b) mean squared error (MSE) and (c) coverage probability (CP) are presented, where CP is obtained from 500 bootstrap samples. For the eigenvalues $\hat{\lambda}_{ij}$, (d) MSE is presented. For $\hat{\bgamma}$, (e) similarity to $\bpi_{2}$ is presented. Data dimension $p=100$. Sample sizes vary from $n=50,100,500,100$ and $T_{i}=T=50,100,500,1000$.}
\end{figure}


\section{The Alzheimer's Disease Neuroimaging Initiative Study}
\label{sec:real}

Data used in this study are obtained from the Alzheimer's Disease Neuroimaging Initiative (ADNI) database (\url{adni.loni.usc.edu}). The ADNI was launched in 2003 as a public-private partnership, led by Principal Investigator Michael W. Weiner, MD. The primary goal of ADNI has been to test whether serial magnetic resonance imaging (MRI), positron emission tomography (PET), other biological markers, and clinical and neuropsychological assessment can be combined to measure the progression of mild cognitive impairment (MCI) and early Alzheimer's disease (AD). 

We apply the proposed approach to ADNI resting-state functional magnetic resonance imaging (fMRI) data acquired at the baseline screening. AD is an irreversible neurodegenerative disease that destroys memory and related brain functions causing problems in cognition and behavior. 
Apolipoprotein E $\varepsilon$4 (APOE-$\varepsilon$4) has been consistently identified as a strong genetic risk factor for AD. With an increasing number of APOE-$\varepsilon$4 alleles, the lifetime risk of developing AD increases, and the age of onset decreases~\citep{corder1993gene}. Thus, APOE-$\varepsilon$4 is generally treated as a potential therapeutic target~\citep{safieh2019apoe4}. In AD studies, resting-state fMRI is another emerging biomarker for diagnosis~\citep{koch2012diagnostic}. It is important to articulate the genetic impact on brain functional architecture. In this study, $n=194$ subjects diagnosed with either MCI or AD are analyzed. Resting-state fMRI data collected at the initial screening are preprocessed. Time courses are extracted from $p=75$ brain regions, including 60 cortical and 15 subcortical regions grouped into 10 functional modules, using the Harvard-Oxford Atlas in FSL~\citep{smith2004advances}. For each time course, a subsample is taken with an effective sample size of $T=67$ to remove the temporal dependence. In the regression model, APOE-$\varepsilon$4, sex and age are entered as the covariates.

The CS-CAP approach is applied to identify brain subnetworks within which the functional connectivity demonstrates a significant association with APOE-$\varepsilon$4. Using the deviation from diagonality criterion, CS-CAP identifies three components. The model coefficients and 95\% bootstrap confidence interval from 500 bootstrap samples are presented in Table~\ref{table:real_beta}. From the table, C3 is significantly associated with APOE-$\varepsilon$4 and age; C1 and C2 are significantly associated with sex and age. To better interpret C3, a fused lasso regression~\citep{tibshirani2005sparsity} is employed to sparsify the loading profile, similarly as in the sparse principal component analysis proposed in \citet{zou2006sparse}. The fused lasso regularization is defined based on the modular information to impose local smoothness and consistency~\citep{grosenick2013interpretable,zhao2020sparse}. Figure~\ref{subfig:real_C3_loading} presents the sparse loading profile colored by the corresponding functional module, and Figure~\ref{subfig:real_C3_riverplot} is the river plot illustrating the loading configuration. In C3, all regions with negative loadings are subcortical regions. Contributions to positive loadings are from regions in the default mode network (DMN), the ventral- and dorsal-attention networks, and the somato-motor network. Figure~\ref{subfig:real_C3_brainmap} presents these regions on a brain map. C3 is negatively associated with APOE-$\varepsilon$4 indicating that functional connectivity between regions in the same sign among APOE-$\varepsilon$4 carriers is lower, while connectivity between regions in the opposite signs among APOE-$\varepsilon$4 carriers is higher. The findings are in line with existing knowledge about AD. Compared to APOE-$\varepsilon$4 non-carriers, more functional connectivity between the left hippocampus and the insular/prefrontal cortex while more functional disconnection of the hippocampus has been observed in APOE-$\varepsilon$4 carriers~\citep{de2017apoe}. Alterations in DMN connectivity in cognitively normal APOE-$\varepsilon$4 carriers have been reported across all age groups~\citep{badhwar2017resting}. Increased connectivity in the limbic system, including the hippocampus, the amygdala and the thalamus, has been detected in individuals with memory impairment~\citep{gour2011basal,gour2014functional}, though the effect of APOE-$\varepsilon$4 carriage lacks consensus~\citep{badhwar2017resting}. It was shown that the limbic hyperconnectivity is positively associated with the memory performance, suggesting the preservation of brain function due to increased connectivity in the medial temporal lobe pathology~\citep{gour2014functional}.

\begin{table}
	\caption{\label{table:real_beta}Model coefficient estimate and 95\% bootstrap confidence interval using the PS-CAP approach. The intervals are obtained over 500 bootstrap samples.}
	\begin{center}
		\begin{tabular}{l r r r}
			\hline
			& \multicolumn{1}{c}{APOE-$\varepsilon$4} & \multicolumn{1}{c}{Sex} & \multicolumn{1}{c}{Age} \\
			\hline
			C1 & $0.012$ $(-0.031, ~~0.263)$ & $-0.431$ $(-0.636, -0.230)$ & $-0.227$ $(-0.319, -0.129)$ \\
			C2 & $0.049$ $(-0.191, ~~0.309)$ & $-0.544$ $(-0.867, -0.186)$ & $-0.232$ $(-0.383, -0.066)$ \\
			C3 & $-0.156$ $(-0.270, -0.045)$ & $-0.061$ $(-0.201, ~~0.075)$ & $-0.241$ $(-0.328, -0.172)$ \\
			\hline
		\end{tabular}
	\end{center}
\end{table}
\begin{figure}
	\begin{center}
		\subfloat[\label{subfig:real_C3_loading}Sparse loading profile of C3.]{\includegraphics[width=\textwidth]{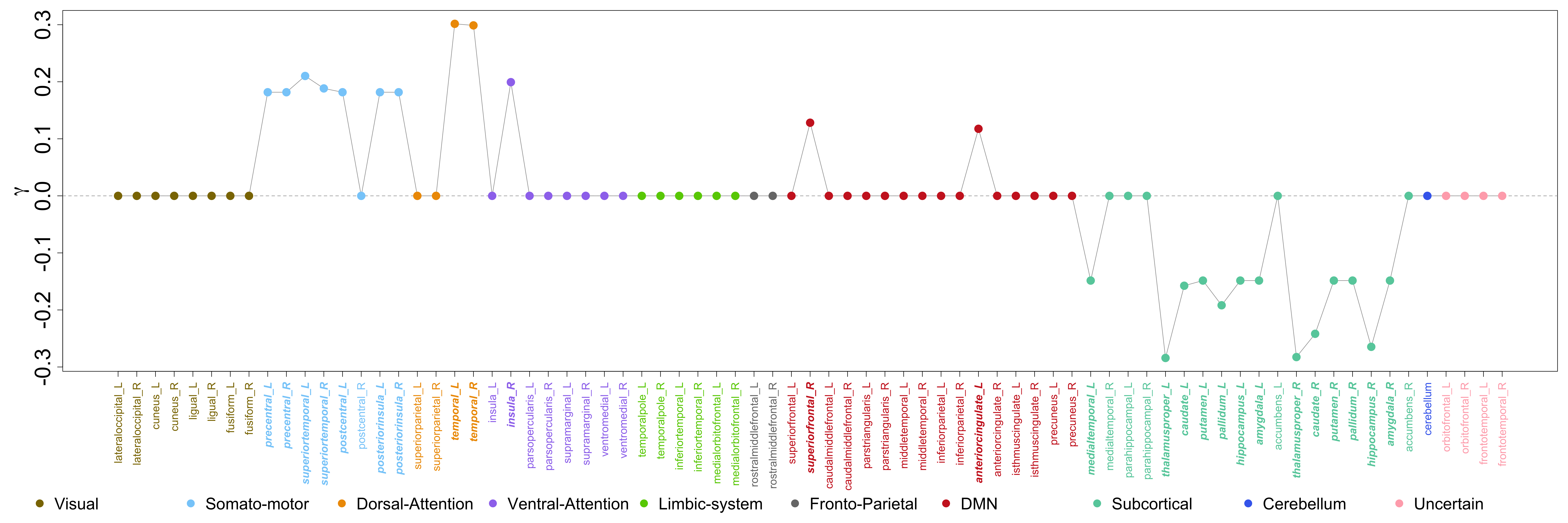}}

		\subfloat[\label{subfig:real_C3_riverplot}River plot of C3 loading.]{{\includegraphics[width=0.45\textwidth,trim={3cm 25cm 5cm 25cm},clip]{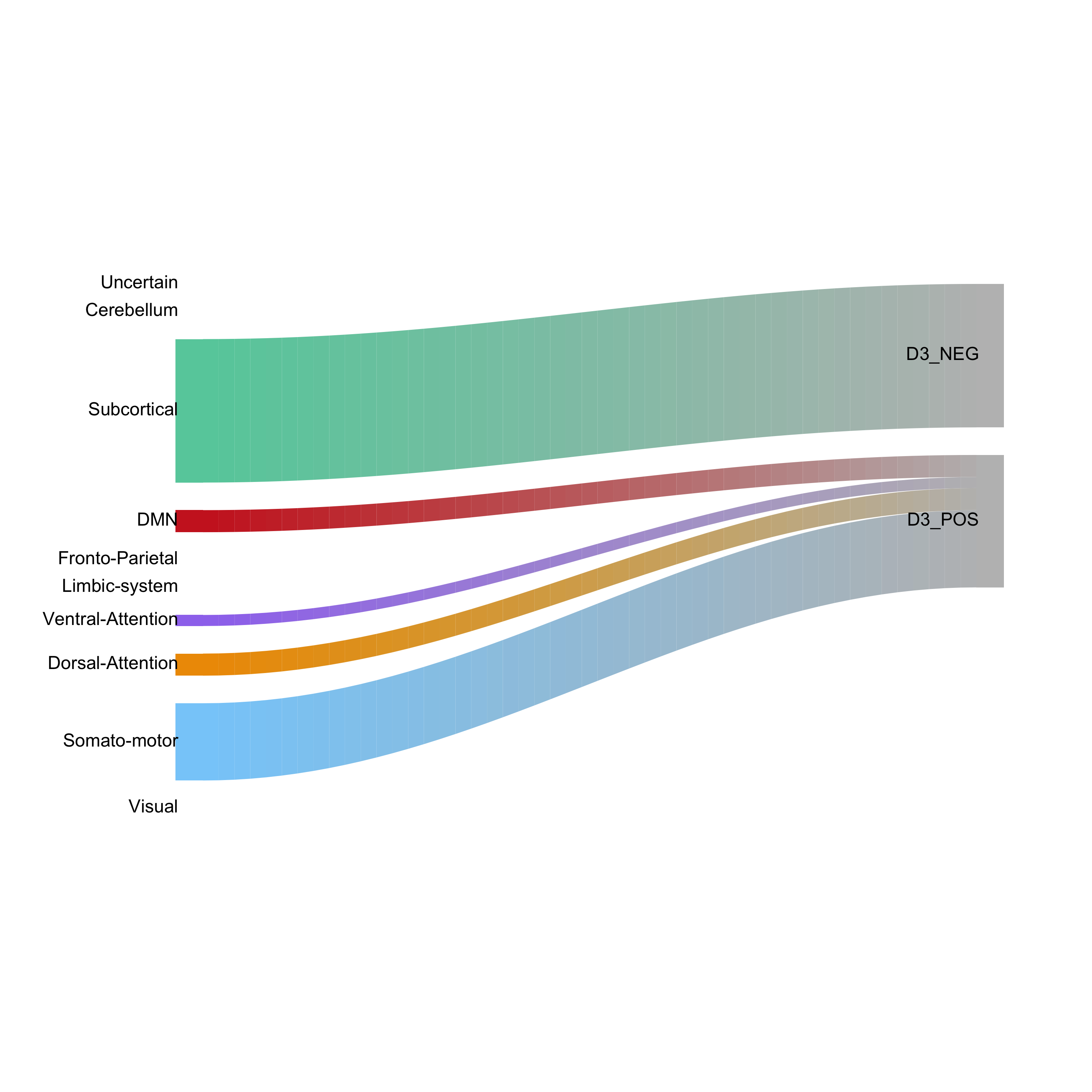}}}
		\enskip{}
		\subfloat[\label{subfig:real_C3_brainmap}Brain map of C3.]{\includegraphics[width=0.45\textwidth]{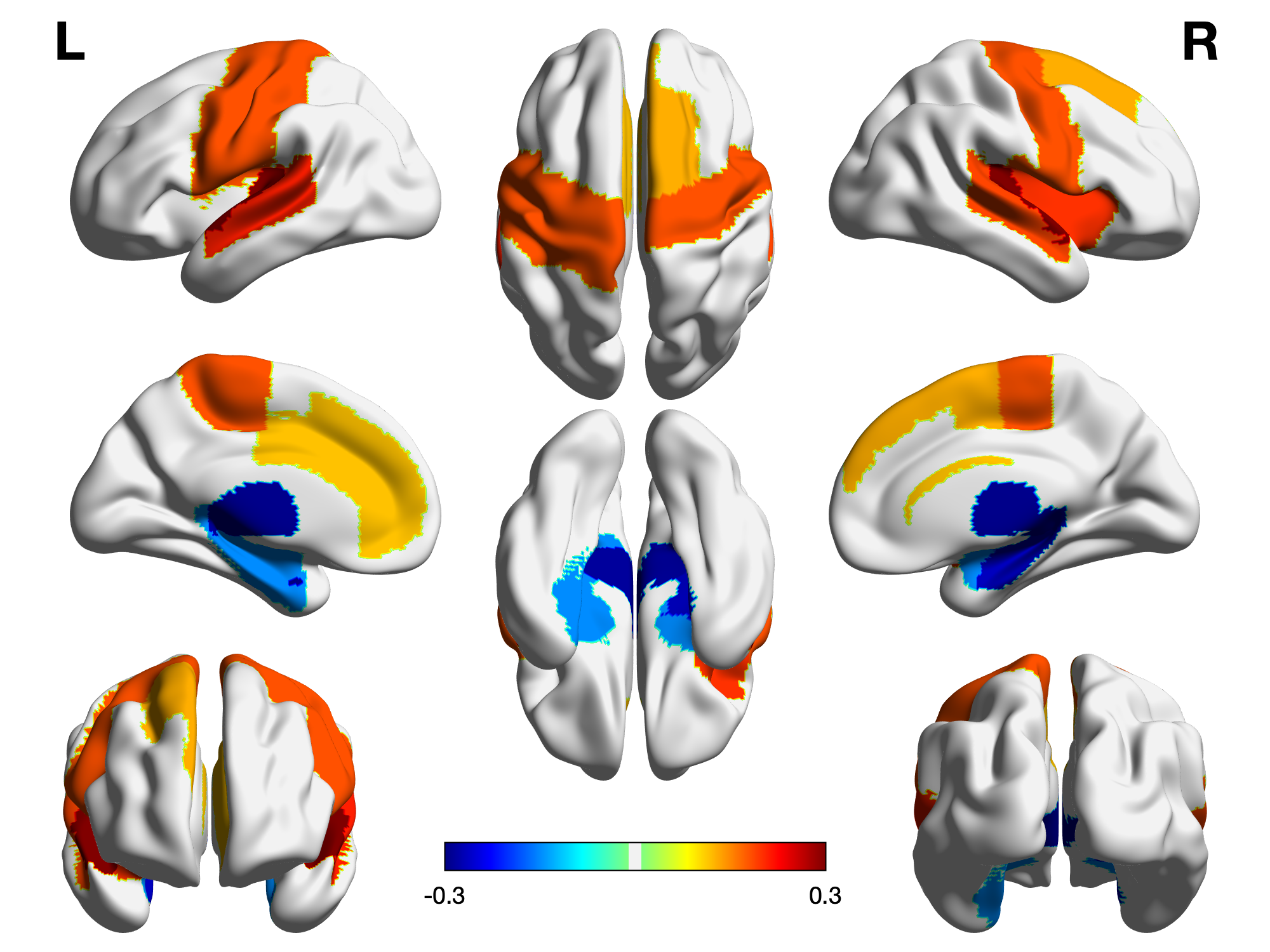}}
	\end{center}
	\caption{\label{fig:real_C3}(a)The sparsified loading profile, (b) the module river plot, and (c) regions with nonzero loadings in a brain map of C3. In (a) and (b), the figure and the legend are colored by brain functional modules. In (c), the brain maps are colored by the loading weights.}
\end{figure}

\section{Discussion}
\label{sec:discussion}

In this study, we introduce an approach to perform linear regression with multiple high dimensional covariance matrices as the outcome. A linear shrinkage estimator of the covariance matrix is firstly introduced, where the shrinkage coefficients are shared parameters across subjects. It is showed that the proposed estimator is optimal achieving the uniformly minimum quadratic loss asymptotically among all linear combinations of the identity matrix and the sample covariance matrix. Utilizing the well-conditioned estimator of the covariance matrices, a pseudo-likelihood based approach is considered to estimate the linear projection parameter and the model coefficient. Through simulation studies, the proposed approach demonstrates superior performance in estimating the covariance matrices and the model coefficients with lower estimation bias and variation over the existing methods. Applying to a resting-state fMRI data set acquired from ADNI, the findings are consistent with existing knowledge about AD. 

The proposed framework extends the proposal in \citet{zhao2019covariate} to high dimensional scenario. When $p$ is small, the proposed shrinkage estimator demonstrates lower squared loss than the sample covariance matrix as suggested in both theoretical results and simulation studies. Different from the linear shrinkage estimator introduced in \citet{ledoit2004well}, which was proposed for a single covariance matrix estimation, the shrinkage coefficients considered in this study are population level parameters shared across subjects. This is superior than the individual shrinkage as the proposed one leverages the accuracy of the sample covariance matrix and the variability in the eigenvalues across subjects.

In this study, the asymptotic properties are studied under the assumption that the covariance matrices have the same eigendecomposition. We leave the study of the consistency relaxing this assumption to future research.
The proposed shrinkage estimator is optimal with respect to a squared risk. However, this may overshrink the small eigenvalues~\citep{daniels2001shrinkage}. Other types of loss function, such as the Stein's loss, will be considered in the future.



\section*{Acknowledgments}


Data collection and sharing for this project was funded by the Alzheimer's Disease Neuroimaging Initiative (ADNI) (National Institutes of Health Grant U01 AG024904) and DOD ADNI (Department of Defense award number W81XWH-12-2-0012). ADNI is funded by the National Institute on Aging, the National Institute of Biomedical Imaging and Bioengineering, and through generous contributions from the following: AbbVie, Alzheimer's Association; Alzheimer's Drug Discovery Foundation; Araclon Biotech; BioClinica, Inc.; Biogen; Bristol-Myers Squibb Company; CereSpir, Inc.; Cogstate; Eisai Inc.; Elan Pharmaceuticals, Inc.; Eli Lilly and Company; EuroImmun; F. Hoffmann-La Roche Ltd and its affiliated company Genentech, Inc.; Fujirebio; GE Healthcare; IXICO Ltd.; Janssen Alzheimer Immunotherapy Research \& Development, LLC.; Johnson \& Johnson Pharmaceutical Research \& Development LLC.; Lumosity; Lundbeck; Merck \& Co., Inc.; Meso Scale Diagnostics, LLC.; NeuroRx Research; Neurotrack Technologies; Novartis Pharmaceuticals Corporation; Pfizer Inc.; Piramal Imaging; Servier; Takeda Pharmaceutical Company; and Transition Therapeutics. The Canadian Institutes of Health Research is providing funds to support ADNI clinical sites in Canada. Private sector contributions are facilitated by the Foundation for the National Institutes of Health (\url{www.fnih.org}). The grantee organization is the Northern California Institute for Research and Education, and the study is coordinated by the Alzheimer’s Therapeutic Research Institute at the University of Southern California. ADNI data are disseminated by the Laboratory for Neuro Imaging at the University of Southern California.

\clearpage

\appendix
\counterwithin{figure}{section}
\counterwithin{table}{section}
\counterwithin{equation}{section}
\counterwithin{lemma}{section}
\counterwithin{theorem}{section}

\bigskip\bigskip
\begin{center}
  \Large{\textbf{Supplementary Materials}}
\end{center}

This supplementary material collects the technical proof of the theorems in the main text and additional simulation results.
\section{Theory and Proof}
\label{appendix:sec:proof}

\subsection{Proof of Theorem~\ref{thm:Sigma_solution} and Lemma~\ref{lemma:par_relation}}
\label{appendix:sub:proof_SigmaSolution}

\begin{proof}
  Given $(\bgamma,\bbeta)$, $\mathbb{E}(\bgamma^\top\bS_{i}\bgamma)=\bgamma^\top\Sigma_{i}\bgamma=\exp(\bx_{i}^\top\bbeta)$. For the objective function in \eqref{eq:opt_cov}, under the constraint that $\Sigma_{i}^{*}=\rho\mu\boldsymbol{\mathrm{I}}+(1-\rho)\bS_{i}$, we have
  \begin{eqnarray*}
    f(\mu,\rho) &=& \frac{1}{n}\sum_{i=1}^{n}\mathbb{E}\left\{\bgamma^\top\Sigma_{i}^{*}\bgamma-\exp(\bx_{i}^\top\bbeta)\right\}^{2} \\
    &=& \frac{1}{n}\sum_{i=1}^{n}\left[\rho^{2}\left\{\mu(\bgamma^\top\bgamma)-\exp(\bx_{i}^\top\bbeta)\right\}^{2}+(1-\rho)^{2}\mathbb{E}\left\{\bgamma^\top\bS_{i}\bgamma-\exp(\bx_{i}^\top\bbeta)\right\}^{2}\right].
  \end{eqnarray*}
  In order to minimize the objective function, as the objective function is convex, derivatives are firstly taken over $\mu$ and $\rho$.

  For $\mu$,
  \[
    \frac{\partial f}{\partial\mu}=\rho^{2}\frac{1}{n}\sum_{i=1}^{n}2\left\{\mu(\bgamma^\top\bgamma)-\exp(\bx_{i}^\top\bbeta)\right\}(\bgamma^\top\bgamma)=0,
  \]
  \[
    \Rightarrow \quad \mu=\frac{1}{n(\bgamma^\top\bgamma)}\sum_{i=1}^{n}\exp(\bx_{i}^\top\bbeta).
  \]

  For $\rho$, let $\phi_{i}^{2}=\{\mu(\bgamma^\top\bgamma)-\exp(\bx_{i}^\top\bbeta)\}^{2}$ and $\psi_{i}^{2}=\mathbb{E}\{\bgamma^\top\bS_{i}\bgamma-\exp(\bx_{i}^\top\bbeta)\}^{2}$,
  \[
    \frac{\partial f}{\partial\rho}=2\rho\left(\frac{1}{n}\sum_{i=1}^{n}\phi_{i}^{2}\right)-2(1-\rho)\left(\frac{1}{n}\sum_{i=1}^{n}\psi_{i}^{2}\right)=0,
  \]
  \[
    \Rightarrow \quad \rho=\frac{\sum_{i=1}^{n}\psi_{i}^{2}}{\sum_{k=1}^{n}\phi_{i}^{2}+\sum_{i=1}^{n}\psi_{i}^{2}}.
  \]
  Let $\delta_{i}^{2}=\mathbb{E}\{\bgamma^\top\bS_{i}\bgamma-\mu(\bgamma^\top\bgamma)\}^{2}$, then $\delta_{i}^{2}=\phi_{i}^{2}+\psi_{i}^{2}$. Let $\phi^{2}=\sum_{i=1}^{n}\phi_{i}^{2}/n$, $\psi^{2}=\sum_{i=1}^{n}\psi_{i}^{2}/n$, and $\delta^{2}=\sum_{i=1}^{n}\delta_{i}^{2}/n$ (thus, $\delta^{2}=\phi^{2}+\psi^{2}$), the optimizer of problem \eqref{eq:opt_cov} is
  \[
    \Sigma_{i}^{*}=\frac{\psi^{2}}{\delta^{2}}\mu\boldsymbol{\mathrm{I}}+\frac{\phi^{2}}{\delta^{2}}\bS_{i}, \quad i=1,\dots,n.
  \]
  The minimum value of the function is
  \begingroup
  \allowdisplaybreaks
  \begin{eqnarray*}
    && \frac{1}{n}\sum_{i=1}^{n}\mathbb{E}\left\{\bgamma^\top\Sigma_{i}^{*}\bgamma-\exp(\bx_{i}^\top\bbeta)\right\}^{2} \\
    &=& \frac{1}{n}\sum_{i=1}^{n}\mathbb{E}\left\{\frac{\psi^{2}}{\delta^{2}}\mu\bgamma^\top\bgamma+\frac{\phi^{2}}{\delta^{2}}\bgamma^\top\bS_{i}\bgamma-\frac{\psi^{2}+\phi^{2}}{\delta^{2}}\exp(\bx_{i}^\top\bbeta)\right\}^{2} \\
    &=& \frac{1}{n}\sum_{i=1}^{n}\left[\mathbb{E}\left\{\frac{\psi^{2}}{\delta^{2}}\mu\bgamma^\top\bgamma-\frac{\psi^{2}}{\delta^{2}}\exp(\bx_{i}^\top\bbeta)\right\}^{2}+\mathbb{E}\left\{\frac{\phi^{2}}{\delta^{2}}\bgamma^\top\bS_{i}\bgamma-\frac{\phi^{2}}{\delta^{2}}\exp(\bx_{i}^\top\bbeta)\right\}^{2}\right] \\
    &=& \frac{1}{n}\sum_{i=1}^{n}\left(\frac{\psi^{4}}{\delta^{4}}\phi_{i}^{2}+\frac{\phi^{4}}{\delta^{4}}\psi_{i}^{2}\right) \\
    &=& \frac{\psi^{4}\phi^{2}+\phi^{4}\psi^{2}}{\delta^{4}} \\
    &=& \frac{\phi^{2}\psi^{2}}{\delta^{2}}.
  \end{eqnarray*}
  \endgroup
\end{proof}

\subsection{Proof of Proposition~\ref{prop:asymp_alg1}}
\label{appendix:sub:proof_asymp_alg1}

\begin{proof}
  Under Assumptions A2 and A5, the eigenvectors of $\bar{\bS}$ are consistent estimators of $\Pi$. Replace $\bgamma$ with its estimate in Theorems~\ref{thm:Si_consistent}--\ref{thm:optimal2} and Theorem~\ref{thm:beta_consist}, the consistency of $\bbeta$ follows.
\end{proof}

\subsection{Proof of Lemma~\ref{lemma:shpar_bound}}

\begin{proof}
  \begin{enumerate}[(1)]
    \item For $\mu$,
      \[
        \mu=\frac{1}{n(\bgamma^\top\bgamma)}\sum_{i=1}^{n}\exp(\bx_{i}^\top\bbeta)=\frac{1}{n}\sum_{i=1}^{n}\frac{\bgamma^\top\Sigma_{i}\bgamma}{\bgamma^\top\bgamma}\leq\frac{1}{n}\sum_{i=1}^{n}\|\Sigma_{i}\|_{2}^{2}.
      \] 
    Under Assumption A2,
      \begin{eqnarray*}
        \frac{1}{n}\sum_{i=1}^{n}\|\Sigma_{i}\|_{2}^{2} &=& \frac{1}{n}\sum_{i=1}^{n}\|\Lambda_{i}\|_{2}^{2} \\
        &\leq& \frac{1}{n}\sum_{i=1}^{n}\|\Lambda_{i}\|_{F}^{2} \\
        &=& \frac{1}{n}\sum_{i=1}^{n}\left\{\frac{1}{p}\sum_{j=1}^{p}\mathbb{E}(z_{i1j}^{2})^{2}\right\} \\
        &=& \frac{1}{n}\sum_{i=1}^{n}\left\{\frac{1}{p}\sum_{j=1}^{p}\mathbb{E}(z_{i1j}^{4})\right\} \\
        &\leq& \frac{1}{n}\sum_{i=1}^{n}\sqrt{\frac{1}{p}\sum_{j=1}^{p}\mathbb{E}(z_{i1j})^{8}} \\
        &\leq& \frac{1}{n}\sum_{i=1}^{n}\sqrt{C_{2}} \\
        &=& \sqrt{C_{2}},
      \end{eqnarray*}
    where $\|\cdot\|_{F}$ is the Frobenius norm of a matrix.
    \item For $\phi^{2}$, upper limits of $\phi_{i}^{2}$ is derived first.
      \begin{eqnarray*}
        \phi_{i}^{2} &=& \left\{\mu(\bgamma^\top\bgamma)-\exp(\bx_{i}^\top\bbeta)\right\}^{2} \\
        &\leq& \mu^{2}(\bgamma^\top\bgamma)^{2}+\{\exp(\bx_{i}^\top\bbeta)\}^{2} \\
        &=& \mu^{2}(\bgamma^\top\bgamma)^{2}+(\bgamma^\top\Sigma_{i}\bgamma)^{2} \\
        &\leq& \left(\mu^{2}+\|\Sigma_{i}\|_{2}^{4}\right)(\bgamma^\top\bgamma)^{2}.
      \end{eqnarray*}
    From the above derivation, we have
      \[
        \mu^{2}\leq C_{2}, \text{ and } \|\Sigma_{i}\|_{2}^{2}=\|\Lambda_{i}\|_{2}^{2}\leq\|\Lambda_{i}\|_{F}^{2}\leq\sqrt{C_{2}}.
      \]
    Since $\bgamma$ is given, without loss of generality, assume that $\|\bgamma\|_{2}=1$, i.e., $\bgamma^\top\bgamma=1$. Then,
      \[
        \phi_{i}^{2}\leq 2C_{2}(\bgamma^\top\bgamma)=2C_{2}.
      \]
    Thus,
      \[
        \phi^{2}=\frac{1}{n}\sum_{i=1}^{n}\phi_{i}^{2}\leq 2C_{2}.
      \]
    \item For $\psi^{2}$, analogously, $\psi_{i}^{2}$ is considered first.
      \[
        \psi_{i}^{2}=\mathbb{E}\left\{\bgamma^\top\bS_{i}\bgamma-\exp(\bx_{i}^\top\bbeta)\right\}^{2}=\mathbb{E}\left\{\bgamma^\top(\bS_{i}-\Sigma_{i})\bgamma\right\}^{2}\leq(\bgamma^\top\bgamma)^{2}\mathbb{E}\|\bS_{i}-\Sigma_{i}\|_{2}^{2}
      \]
      \begingroup
      \allowdisplaybreaks
      \begin{eqnarray*}
        \mathbb{E}\|\bS_{i}-\Sigma_{i}\|_{F}^{2} &=& \frac{1}{p}\sum_{j=1}^{p}\sum_{k=1}^{p}\mathbb{E}\left\{\left(\frac{1}{T_{i}}\sum_{t=1}^{T_{i}}y_{itj}y_{itk}-\sigma_{ijk}\right)^{2}\right\} \\
        &=& \frac{1}{p}\sum_{j=1}^{p}\sum_{k=1}^{p}\mathbb{E}\left\{\left(\frac{1}{T_{i}}\sum_{t=1}^{T_{i}}z_{itj}z_{itk}-\lambda_{ijk}\right)^{2}\right\} \\
        &=& \frac{1}{p}\sum_{j=1}^{p}\sum_{k=1}^{p}\mathrm{Var}\left(\frac{1}{T_{i}}\sum_{t=1}^{T_{i}}z_{itj}z_{itk}\right) \\
        &=& \frac{1}{p}\sum_{j=1}^{p}\sum_{k=1}^{p}\frac{1}{T_{i}}\mathrm{Var}(z_{i1j}z_{i1k}) \\
        &\leq& \frac{1}{pT_{i}}\sum_{j=1}^{p}\sum_{k=1}^{p}\mathbb{E}(z_{i1j}^{2}z_{i1k}^{2}) \\
        &\leq& \frac{1}{pT_{i}}\sum_{j=1}^{p}\sum_{k=1}^{p}\sqrt{\mathbb{E}z_{i1j}^{4}}\sqrt{\mathbb{E}z_{i1k}^{4}} \\
        &\leq& \frac{p}{T_{i}}\left(\frac{1}{p}\sum_{j=1}^{p}\sqrt{\mathbb{E}z_{i1j}^{4}}\right)^{2} \\
        &\leq& \frac{p}{T_{i}}\left(\frac{1}{p}\sum_{j=1}^{p}\mathbb{E}z_{i1j}^{4}\right) \\
        &\leq& \frac{p}{T_{i}}\sqrt{\frac{1}{p}\sum_{j=1}^{p}\mathbb{E}z_{i1j}^{8}} \\
        &\leq& C_{1}\sqrt{C_{2}}
      \end{eqnarray*}
      \endgroup
    Thus, for $\psi^{2}$,
      \[
        \psi^{2}=\frac{1}{n}\sum_{i=1}^{n}\psi_{i}^{2}\leq\frac{1}{n}\sum_{i=1}^{n}(\bgamma^\top\bgamma)^{2}C_{1}\sqrt{C_{2}}=C_{1}\sqrt{C_{2}}.
      \]
    \item Finally, for $\delta^{2}$, 
      \[
        \delta^{2}=\phi^{2}+\psi^{2}\leq 2C_{2}+C_{1}\sqrt{C_{2}}.
      \]
  \end{enumerate}
\end{proof}

\subsection{Proof of Lemma~\ref{lemma:shpar_consist}}
\begin{proof}
  In the proof of Lemma~\ref{lemma:shpar_consist}, here, it is assumed that $\bgamma$ is a column of $\Pi_{i}$ indexed by $j_{i}$, for $i=1,\dots,n$ (Assumption A4).
  \begin{enumerate}[(i)]
    \item First, we prove the consistency of $\hat{\delta}_{i}^{2}$.
      \begin{eqnarray*}
        \hat{\delta}_{i}^{2}-\delta_{i}^{2} &=& \left\{\bgamma^\top\bS_{i}\bgamma-\mu(\bgamma^\top\bgamma)\right\}^{2}-\mathbb{E}\left\{\bgamma^\top\bS_{i}\bgamma-\mu(\bgamma^\top\bgamma)\right\}^{2} \\
        &=& \left\{(\bgamma^\top\bS_{i}\bgamma)^{2}-\mathbb{E}(\bgamma^\top\bS_{i}\bgamma)^{2}\right\}-2\mu(\bgamma^\top\bgamma)\left\{(\bgamma^\top\bS_{i}\bgamma)-\mathbb{E}(\bgamma^\top\bS_{i}\bgamma)\right\}
      \end{eqnarray*}
    Under Assumption A4, 
      \[
        \bgamma^\top\bS_{i}\bgamma=\frac{1}{T_{i}}\sum_{t=1}^{T_{i}}\bgamma^\top\by_{it}\by_{it}^\top\bgamma=\frac{1}{T_{i}}\sum_{t=1}^{T_{i}}z_{itj_{i}}^{2}.
      \]
      \[
        (\bgamma^\top\bS_{i}\bgamma)^{2}=\frac{1}{T_{i}^{2}}\left(\sum_{t=1}^{T_{i}}z_{itj_{i}}^{2}\right)^{2}=\frac{1}{T_{i}^{2}}\sum_{t=1}^{T_{i}}z_{itj_{i}}^{4}+\frac{1}{T_{i}^{2}}\sum_{t\neq s}z_{itj_{i}}^{2}z_{isj_{i}}^{2}.
      \]
      \[
        \mathbb{E}(\bgamma^\top\bS_{i}\bgamma)^{2}=\frac{1}{T_{i}^{2}}T_{i}\mathbb{E}z_{i1j_{i}}^{4}+\frac{1}{T_{i}}T_{i}(T_{i}-1)\left(\mathbb{E}z_{itj_{i}}^{2}\right)^{2}=\frac{1}{T_{i}}\mathbb{E}z_{i1j_{i}}^{4}+\frac{T_{i}(T_{i}-1)}{T_{i}^{2}}(\bgamma^\top\Sigma_{i}\bgamma)^{2}.
      \]
    For $\forall~\epsilon>0$,
      \begin{eqnarray*}
        \mathbb{P}\left\{|(\bgamma^\top\bS_{i}\bgamma)-\mathbb{E}(\bgamma^\top\bS_{i}\bgamma)|\geq\epsilon\right\} &\leq& \frac{1}{\epsilon^{2}}\mathrm{Var}(\bgamma^\top\bS_{i}\bgamma) \\
        &=& \frac{1}{\epsilon^{2}}\left[\mathbb{E}(\bgamma^\top\bS_{i}\bgamma)^{2}-\left\{\mathbb{E}(\bgamma^\top\bS_{i}\bgamma)\right\}^{2}\right] \\
        &=& \frac{1}{\epsilon^{2}}\left\{\frac{1}{T_{i}^{2}}\mathbb{E}z_{i1j_{i}}^{4}+\frac{T_{i}(T_{i}-1)}{T_{i}^{2}}(\bgamma^\top\Sigma_{i}\bgamma)^{2}-(\bgamma^\top\Sigma_{i}\bgamma)^{2}\right\} \\
        &\overset{T_{i}\rightarrow\infty}{\longrightarrow}& 0.
      \end{eqnarray*}
      \begingroup
      \allowdisplaybreaks
      \begin{eqnarray*}
        && \mathbb{E}(\bgamma^\top\bS_{i}\bgamma)^{4} \\
        &=& \frac{1}{T_{i}^{4}}\mathbb{E}\left(\sum_{t=1}^{T_{i}}z_{itj_{i}}^{2}\right)^{4} \\
        &=& \frac{1}{T_{i}^{4}}\left\{\sum_{t}\mathbb{E}z_{itj_{i}}^{8}+2\sum_{t\neq s}\mathbb{E}z_{itj_{i}}^{4}z_{isj_{i}}^{4}+2\sum_{u}\mathbb{E}\left(z_{iuj_{i}}^{4}\sum_{t\neq s}z_{itj_{i}}^{2}z_{isj_{i}}^{2}\right)+\sum_{u\neq v}\sum_{t\neq s}\mathbb{E}\left(z_{itj_{i}}^{2}z_{isj_{i}}^{2}z_{iuj_{i}}^{2}z_{ivj_{i}}^{2}\right)\right\} \\
        &=& \frac{1}{T_{i}^{4}}\left\{T_{i}\mathbb{E}z_{i1j_{i}}^{8}+2T_{i}(T_{i}-1)(\mathbb{E}z_{i1j_{i}}^{4})^{2}+2T_{i}^{2}(T_{i}-1)\mathbb{E}z_{i1j_{i}}^{4}(\mathbb{E}z_{i1j_{i}}^{2})^{2}+T_{i}^{2}(T_{i}-1)^{2}(\mathbb{E}z_{i1j_{i}}^{2})^{4}\right\}.
      \end{eqnarray*}
      \endgroup
      \[
        \left\{\mathbb{E}(\bgamma^\top\bS_{i}\bgamma)^{2}\right\}^{2}=\frac{1}{T_{i}^{2}}(\mathbb{E}z_{i1j_{i}}^{4})^{2}+\frac{2T_{i}(T_{i}-1)}{T_{i}^{3}}\mathbb{E}z_{i1j_{i}}^{4}(\bgamma\Sigma_{i}\bgamma)^{2}+\frac{T_{i}^{2}(T_{i}-1)^{2}}{T_{i}^{4}}(\bgamma\Sigma_{i}\bgamma)^{4}.
      \]
    For $\forall~\epsilon>0$,
      \begin{eqnarray*}
        \mathbb{P}\left\{|(\bgamma^\top\bS_{i}\bgamma)^{2}-\mathbb{E}(\bgamma^\top\bS_{i}\bgamma)^{2}|\geq\epsilon\right\} &\leq& \frac{1}{\epsilon^{2}}\mathrm{Var}(\bgamma^\top\bS_{i}\bgamma)^{2} \\
        &=& \frac{1}{\epsilon^{2}}\left[\mathbb{E}(\bgamma^\top\bS_{i}\bgamma)^{4}-\left\{\mathbb{E}(\bgamma^\top\bS_{i}\bgamma)^{2}\right\}^{2}\right] \\
        &=& \frac{1}{\epsilon^{2}}\left\{\frac{1}{T_{i}^{3}}\mathbb{E}z_{i1j_{i}}^{8}+\frac{T_{i}-2}{T_{i}^{3}}(\mathbb{E}z_{i1j_{i}}^{4})^{2}\right\} \\
        &\overset{T_{i}\rightarrow\infty}{\longrightarrow}& 0.
      \end{eqnarray*}
    Therefore, as $T_{\min}=\min_{i}T_{i}\rightarrow\infty$,
      \[
        \mathbb{E}\left(\hat{\delta}_{i}^{2}-\delta_{i}^{2}\right)^{2}\rightarrow 0, \text{ for } i=1,\dots,n, \text{ and } \mathbb{E}\left(\hat{\delta}^{2}-\delta^{2}\right)^{2}\rightarrow 0.
      \]
    \item Secondly, prove the consistency of $\hat{\psi}_{i}^{2}$, for $i=1,\dots,n$.
      \[
        \hat{\psi}_{i}^{2}-\psi_{i}^{2}=\frac{1}{T_{i}}\left\{\bgamma^\top\bS_{i}\bgamma-\exp(\bx_{i}^\top\bbeta)\right\}^{2}-\mathbb{E}\left\{\bgamma^\top\bS_{i}\bgamma-\exp(\bx_{i}^\top\bbeta)\right\}^{2}.
      \]
      \begin{eqnarray*}
        \mathbb{E}\left\{\bgamma^\top\bS_{i}\bgamma-\exp(\bx_{i}^\top\bbeta)\right\}^{2} &=& \mathbb{E}\left\{\frac{1}{T_{i}}\sum_{t}z_{itj_{i}}^{2}-\exp(\bx_{i}^\top\bbeta)\right\}^{2} \\
        &=& \frac{1}{T_{i}^{2}}\sum_{t}\mathrm{Var}(z_{itj_{i}}^{2}) \\
        &=& \frac{1}{T_{i}}\mathrm{Var}(z_{i1j_{i}}^{2}).
      \end{eqnarray*}
      \begin{eqnarray*}
        \hat{\psi}_{i}^{2}-\psi_{i}^{2} &=& \frac{1}{T_{i}}\left[\left\{\bgamma^\top\bS_{i}\bgamma-\exp(\bx_{i}^\top\bbeta)\right\}^{2}-\mathrm{Var}(z_{i1j_{i}}^{2})\right] \\
        &=& \frac{1}{T_{i}}\left[(\bgamma^\top\bS_{i}\bgamma)^{2}-\mathbb{E}z_{i1j_{i}}^{4}-2\exp(\bx_{i}^\top\bbeta)\left\{\bgamma^\top\bS_{i}\bgamma-\exp(\bx_{i}^\top\bbeta)\right\}\right].
      \end{eqnarray*}
    From above derivation and the fact that $\mathbb{E}(\bgamma^\top\bS_{i}\bgamma)=\bgamma^\top\Sigma_{i}\bgamma=\exp(\bx_{i}^\top\bbeta)$, as $T_{i}\rightarrow\infty$, for $\forall~\epsilon>0$,
      \[
        \mathbb{P}\left\{|(\bgamma^\top\bS_{i}\bgamma)-\mathbb{E}(\bgamma^\top\bS_{i}\bgamma)|\geq\epsilon\right\}\rightarrow 0.
      \]
    As both $(\bgamma^\top\bS_{i}\bgamma)^{2}$ and $\mathbb{E}z_{i1j_{i}}^{4}$ are bounded, then, as $T_{\min}=\min_{i}T_{i}\rightarrow\infty$,
      \[
        \mathbb{E}\left(\hat{\psi}_{i}^{2}-\psi_{i}^{2}\right)^{2}\rightarrow0, \text{ for } i=1,\dots,n.
      \]
    Let $\tilde{\psi}_{i}^{2}=\min(\hat{\psi}_{i}^{2},\hat{\delta}_{i}^{2})$.
      \[
        \tilde{\psi}_{i}^{2}-\psi_{i}^{2}=\min(\hat{\psi}_{i}^{2},\hat{\delta}_{i}^{2})-\psi_{i}^{2}\leq\hat{\psi}_{i}^{2}-\psi_{i}^{2}\leq|\hat{\psi}_{i}^{2}-\psi_{i}^{2}|\leq\max\left(|\hat{\psi}_{i}^{2}-\psi_{i}^{2}|,|\hat{\delta}_{i}^{2}-\delta_{i}^{2}|\right).
      \]
    $\delta_{i}^{2}=\phi_{i}^{2}+\psi_{i}^{2}\geq\psi_{i}^{2}$, then
      \begin{eqnarray*}
        \tilde{\psi}_{i}^{2}-\psi_{i}^{2} &=& \min(\hat{\psi}_{i}^{2},\hat{\delta}_{i}^{2})-\psi_{i}^{2} \\
        &=& \min\left(\hat{\psi}_{i}^{2}-\psi_{i}^{2},\hat{\delta}_{i}^{2}-\psi_{i}^{2}\right) \\
        &\geq& \min\left(\hat{\psi}_{i}^{2}-\psi_{i}^{2},\hat{\delta}_{i}^{2}-\delta_{i}^{2}\right) \\
        &\geq& \min\left(-|\hat{\psi}_{i}-\psi_{i}^{2}|,-|\hat{\delta}_{i}^{2}-\delta_{i}^{2}|\right) \\
        &\geq& -\max\left(|\hat{\psi}_{i}-\psi_{i}^{2}|,|\hat{\delta}_{i}^{2}-\delta_{i}^{2}|\right).
      \end{eqnarray*}
      \[
        \mathbb{E}(\tilde{\psi}_{i}^{2}-\psi_{i}^{2})^{2}\leq\mathbb{E}\left\{\max\left(|\hat{\psi}_{i}-\psi_{i}^{2}|,|\hat{\delta}_{i}^{2}-\delta_{i}^{2}|\right)^{2}\right\}\leq\mathbb{E}(\hat{\psi}_{i}^{2}-\psi_{i}^{2})^{2}+\mathbb{E}(\hat{\delta}_{i}^{2}-\delta_{i}^{2})^{2}.
      \]
    Therefore, as $T_{\min}=\min_{i}T_{i}\rightarrow\infty$,
      \[
        \mathbb{E}\left(\tilde{\psi}_{i}^{2}-\psi_{i}^{2}\right)^{2}\rightarrow0, \text{ for } i=1,\dots,n, \text{ and } \mathbb{E}\left(\hat{\psi}^{2}-\psi^{2}\right)^{2}\rightarrow0.
      \]
    \item Lastly, $\hat{\phi}_{i}^{2}=\hat{\delta}_{i}^{2}-\hat{\psi}_{i}^{2}$. The consistency of $\hat{\phi}_{i}^{2}$ (for $i=1,\dots,n$) and $\hat{\phi}^{2}$ are straightforward.
  \end{enumerate}
\end{proof}

\subsection{Proof of Theorem~\ref{thm:Si_consistent}}

In order to prove Theorem~\ref{thm:Si_consistent}, the following lemma is firstly introduced. This lemma is also used to prove Lemma~\ref{appendix:lemma:const_prod} in the next section.
\begin{lemma}\label{appendix:lemma:converg_seq}
  If $a_{i}^{2}$ is a sequence of nonnegative random variables (implicitly indexed by $T_{i}$) whose expectations converge to zero, for $i=1,\dots,n$, and $\kappa_{1},\kappa_{2}$ are two nonrandom scalars, and
  \[
    \frac{a_{i}^{2}}{\hat{\delta}_{i}^{\kappa_{1}}\delta_{i}^{\kappa_{2}}}\leq 2(\hat{\delta}_{i}^{2}+\delta_{i}^{2}) \quad \text{a.s.},
  \]
  then, as $T_{\min}=\min_{i}T_{i}\rightarrow\infty$,
  \[
    \mathbb{E}\left(\frac{a_{i}^{2}}{\hat{\delta}_{i}^{\kappa_{1}}\delta_{i}^{\kappa_{2}}}\right)\rightarrow 0.
  \]
  Analogously, if $a^{2}$ is a sequence of nonnegative random variables (implicitly indexed by $T_{\min}=\min_{i}T_{i}$) whose expectations converge to zero, and $\kappa_{1},\kappa_{2}$ are two nonrandom scalars, and
  \[
    \frac{a^{2}}{\hat{\delta}^{\kappa_{1}}\delta^{\kappa_{2}}}\leq 2(\hat{\delta}^{2}+\delta^{2}) \quad \text{a.s.},
  \]
  then, as $T_{\min}=\min_{i}T_{i}\rightarrow\infty$,
  \[
    \mathbb{E}\left(\frac{a^{2}}{\hat{\delta}^{\kappa_{1}}\delta^{\kappa_{2}}}\right)\rightarrow 0.
  \]
\end{lemma}
\begin{proof}
  For a fixed $\epsilon>0$, let $\mathcal{T}_{i}$ denote the set of indices $T_{i}$ such that $\delta_{i}^{2}\leq\epsilon/8$. In Lemma~\ref{lemma:shpar_consist}, it is proved that $\mathbb{E}(\hat{\delta}_{i}^{2}-\delta_{i}^{2})^{2}\rightarrow 0$. Thus, there exists an integer $T_{i1}$ such that $\forall~T_{i}\geq T_{i1}$,
    \[
      \mathbb{E}|\hat{\delta}_{i}^{2}-\delta_{i}^{2}|\leq\epsilon/4.
    \]
  For $\forall~T_{i}\geq T_{i1}$ in the set $\mathcal{T}_{i}$,
    \[
      \mathbb{E}\left(\frac{a_{i}^{2}}{\hat{\delta}_{i}^{\kappa_{1}}\delta_{i}^{\kappa_{2}}}\right)\leq 2\left(\mathbb{E}\hat{\delta}_{i}^{2}+\delta_{i}^{2}\right)\leq 2\left(\mathbb{E}|\hat{\delta}_{i}^{2}-\delta_{i}^{2}|+2\delta_{i}^{2}\right)\leq 2\left(\frac{\epsilon}{4}+2\times\frac{\epsilon}{8}\right)=\epsilon.
    \]
  Consider the complementary of set $\mathcal{T}_{i}$, since $\mathbb{E}a_{i}^{2}\rightarrow 0$, there exists an integer $T_{i2}$ such that, $\forall~T_{i}\geq T_{i2}$,
    \[
      \mathbb{E}a^{2}\leq\frac{\epsilon^{\kappa_{1}+\kappa_{2}+1}}{2^{4\kappa_{1}+3\kappa_{2}+1}}.
    \]
  $\delta_{i}^{2}$ is bounded by $2C_{2}+C_{1}\sqrt{C_{2}}$. Then, there exists an integer $T_{i3}$ such that, for $\forall~T_{i}\geq T_{i3}$
    \[
      \mathbb{P}\left(|\hat{\delta}_{i}^{2}-\delta_{i}^{2}|\geq\frac{\epsilon}{16}\right)\leq\frac{4\epsilon}{16(2C_{2}+C_{1}\sqrt{C_{2}})+\epsilon}.
    \]
  Let $\boldsymbol{\mathrm{1}}_{\{\cdot\}}$ denote the indicator function. For $\forall~T_{i}\geq\max(T_{i2},T_{i3})$ outside the set $\mathcal{T}_{i}$, then
  \begin{eqnarray*}
    && \mathbb{E}\left(\frac{a_{i}^{2}}{\hat{\delta}_{i}^{\kappa_{1}}\delta_{i}^{\kappa_{2}}}\right) \\
    &=& \mathbb{E}\left(\frac{a_{i}^{2}}{\hat{\delta}_{i}^{\kappa_{1}}\delta_{i}^{\kappa_{2}}}\boldsymbol{\mathrm{1}}_{\{\hat{\delta}_{i}^{2}\leq\epsilon/16\}}\right)+\mathbb{E}\left(\frac{a_{i}^{2}}{\hat{\delta}_{i}^{\kappa_{1}}\delta_{i}^{\kappa_{2}}}\boldsymbol{\mathrm{1}}_{\{\hat{\delta}_{i}^{2}>\epsilon/16\}}\right) \\
    &\leq& \mathbb{E}\left\{2(\hat{\delta}_{i}^{2}+\delta_{i}^{2})\boldsymbol{\mathrm{1}}_{\{\hat{\delta}_{i}^{2}\leq\epsilon/16\}}\right\}+\left(\frac{16}{\epsilon}\right)^{\kappa_{1}}\left(\frac{8}{\epsilon}\right)^{\kappa_{2}}\mathbb{E}\left(a_{i}^{2}\boldsymbol{\mathrm{1}}_{\{\hat{\delta}_{i}^{2}>\epsilon/16\}}\right) \\
    &\leq& 2\left\{(2C_{2}+C_{1}\sqrt{C_{2}})+\frac{\epsilon}{16}\right\}\mathbb{P}\left(|\hat{\delta}_{i}^{2}-\delta_{i}^{2}|\geq\frac{\epsilon}{16}\right)+\left(\frac{16}{\epsilon}\right)^{\kappa_{1}}\left(\frac{8}{\epsilon}\right)^{\kappa_{2}}\mathbb{E}(a_{i}^{2}) \\
    &\leq& 2\left\{(2C_{2}+C_{1}\sqrt{C_{2}})+\frac{\epsilon}{16}\right\}\frac{4\epsilon}{16(2C_{2}+C_{1}\sqrt{C_{2}})+\epsilon}+\left(\frac{16}{\epsilon}\right)^{\kappa_{1}}\left(\frac{8}{\epsilon}\right)^{\kappa_{2}}\frac{\epsilon^{\kappa_{1}+\kappa_{2}+1}}{2^{4\kappa_{1}+3\kappa_{2}+1}} \\
    &\leq& \epsilon.
  \end{eqnarray*}
  Bringing together the results inside and outside the set $\mathcal{T}_{i}$, for $\forall~T_{i}\geq\max(T_{i1},T_{i2},T_{i3})$,
    \[
      \mathbb{E}\left(\frac{a_{i}^{2}}{\hat{\delta}_{i}^{\kappa_{1}}\delta_{i}^{\kappa_{2}}}\right)\leq\epsilon.
    \]
  The proof of the second part follows the same strategy.
\end{proof}
Now, we prove Theorem~\ref{thm:Si_consistent}.
\begin{proof}
  We first prove that $\bS_{i}^{*}$ is a consistent estimator of $\Sigma_{i}^{*}$.
  \begin{eqnarray*}
    \|\bS_{i}^{*}-\Sigma_{i}^{*}\|^{2} &=& \max_{\bgamma\neq\boldsymbol{\mathrm{0}}}\frac{\|\bgamma^\top(\bS_{i}^{*}-\Sigma_{i}^{*})\bgamma\|^{2}}{\bgamma^\top\bgamma} \\
    &=& \max_{\bgamma\neq\boldsymbol{\mathrm{0}}}\frac{1}{\bgamma^\top\bgamma}\left\|\left(\frac{\hat{\phi}^{2}}{\hat{\delta}^{2}}-\frac{\phi^{2}}{\delta^{2}}\right)\left(\bgamma^\top\bS_{i}\bgamma-\mu\bgamma^\top\bgamma\right)\right\|^{2} \\
    &=& \max_{\bgamma\neq\boldsymbol{\mathrm{0}}}\frac{1}{\bgamma^\top\bgamma}\left(\frac{\hat{\phi}^{2}}{\hat{\delta}^{2}}-\frac{\phi^{2}}{\delta^{2}}\right)^{2}\hat{\delta}_{i}^{2}.
  \end{eqnarray*}
  \begin{eqnarray*}
    \frac{1}{n}\sum_{i=1}^{n}\|\bS_{i}^{*}-\Sigma_{i}^{*}\|^{2} &=& \max_{\bgamma\neq\boldsymbol{\mathrm{0}}}\frac{1}{\bgamma^\top\bgamma}\frac{(\hat{\phi}^{2}\delta^{2}-\phi^{2}\hat{\delta}^{2})^{2}}{\hat{\delta}^{4}\delta^{4}}\frac{1}{n}\sum_{i=1}^{n}\hat{\delta}_{i}^{2} \\
    &=& \max_{\bgamma\neq\boldsymbol{\mathrm{0}}}\frac{1}{\bgamma^\top\bgamma}\frac{(\hat{\phi}^{2}\delta^{2}-\phi^{2}\hat{\delta}^{2})^{2}}{\hat{\delta}^{2}\delta^{4}}.
  \end{eqnarray*}
  Using the fact that $\phi^{2}\leq\delta^{2}$ and $\hat{\phi}^{2}\leq\hat{\delta}^{2}$,
    \[
      \frac{(\hat{\phi}^{2}\delta^{2}-\phi^{2}\hat{\delta}^{2})^{2}}{\hat{\delta}^{2}\delta^{4}}\leq\hat{\delta}^{2}\leq2(\hat{\delta}^{2}+\delta^{2}).
    \]
  In Lemma~\ref{lemma:shpar_consist}, it is shown that $\mathbb{E}(\hat{\phi}^{2}-\phi^{2})^{2}$ and $\mathbb{E}(\hat{\delta}^{2}-\delta^{2})^{2}$ converge to zero. In addition, Lemma~\ref{lemma:shpar_bound} shows that $\phi^{2}$ and $\delta^{2}$ are bounded. Thus,
    \begin{eqnarray*}
      \mathbb{E}\left(\hat{\phi}^{2}\delta^{2}-\phi^{2}\hat{\delta}^{2}\right)^{2} &=& \mathbb{E}\left\{(\hat{\phi}^{2}-\phi^{2})\delta^{2}-\phi^{2}(\hat{\delta}^{2}-\delta^{2})\right\}^{2} \\
      &\leq& \delta^{4}\mathbb{E}(\hat{\phi}^{2}-\phi^{2})^{2}+\phi^{4}\mathbb{E}(\hat{\delta}^{2}-\delta^{2})^{2} \\
      &\rightarrow& 0.
    \end{eqnarray*}
  Let $a^{2}=(\hat{\phi}^{2}\delta^{2}-\phi^{2}\hat{\delta}^{2})^{2}$, $\kappa_{1}=2$ and $\kappa_{2}=4$, then $\mathbb{E}a^{2}\rightarrow 0$, and using Lemma~\ref{appendix:lemma:converg_seq},
    \[
      \mathbb{E}\frac{(\hat{\phi}^{2}\delta^{2}-\phi^{2}\hat{\delta}^{2})^{2}}{\hat{\delta}^{2}\delta^{4}}\rightarrow 0.
    \]
  Thus,
    \[
      \frac{1}{n}\sum_{i=1}^{n}\mathbb{E}\|\bS_{i}^{*}-\Sigma_{i}^{*}\|^{2}\rightarrow 0.
    \]
  And therefore, for $\forall~i$,
    \[
      \mathbb{E}\|\bS_{i}^{*}-\Sigma_{i}^{*}\|^{2}\rightarrow 0.
    \]
  For the second statement,
    \begin{eqnarray*}
      \mathbb{E}\left|\|\bS_{i}^{*}-\Sigma_{i}\|^{2}-\|\Sigma_{i}^{*}-\Sigma_{i}\|^{2}\right| &=& \mathbb{E}\left|\langle \bS_{i}^{*}-\Sigma_{i}^{*},\bS_{i}^{*}+\Sigma_{i}^{*}-2\Sigma_{i}\rangle\right| \\
      &\leq& \sqrt{\mathbb{E}\|\bS_{i}^{*}-\Sigma_{i}^{*}\|^{2}}\sqrt{\mathbb{E}\|\bS_{i}^{*}+\Sigma_{i}^{*}-2\Sigma_{i}\|^{2}} \\
      &\rightarrow& 0.
    \end{eqnarray*}
  Therefor,
    \[
      \mathbb{E}\left\{\bgamma^\top\bS_{i}^{*}\bgamma-\exp(\bx_{i}^\top\bbeta)\right\}^{2}-\mathbb{E}\left\{\bgamma^\top\Sigma_{i}^{*}\bgamma-\exp(\bx_{i}^\top\bbeta)\right\}^{2}\rightarrow 0.
    \]
\end{proof}

\subsection{Proof of Theorem~\ref{thm:optimal1}}

Before proving Theorem~\ref{thm:optimal1}, we first provide the solution to the optimization problem~\eqref{eq:opt_sample}. Let
\[
  f(\rho_{1},\rho_{2})=\frac{1}{n}\sum_{i=1}^{n}\left\{\bgamma^\top(\rho_{1}\boldsymbol{\mathrm{I}}+\rho_{2}\bS_{i})\bgamma-\exp(\bx_{i}^\top\bbeta)\right\}^{2}.
\]
\[
  \frac{\partial f}{\partial\rho_{1}}=\frac{1}{n}\sum_{i=1}^{n}2(\bgamma^\top\bgamma)\left\{\rho_{1}\bgamma^\top\bgamma+\rho_{2}\bgamma^\top\bS_{i}\bgamma-\exp(\bx_{i}^\top\bbeta)\right\}=0
\]
\[
  \frac{\partial f}{\partial\rho_{2}}=\frac{1}{n}\sum_{i=1}^{n}2(\bgamma^\top\bS_{i}\bgamma)\left\{\rho_{1}(\bgamma^\top\bgamma)+\rho_{2}(\bgamma^\top\bS_{i}\bgamma)-\exp(\bx_{i}^\top\bbeta)\right\}=0.
\]
\[
  \Rightarrow \quad \rho_{2}=\frac{\sum_{i}(\bgamma^\top\bS_{i}\bgamma)\exp(\bx_{i}^\top\bbeta)/n-(\sum_{i}\bgamma^\top\bS_{i}\bgamma/n)(\sum_{i}\exp(\bx_{i}^\top\bbeta)/n)}{\sum_{i}(\bgamma^\top\bS_{i}\bgamma)^{2}/n-(\sum_{i}\bgamma^\top\bS_{i}\bgamma/n)^{2}}.
\]
\begin{eqnarray*}
  \rho_{1} &=& \frac{1}{\bgamma^\top\bgamma}\left\{\frac{1}{n}\sum_{i=1}^{n}\exp(\bx_{i}^\top\bbeta)-\frac{1}{n}\sum_{i=1}^{n}\rho_{2}(\bgamma^\top\bS_{i}\bgamma)\right\} \\
  &=& \frac{1}{\bgamma^\top\bgamma}\frac{(\sum_{i}\bgamma^\top\bS_{i}\bgamma/n)(\sum_{i}(\bgamma^\top\bS_{i}\bgamma)\exp(\bx_{i}^\top\bbeta)/n)-(\sum_{i}\exp(\bx_{i}^\top\bbeta)/n)(\sum_{i}(\bgamma^\top\bS_{i}\bgamma)^{2}/n)}{\sum_{i}(\bgamma^\top\bS_{i}\bgamma)^{2}/n-(\sum_{i}\bgamma^\top\bS_{i}\bgamma/n)^{2}}.
\end{eqnarray*}

In order to prove Theorem~\ref{thm:optimal1}, the following lemma is introduced.
\begin{lemma}\label{appendix:lemma:const_prod}
  For given $(\bgamma,\bbeta)$, let $T_{\min}=\min_{i}T_{i}$, as $T_{\min}\rightarrow\infty$, for $\forall~i\in\{1,\dots,n\}$,
  \[
    \mathbb{E}\left(\left|\frac{\hat{\phi}_{i}^{2}\hat{\psi}_{i}^{2}}{\hat{\delta}_{i}^{2}}-\frac{\phi_{i}^{2}\psi_{i}^{2}}{\delta_{i}^{2}}\right|\right)\rightarrow 0.
  \]
  Then, as $n,T_{\min}\rightarrow \infty$,
  \[
    \mathbb{E}\left(\left|\frac{\hat{\phi}^{2}\hat{\psi}^{2}}{\hat{\delta}^{2}}-\frac{\phi^{2}\psi^{2}}{\delta^{2}}\right|\right)\rightarrow 0.
  \] 
\end{lemma}
\begin{proof}
  \[
    \frac{\hat{\phi}_{i}^{2}\hat{\psi}_{i}^{2}}{\hat{\delta}_{i}^{2}}-\frac{\phi_{i}^{2}\psi_{i}^{2}}{\delta_{i}^{2}}=\frac{\hat{\phi}_{i}^{2}\hat{\psi}_{i}^{2}\delta_{i}^{2}-\phi_{i}^{2}\psi_{i}^{2}\hat{\delta}_{i}^{2}}{\hat{\delta}_{i}^{2}\delta_{i}^{2}}.
  \]
  Let $a_{i}^{2}=|\hat{\phi}_{i}^{2}\hat{\psi}_{i}^{2}\delta_{i}^{2}-\phi_{i}^{2}\psi_{i}^{2}\hat{\delta}_{i}^{2}|$, $\kappa_{1}=2$ and $\kappa_{2}=2$. First need to verify the assumptions in Lemma~\ref{appendix:lemma:converg_seq}.
  \[
    \left|\frac{\hat{\phi}_{i}^{2}\hat{\psi}_{i}^{2}}{\hat{\delta}_{i}^{2}}-\frac{\phi_{i}^{2}\psi_{i}^{2}}{\delta_{i}^{2}}\right|\leq\frac{\hat{\phi}_{i}^{2}\hat{\psi}_{i}^{2}}{\hat{\delta}_{i}^{2}}+\frac{\phi_{i}^{2}\psi_{i}^{2}}{\delta_{i}^{2}}\leq\hat{\phi}_{i}^{2}+\phi_{i}^{2}\leq\hat{\delta}_{i}^{2}+\delta_{i}^{2}\leq 2(\hat{\delta}_{i}^{2}+\delta_{i}^{2}), \quad \text{a.s.}.
  \]
  Furthermore,
    \begin{eqnarray*}
      && \mathbb{E}\left(|\hat{\phi}_{i}^{2}\hat{\psi}_{i}^{2}\delta_{i}^{2}-\phi_{i}^{2}\psi_{i}^{2}\hat{\delta}_{i}^{2}|\right) \\
      &=& \mathbb{E}\left\{\left|(\hat{\phi}_{i}^{2}\hat{\psi}_{i}^{2}-\phi_{i}^{2}\psi_{i}^{2})\delta_{i}^{2}-\phi_{i}^{2}\psi_{i}^{2}(\hat{\delta}_{i}^{2}-\delta_{i}^{2})\right|\right\} \\
      &=& \mathbb{E}\left\{\left|(\hat{\phi}_{i}^{2}-\phi_{i}^{2})(\hat{\psi}_{i}^{2}-\psi_{i}^{2})\delta_{i}^{2}+\phi_{i}^{2}(\hat{\psi}_{i}^{2}-\psi_{i}^{2})\delta_{i}^{2}+(\hat{\phi}_{i}^{2}-\phi_{i}^{2})\psi_{i}^{2}\delta_{i}^{2}-\phi_{i}^{2}\psi_{i}^{2}(\hat{\delta}_{i}^{2}-\delta_{i}^{2})\right|\right\} \\
      &\leq& \sqrt{\mathbb{E}(\hat{\phi}_{i}^{2}-\phi_{i}^{2})^{2}}\sqrt{\mathbb{E}(\hat{\psi}_{i}^{2}-\psi_{i}^{2})^{2}}\delta_{i}^{2}+\phi_{i}^{2}\mathbb{E}|\hat{\psi}_{i}^{2}-\psi_{i}^{2}|\delta_{i}^{2}+\mathbb{E}|\hat{\phi}_{i}^{2}-\phi_{i}^{2}|\psi_{i}^{2}\delta_{i}^{2}-\phi_{i}^{2}\psi_{i}^{2}\mathbb{E}|\hat{\delta}_{i}^{2}-\delta_{i}^{2}|.
    \end{eqnarray*}
  The right-hand side converges to zero. Therefore, $\mathbb{E}a_{i}^{2}\rightarrow 0$, conditions in Lemma~\ref{appendix:lemma:converg_seq} are satisfied. Therefore,
    \[
      \mathbb{E}\left|\frac{\hat{\phi}_{i}^{2}\hat{\psi}_{i}^{2}}{\hat{\delta}_{i}^{2}}-\frac{\phi_{i}^{2}\psi_{i}^{2}}{\delta_{i}^{2}}\right|\rightarrow 0.
    \]
  Analogously, it can be shown that
    \[
      \mathbb{E}\left|\frac{\hat{\phi}^{2}\hat{\psi}^{2}}{\hat{\delta}^{2}}-\frac{\phi^{2}\psi^{2}}{\delta^{2}}\right|\rightarrow 0.
    \]
\end{proof}
Next, we prove Theorem~\ref{thm:optimal1}.
\begin{proof}
  Let $\alpha_{i}=(\bgamma^\top\Sigma_{i}\bgamma)(\bgamma^\top\bS_{i}\bgamma)-\{\mu(\bgamma^\top\bgamma)\}^{2}$ and $\alpha=\sum_{i=1}^{n}\alpha_{i}/n$. $\mathbb{E}(\alpha_{i})=\exp^{2}(\bx_{i}^\top\bbeta)-\mu^{2}(\bgamma^\top\bgamma)^{2}$, then
    \[
      \mathbb{E}\alpha=\frac{1}{n}\sum_{i=1}^{n}\exp^{2}(\bx_{i}^\top\bbeta)-\mu^{2}(\bgamma^\top\bgamma)=\phi^{2}.
    \]
  First, need to prove that $\alpha-\phi^{2}$ converges to zero in quadratic mean.
  \begin{eqnarray*}
    && \mathrm{Var}(\alpha_{i}) \\
    &=& \mathrm{Var}\left\{(\bgamma^\top\Sigma_{i}\bgamma)(\bgamma^\top\bS_{i}\bgamma)-\mu^{2}(\bgamma^\top\bgamma)^{2}\right\} \\
    &=& \mathrm{Var}\left\{(\bgamma^\top\Sigma_{i}\bgamma)(\bgamma^\top\bS_{i}\bgamma)\right\}+\mathrm{Var}\left\{\mu^{2}(\bgamma^\top\bgamma)^{2}\right\}-2\mathrm{Cov}\left\{(\bgamma^\top\Sigma_{i}\bgamma)(\bgamma^\top\bS_{i}\bgamma),\mu^{2}(\bgamma^\top\bgamma)^{2}\right\} \\
    &=& \mathrm{Var}\left\{(\bgamma^\top\Sigma_{i}\bgamma)(\bgamma^\top\bS_{i}\bgamma)\right\}.
  \end{eqnarray*}
  \[
    (\bgamma^\top\Sigma_{i}\bgamma)(\bgamma^\top\bS_{i}\bgamma)=\lambda_{ij_{i}}\left(\frac{1}{T_{i}}\sum_{t=1}^{T_{i}}z_{itj_{i}}^{2}\right).
  \]
  \begingroup
  \allowdisplaybreaks
  \begin{eqnarray*}
    \mathrm{Var}\left\{(\bgamma^\top\Sigma_{i}\bgamma)(\bgamma^\top\bS_{i}\bgamma)\right\} &=& \mathrm{Var}\left\{\frac{1}{T_{i}}\sum_{t=1}^{T_{i}}\lambda_{ij_{i}}z_{itj_{i}}^{2}\right\} \\
    &=& \frac{1}{T_{i}}\mathrm{Var}\left(\lambda_{ij_{i}}z_{i1j_{i}}^{2}\right) \\
    &\leq& \frac{1}{T_{i}}\mathbb{E}\left(\lambda_{ij_{i}}z_{i1j_{i}}^{2}\right)^{2} \\
    &\leq& \frac{1}{T_{i}}\mathbb{E}\lambda_{ij_{i}}^{2}z_{i1j_{i}}^{4} \\
    &\leq& \frac{1}{T_{i}}\left(\mathbb{E}z_{i1j_{i}}^{2}\right)^{2}\mathbb{E}z_{i1j_{i}}^{4} \\
    &\leq& \frac{1}{T_{i}}\left(\mathbb{E}z_{i1j_{i}}^{4}\right)^{2} \\
    &\leq& \frac{1}{T_{i}}\mathbb{E}z_{i1j_{i}}^{8} \\
    &\leq& \frac{C_{2}}{T_{i}}.
  \end{eqnarray*}
  \endgroup
  \[
    \mathrm{Var}(\alpha)=\frac{1}{n^{2}}\sum_{i=1}^{n}\mathrm{Var}(\alpha_{i})\leq\frac{C_{2}}{n^{2}}\sum_{i=1}^{n}\frac{1}{T_{i}}\rightarrow 0, \text{ as } T_{\min}=\min_{i}T_{i}\rightarrow\infty.
  \]
  This proves that $\alpha-\phi^{2}$ converges to $0$ in quadratic mean. In the following, we prove that $\bS_{i}^{*}$ is a consistent estimator of $\Sigma_{i}^{**}$.
  \[
    \bS_{i}^{*}=\frac{\hat{\psi}^{2}}{\hat{\delta}^{2}}\mu\boldsymbol{\mathrm{I}}+\frac{\hat{\phi}^{2}}{\hat{\delta}^{2}}\bS_{i}=\frac{\hat{\delta}^{2}-\hat{\psi}^{2}}{\hat{\delta}^{2}}\mu\boldsymbol{\mathrm{I}}+\frac{\hat{\phi}^{2}}{\hat{\delta}^{2}}\bS_{i}=\mu\boldsymbol{\mathrm{I}}+\frac{\hat{\phi}^{2}}{\hat{\delta}^{2}}(\bS_{i}-\mu\boldsymbol{\mathrm{I}}).
  \]
  \[
    \Sigma_{i}^{**}=\rho_{1}\boldsymbol{\mathrm{I}}+\rho_{2}\bS_{i}=(\rho_{1}+\rho_{2}\mu)\boldsymbol{\mathrm{I}}+\rho_{2}(\bS_{i}-\mu\boldsymbol{\mathrm{I}}).
  \]
  \begin{eqnarray*}
    \frac{1}{n}\sum_{i=1}^{n}\|\bS_{i}^{*}-\Sigma_{i}^{**}\|^{2} &=& \frac{1}{n}\sum_{i=1}^{n}\left\|(\mu-\rho_{1}-\rho_{2}\mu)\boldsymbol{\mathrm{I}}+\left(\frac{\hat{\phi}^{2}}{\hat{\delta}^{2}}-\rho_{2}\right)(\bS_{i}-\mu\boldsymbol{\mathrm{I}})\right\|^{2} \\
    &=& \frac{1}{n}\sum_{i=1}^{n}\left\{\max_{\bgamma\neq 0}~\frac{1}{\bgamma^\top\bgamma}\left\|(\mu-\rho_{1}-\rho_{2}\mu)(\bgamma^\top\bgamma)+\left(\frac{\hat{\phi}^{2}}{\hat{\delta}^{2}}-\rho_{2}\right)(\bgamma^\top\bS_{i}\bgamma-\mu(\bgamma^\top\bgamma))\right\|^{2}\right\} \\
    &=& \max_{\bgamma\neq 0}~\left\{(\mu-\rho_{1}-\rho_{2}\mu)^{2}(\bgamma^\top\bgamma)+\frac{1}{\bgamma^\top\bgamma}\left(\frac{\hat{\phi}^{2}}{\hat{\delta}^{2}}-\rho_{2}\right)^{2}\hat{\delta}_{i}^{2} \right. \\
    && \quad \quad \quad \left.+2(\mu-\rho_{1}-\rho_{2}\mu)\left(\frac{\hat{\phi}^{2}}{\hat{\delta}^{2}}-\rho_{2}\right)\left(\frac{1}{n}\sum_{i=1}^{n}\bgamma^\top\bS_{i}\bgamma-\mu(\bgamma^\top\bgamma)\right)\right\}.
  \end{eqnarray*}
  \begin{eqnarray*}
    && (\mu-\rho_{1}-\rho_{2}\mu)^{2} \\
    &=& \frac{(\sum_{i}\bgamma^\top\bS_{i}\bgamma/n-\sum_{i}\exp(\bx_{i}^\top\bbeta)/n)^{2}\left\{(\sum_{i}\bgamma^\top\bS_{i}\bgamma/n)(\sum_{i}\exp(\bx_{i}^\top\bbeta)/n)-\sum_{i}(\bgamma^\top\bS_{i}\bgamma)\exp(\bx_{i}^\top\bbeta)/n\right\}^{2}}{(\bgamma^\top\bgamma)^{2}\left\{(\sum_{i}\bgamma^\top\bS_{i}\bgamma/n)^{2}-\sum_{i}(\bgamma^\top\bS_{i}\bgamma)^{2}/n\right\}^{2}}.
  \end{eqnarray*}
  \begin{eqnarray*}
    && \mathbb{E}\left\{\frac{1}{n}\sum_{i}\bgamma^\top\bS_{i}\bgamma-\frac{1}{n}\sum_{i}\exp(\bx_{i}^\top\bbeta)\right\}^{2} \\
    &=& \frac{1}{n^{2}}\sum_{i}\mathbb{E}\left\{\bgamma^\top\bS_{i}\bgamma-\exp(\bx_{i}^\top\bbeta)\right\}^{2}+\frac{1}{n^{2}}\sum_{i\neq i'}\mathbb{E}\left\{\bgamma^\top\bS_{i}\bgamma-\exp(\bx_{i}^\top\bbeta)\right\}\left\{\bgamma^\top\bS_{i'}\bgamma-\exp(\bx_{i'}^\top\bbeta)\right\}.
  \end{eqnarray*}
  \begin{eqnarray*}
    \mathbb{E}\left\{\bgamma^\top\bS_{i}\bgamma-\exp(\bx_{i}^\top\bbeta)\right\}^{2} &=& \mathbb{E}\left\{\bgamma^\top\bS_{i}\bgamma-\mathbb{E}(\bgamma^\top\bS_{i}\bgamma)\right\}^{2} \\
    &=& \mathrm{Var}(\bgamma^\top\bS_{i}\bgamma) \\
    &=& \frac{1}{T_{i}}\mathbb{E}z_{i1j_{i}}^{4}+\frac{T_{i}(T_{i}-1)}{T_{i}^{2}}(\bgamma^\top\Sigma_{i}\bgamma)^{2}-(\bgamma^\top\Sigma_{i}\bgamma)^{2} \\
    &\overset{T_{i}\rightarrow\infty}{\longrightarrow}& 0.
  \end{eqnarray*}
  It is assumed that the samples/subjects are independent, therefore,
    \[
      \mathbb{E}\left\{\bgamma^\top\bS_{i}\bgamma-\exp(\bx_{i}^\top\bbeta)\right\}\left\{\bgamma^\top\bS_{i'}\bgamma-\exp(\bx_{i'}^\top\bbeta)\right\}=0.
    \]
  Thus,
    \[
      \mathbb{E}\left\{\frac{1}{n}\sum_{i}\bgamma^\top\bS_{i}\bgamma-\frac{1}{n}\sum_{i}\exp(\bx_{i}^\top\bbeta)\right\}^{2}\rightarrow0, \text{ as } T_{\min}\rightarrow\infty.
    \]
    \begin{eqnarray*}
      && \mathbb{E}\left\{\left(\frac{1}{n}\sum_{i}\bgamma^\top\bS_{i}\bgamma\right)\left(\frac{1}{n}\sum_{i}\exp(\bx_{i}^\top\bbeta)\right)-\frac{1}{n}\sum_{i}(\bgamma^\top\bS_{i}\bgamma)\exp(\bx_{i}^\top\bbeta)\right\}^{2} \\
      &\leq& \mathbb{E}\left(\frac{1}{n}\sum_{i}\bgamma^\top\bS_{i}\bgamma\right)^{2}\left(\frac{1}{n}\sum_{i}\exp(\bx_{i}^\top\bbeta)\right)^{2}+\mathbb{E}\left\{\frac{1}{n}\sum_{i}(\bgamma^\top\bS_{i}\bgamma)\exp(\bx_{i}^\top\bbeta)\right\}^{2}.
    \end{eqnarray*}
    \begin{eqnarray*}
      && \mathbb{E}\left(\frac{1}{n}\sum_{i}\bgamma^\top\bS_{i}\bgamma\right)^{2} \\
      &=& \frac{1}{n^{2}}\sum_{i}\mathbb{E}(\bgamma^\top\bS_{i}\bgamma)^{2}+\frac{1}{n^{2}}\sum_{i\neq i'}\mathbb{E}(\bgamma^\top\bS_{i}\bgamma)(\bgamma^\top\bS_{i'}\bgamma) \\
      &=& \frac{1}{n^{2}}\sum_{i}\left\{\frac{1}{T_{i}^{2}}\mathbb{E}z_{itj_{i}}^{4}+\frac{1}{T_{i}^{2}}\sum_{t\neq s}\mathbb{E}z_{itj_{i}}^{2}z_{isj_{i}}^{2}\right\}+\frac{1}{n^{2}}\sum_{i\neq i'}\left(\frac{1}{T_{i}^{2}}\sum_{t=1}^{T_{i}}\mathbb{E}z_{itj_{i}}^{2}\right)\left(\frac{1}{T_{i'}^{2}}\sum_{t=1}^{T_{i'}}\mathbb{E}z_{i'tj_{i'}}^{2}\right) \\
      &=& \frac{1}{n^{2}}\sum_{i}\left\{\frac{1}{T_{i}}\mathbb{E}z_{i1j_{i}}^{4}+\frac{T_{i}(T_{i}-1)}{T_{i}^{2}}(\bgamma^\top\Sigma_{i}\bgamma)^{2}\right\}+\frac{1}{n^{2}}\sum_{i\neq i'}\left(\frac{1}{T_{i}}(\bgamma^\top\Sigma_{i}\bgamma)\right)\left(\frac{1}{T_{i'}}(\bgamma^\top\Sigma_{i'}\bgamma)\right) \\
      &\overset{T_{\min}\rightarrow\infty}{\longrightarrow}& \frac{1}{n^{2}}\sum_{i}(\bgamma^\top\Sigma_{i}\bgamma)^{2}.
    \end{eqnarray*}
    \begingroup
    \allowdisplaybreaks
    \begin{eqnarray*}
      && \mathbb{E}\left\{\frac{1}{n}\sum_{i}(\bgamma^\top\bS_{i}\bgamma)\exp(\bx_{i}^\top\bbeta)\right\}^{2} \\
      &=& \frac{1}{n^{2}}\sum_{i}\sum_{i}\mathbb{E}(\bgamma^\top\bS_{i}\bgamma)^{2}\exp^{2}(\bx_{i}^\top\bbeta)+\frac{1}{n^{2}}\sum_{i\neq i'}\mathbb{E}(\bgamma^\top\bS_{i}\bgamma)\exp(\bx_{i}^\top\bbeta)\mathbb{E}(\bgamma^\top\bS_{i'}\bgamma)\exp(\bx_{i'}^\top\bbeta) \\
      &=& \frac{1}{n^{2}}\sum_{i}\left\{\frac{1}{T_{i}}\mathbb{E}z_{itj_{i}}^{4}+\frac{T_{i}(T_{i}-1)}{T_{i}^{2}}(\bgamma^\top\Sigma_{i}\bgamma)^{2}\right\}(\bgamma^\top\Sigma_{i}\bgamma)^{2}+\frac{1}{n^{2}}\sum_{i\neq i'}(\bgamma^\top\Sigma_{i}\bgamma)^{2}(\bgamma^\top\Sigma_{i'}\bgamma)^{2} \\
      &\overset{T_{\min}\rightarrow\infty}{\longrightarrow}& \frac{1}{n^{2}}\sum_{i}(\bgamma^\top\Sigma_{i}\bgamma)^{4}+\frac{1}{n^{2}}\sum_{i\neq i'}(\bgamma^\top\Sigma_{i}\bgamma)^{2}(\bgamma^\top\Sigma_{i'}\bgamma)^{2}.
    \end{eqnarray*}
    \endgroup
    \begin{eqnarray*}
      && \mathbb{E}\left\{\left(\frac{1}{n}\sum_{i}\bgamma^\top\bS_{i}\bgamma\right)\left(\frac{1}{n}\sum_{i}\exp(\bx_{i}^\top\bbeta)\right)-\frac{1}{n}\sum_{i}(\bgamma^\top\bS_{i}\bgamma)\exp(\bx_{i}^\top\bbeta)\right\}^{2} \\
      &\leq& \mathbb{E}\left(\frac{1}{n}\sum_{i}\bgamma^\top\bS_{i}\bgamma\right)^{2}\left(\frac{1}{n}\sum_{i}\exp(\bx_{i}^\top\bbeta)\right)^{2}+\mathbb{E}\left\{\frac{1}{n}\sum_{i}(\bgamma^\top\bS_{i}\bgamma)\exp(\bx_{i}^\top\bbeta)\right\}^{2} \\
      &\overset{T_{\min}\rightarrow\infty}{\longrightarrow}& \frac{1}{n^{2}}\sum_{i}(\bgamma^\top\Sigma_{i}\bgamma)^{2}\left(\frac{1}{n}\sum_{i}(\bgamma^\top\Sigma_{i}\bgamma)\right)^{2}+\frac{1}{n^{2}}\sum_{i}(\bgamma^\top\Sigma_{i}\bgamma)^{4}+\frac{1}{n^{2}}\sum_{i\neq i'}(\bgamma^\top\Sigma_{i}\bgamma)^{2}(\bgamma^\top\Sigma_{i'}\bgamma)^{2}.
    \end{eqnarray*}
  The above quantity on the right is bounded by a constant from above. Therefore, as $T_{\min}\rightarrow\infty$,
    \[
      (\mu-\rho_{1}-\rho_{2}\mu)^{2}\rightarrow0.
    \]

    \[
      \left(\frac{\hat{\phi}^{2}}{\hat{\delta}^{2}}-\rho_{2}\right)^{2}=\left(\frac{\hat{\phi}^{2}}{\hat{\delta}^{2}}-\frac{{\phi}^{2}}{\hat{\delta}^{2}}\right)^{2}+\left(\frac{{\phi}^{2}}{\hat{\delta}^{2}}-\frac{\alpha}{\hat{\delta}^{2}}\right)^{2}+\left(\frac{\alpha}{\hat{\delta}^{2}}-\rho_{2}\right)^{2}.
    \]
    Since $\hat{\delta}^{4}$ is bounded,
    \[
      \mathbb{E}(\hat{\phi}^{2}-\phi^{2})^{2}\rightarrow0 \quad \Rightarrow \quad \mathbb{E}\left(\frac{\hat{\phi}^{2}}{\hat{\delta}^{2}}-\frac{{\phi}^{2}}{\hat{\delta}^{2}}\right)^{2}\rightarrow 0;
    \]
    \[
      \mathbb{E}(\phi^{2}-\alpha)^{2}\rightarrow 0 \quad \Rightarrow \quad \mathbb{E}\left(\frac{{\phi}^{2}}{\hat{\delta}^{2}}-\frac{\alpha}{\hat{\delta}^{2}}\right)^{2}\rightarrow 0.
    \]
  Let $\rho_{2}=\rho_{2}^{(1)}/\rho_{2}^{(2)}$, where
    \begin{eqnarray*}
      \rho_{2}^{(1)} &=& \frac{1}{n}\sum_{i}(\bgamma^\top\bS_{i}\bgamma)\exp(\bx_{i}^\top\bbeta)-\left(\frac{1}{n}\sum_{i}\bgamma^\top\bS_{i}\bgamma\right)\left(\frac{1}{n}\sum_{i}\exp(\bx_{i}^\top\bbeta)\right), \\
      \rho_{2}^{(2)} &=& \frac{1}{n}\sum_{i}(\bgamma^\top\bS_{i}\bgamma)^{2}-\left(\frac{1}{n}\sum_{i}\bgamma^\top\bS_{i}\bgamma\right)^{2}.
    \end{eqnarray*}
    \[
      \mathbb{E}\left(\alpha-\rho_{2}^{(1)}\right)^{2}=\left(\frac{1}{n}\sum_{i}\exp(\bx_{i}^\top\bbeta)\right)^{2}\mathbb{E}\left\{\frac{1}{n}\sum_{i}(\bgamma^\top\bS_{i}\bgamma)-\frac{1}{n}\sum_{i}\exp(\bx_{i}^\top\bbeta)\right\}^{2}\rightarrow 0.
    \]
    \begin{eqnarray*}
      \hat{\delta}^{2} &=& \frac{1}{n}\sum_{i=1}^{n}\left\{\bgamma^\top\bS_{i}\bgamma-\mu(\bgamma^\top\bgamma)\right\}^{2} \\
      &=& \frac{1}{n}\sum_{i=1}^{n}(\bgamma^\top\bS_{i}\bgamma)^{2}-2\left(\frac{1}{n}\sum_{i=1}^{n}\bgamma^\top\bS_{i}\bgamma\right)\left(\frac{1}{n}\sum_{i=1}^{n}\exp(\bx_{i}^\top\bbeta)\right)+\left(\frac{1}{n}\sum_{i=1}^{n}\exp(\bx_{i}^\top\bbeta)\right)^{2}.
    \end{eqnarray*}
  It can be concluded that as $T_{\min}\rightarrow\infty$,
    \[
      \mathbb{E}\left(\hat{\delta}-\rho_{2}^{2}\right)^{2}=\mathbb{E}\left\{\frac{1}{n}\sum_{i=1}^{n}(\bgamma^\top\bS_{i}\bgamma)-\frac{1}{n}\sum_{i=1}^{n}\exp(\bx_{i}^\top\bbeta)\right\}^{2}\rightarrow 0,
    \]
  and
    \[
      \mathbb{E}\left(\frac{\hat{\phi}^{2}}{\hat{\delta}^{2}}-\rho_{2}\right)^{2}\rightarrow 0.
    \]

    \[
      \mathbb{E}\left\{\frac{1}{n}\sum_{i=1}^{n}\|\bS_{i}^{*}-\Sigma_{i}^{**}\|^{2}\right\}\rightarrow 0, \quad \Rightarrow \quad \mathbb{E}\|\bS_{i}^{*}-\Sigma_{i}^{**}\|^{2}\rightarrow 0.
    \]
  This implies that
    \[
      \mathbb{E}\left\{\bgamma^\top\bS_{i}^{*}\bgamma-\exp(\bx_{i}^\top\bbeta)\right\}^{2}-\mathbb{E}\left\{\bgamma^\top\Sigma_{i}^{**}\bgamma-\exp(\bx_{i}^\top\bbeta)\right\}^{2}\rightarrow 0.
    \]
\end{proof}

\subsection{Proof of Theorem~\ref{thm:optimal2}}

\begin{proof}
  For the first statement,
  \begin{eqnarray*}
    && \lim_{T_{\min}\rightarrow\infty}\inf_{T_{i}\geq T_{\min}}\left[\frac{1}{n}\sum_{i=1}^{n}\mathbb{E}\left\{\bgamma^\top\hat{\Sigma}_{i}\bgamma-\exp(\bx_{i}^\top\bbeta)\right\}^{2}-\frac{1}{n}\sum_{i=1}^{n}\mathbb{E}\left\{\bgamma^\top\bS_{i}^{*}\bgamma-\exp(\bx_{i}^\top\bbeta)\right\}^{2}\right] \\
    &\geq& \inf\left[\frac{1}{n}\sum_{i=1}^{n}\mathbb{E}\left\{\bgamma^\top\hat{\Sigma}_{i}\bgamma-\exp(\bx_{i}^\top\bbeta)\right\}^{2}-\frac{1}{n}\sum_{i=1}^{n}\mathbb{E}\left\{\bgamma^\top\Sigma_{i}^{**}\bgamma-\exp(\bx_{i}^\top\bbeta)\right\}^{2}\right] \\
    && +\lim\left[\frac{1}{n}\sum_{i=1}^{n}\mathbb{E}\left\{\bgamma^\top\Sigma_{i}^{**}\bgamma-\exp(\bx_{i}^\top\bbeta)\right\}^{2}-\frac{1}{n}\sum_{i=1}^{n}\mathbb{E}\left\{\bgamma^\top\bS_{i}^{*}\bgamma-\exp(\bx_{i}^\top\bbeta)\right\}^{2}\right].
  \end{eqnarray*}
  By Theorem~\ref{thm:optimal1}, the second term on the right converges to zero, and the first term is $\geq 0$ by the definition of $\Sigma_{i}^{**}$.

  For the second statement,
  \begin{eqnarray*}
    && \lim_{T_{\min}\rightarrow\infty}\left[\frac{1}{n}\sum_{i=1}^{n}\mathbb{E}\left\{\bgamma^\top\hat{\Sigma}_{i}\bgamma-\exp(\bx_{i}^\top\bbeta)\right\}^{2}-\frac{1}{n}\sum_{i=1}^{n}\mathbb{E}\left\{\bgamma^\top\bS_{i}^{*}\bgamma-\exp(\bx_{i}^\top\bbeta)\right\}^{2}\right]=0 \\
    &\Leftrightarrow& \lim_{T_{\min}\rightarrow\infty}\left[\frac{1}{n}\sum_{i=1}^{n}\mathbb{E}\left\{\bgamma^\top\hat{\Sigma}_{i}\bgamma-\exp(\bx_{i}^\top\bbeta)\right\}^{2}-\frac{1}{n}\sum_{i=1}^{n}\mathbb{E}\left\{\bgamma^\top\Sigma_{i}^{**}\bgamma-\exp(\bx_{i}^\top\bbeta)\right\}^{2}\right]=0 \\
    &\Leftrightarrow& \lim_{T_{\min}\rightarrow\infty} \mathbb{E}\left\{\bgamma^\top\hat{\Sigma}_{i}\bgamma-\exp(\bx_{i}^\top\bbeta)\right\}^{2}-\mathbb{E}\left\{\bgamma^\top\Sigma_{i}^{**}\bgamma-\exp(\bx_{i}^\top\bbeta)\right\}^{2}=0 \\
    &\Leftrightarrow& \lim_{T_{\min}\rightarrow\infty}\mathbb{E}\|\bgamma^\top\hat{\Sigma}_{i}\bgamma-\bgamma^\top\Sigma_{i}^{**}\bgamma\|^{2}=0 \\
    &\Leftrightarrow& \lim_{T_{\min}\rightarrow\infty}\mathbb{E}\|\bgamma^\top\hat{\Sigma}_{i}\bgamma-\bgamma^\top\bS_{i}^{*}\bgamma\|^{2}=0 \\
    &\Leftrightarrow& \lim_{T_{\min}\rightarrow\infty}\mathbb{E}\|\hat{\Sigma}_{i}-\bS_{i}^{*}\|^{2}=0.
  \end{eqnarray*}
  This finishes the proof of this theorem.
\end{proof}

\subsection{\texorpdfstring{$\bS_{i}^{*}$}{} is well-conditioned}
\label{appendix:sub:Si_willcondition}

In this section, we show that the proposed estimator $\bS_{i}^{*}$ is well-conditioned and thus, invertible. This is achieved by two steps: for $i=1,\dots,n$, (1) prove that the largest eigenvalue of $\bS_{i}^{*}$ is bounded in probability; (2) prove that the smallest eigenvalue of $\bS_{i}^{*}$ is bounded away from zero in probability. The proof follows the same strategy as in \citet{ledoit2004well}, but considers the case with multiple covariance matrices.

The covariance matrix $\Sigma_{i}$ has the eigendecomposition as $\Sigma_{i}=\Pi_{i}\Lambda_{i}\Pi_{i}^\top$. Let $\bU_{i}=\Lambda_{i}^{-1/2}\bY_{i}$. Denote $\lambda_{\max}(\bA)$ and $\lambda_{\min}(\bA)$ as the maximum and minimum eigenvalue of a matrix $\bA$, respectively.
\begin{eqnarray*}
  \lambda_{\max}(\bS_{i}^{*}) &=& \lambda_{\max}\left(\frac{\hat{\psi}^{2}}{\hat{\delta}^{2}}\mu\boldsymbol{\mathrm{I}}+\frac{\hat{\phi}^{2}}{\hat{\delta}^{2}}\bS_{i}\right) \\
  &=& \frac{\hat{\psi}^{2}}{\hat{\delta}^{2}}\mu+\frac{\hat{\phi}^{2}}{\hat{\delta}^{2}}\lambda_{\max}(\bS_{i}).
\end{eqnarray*}
\[
  \mu=\frac{1}{n}\sum_{i=1}^{n}\exp(\bx_{i}^\top\bbeta)=\frac{1}{n}\sum_{i=1}^{n}\lambda_{ij_{i}}\leq\max_{i}\lambda_{\max}(\Lambda_{i}).
\]
\[
  \lambda_{\max}(\bS_{i})=\lambda_{\max}\left(\frac{1}{T_{i}}\Lambda^{1/2}\bU_{i}\bU_{i}^\top\Lambda_{i}^{1/2}\right)\leq\lambda_{\max}\left(\frac{1}{T_{i}}\bU_{i}\bU_{i}^\top\right)\lambda_{\max}(\Lambda_{i})\leq \lambda_{\max}\left(\frac{1}{T_{i}}\bU_{i}\bU_{i}^\top\right)\max_{i}\lambda_{\max}(\Lambda_{i}).
\]
Assume that $p/T_{\max}$ converges to a limit, denoted as $c$. Based on Assumption A1, $c\leq C_{1}$. Based on the results in \citet{yin1988limit}, as $T_{\min}=\min_{i}T_{i}\rightarrow\infty$, for $i=1,\dots,n$,
\[
  \lim~\lambda_{\max}\left(\frac{1}{T_{i}}\bU_{i}\bU_{i}^\top\right)= (1+\sqrt{c})^{2}, \quad \text{a.s.}
\]
This implies that
\[
  \mathbb{P}\left\{\lambda_{\max}(\bS_{i}^{*})\leq(1+\sqrt{c})^{2}\max_{i}\lambda_{\max}(\Lambda_{i})\right\}\rightarrow 1,
\]
and
\[
  \mathbb{P}\left\{\lambda_{\max}(\bS_{i}^{*})\leq(1+\sqrt{C_{1}})^{2}\max_{i}\lambda_{\max}(\Lambda_{i})\right\}\rightarrow 1.
\]
Therefore, if $p/T_{\max}$ converges to a constant, the largest eigenvalue of $\bS_{i}^{*}$ is bounded in probability. If $p/T_{\max}$ has no limit, under Assumption A1, there exists a subsequence such that $p/T_{\max}$ converges. Along this sequence, the largest eigenvalue of $\bS_{i}^{*}$ is bounded in probability. This is true for any converging sequence, and in addition, the upper bound is independent of the particular subsequence. As a result, it holds for the whole sequence.

Next, we show that the smallest eigenvalue of $\bS_{i}^{*}$ is bounded away from zero in probability. Analogously, we have
\[
  \lambda_{\min}(\bS_{i})=\lambda_{\min}\left(\frac{1}{T_{i}}\Lambda^{1/2}\bU_{i}\bU_{i}^\top\Lambda_{i}^{1/2}\right)\geq\lambda_{\min}\left(\frac{1}{T_{i}}\bU_{i}\bU_{i}^\top\right)\lambda_{\min}(\Lambda_{i})\geq \lambda_{\min}\left(\frac{1}{T_{i}}\bU_{i}\bU_{i}^\top\right)\min_{i}\lambda_{\min}(\Lambda_{i}).
\]
First, assume $p/T_{\max}$ converges to a constant $c$. If $c\in(0,1)$, based on the results in \citet{bai1993limit}, 
\[
  \lim~\lambda_{\min}\left(\frac{1}{T_{i}}\bU_{i}\bU_{i}^\top\right)=(1-\sqrt{c})^{2}, \quad \text{a.s.}
\]
Assume $c\leq1-\kappa$ for some $\kappa\in(0,1)$. One can conclude that
\[
  \mathbb{P}\left\{\lambda_{\min}(\bS_{i}^{*})\geq(1-\sqrt{1-\kappa})^{2}\min_{i}\lambda_{\min}(\Lambda_{i})\right\}\rightarrow 1.
\]
When $c>1-\kappa$, we propose to identify a lower bound from the following
\[
  \lambda_{\min}(\bS_{i}^{*})=\lambda_{\min}\left(\frac{\hat{\psi}^{2}}{\hat{\delta}^{2}}\mu\boldsymbol{\mathrm{I}}+\frac{\hat{\phi}^{2}}{\hat{\delta}^{2}}\bS_{i}\right)\geq \frac{\hat{\psi}^{2}}{\hat{\delta}^{2}}\mu.
\]
Compare the right-hand side in the above to it population counterpart,
\[
  \frac{\hat{\psi}^{2}}{\hat{\delta}^{2}}\mu-\frac{\psi^{2}}{\delta^{2}}\mu=\mu\left\{\frac{\hat{\psi}^{2}-\psi^{2}}{\delta^{2}}+\hat{\psi}^{2}\left(\frac{1}{\hat{\delta}^{2}}-\frac{1}{\delta^{2}}\right)\right\}.
\]
From Lemmas \ref{lemma:shpar_bound} and \ref{lemma:shpar_consist}, we can show that the above converges to zero in probability. First, consider $\psi^{2}=\sum_{i=1}^{n}\psi_{i}^{2}/n$, where $\psi_{i}^{2}=\mathbb{E}\{\bgamma^\top(\bS_{i}-\Sigma_{i})\bgamma\}^{2}$. From the proof of Lemma~\ref{lemma:shpar_bound},
\begin{eqnarray*}
  \mathbb{E}\|\bS_{i}-\Sigma_{i}\|^{2} &=& \frac{1}{pT_{i}}\sum_{j=1}^{p}\sum_{k=1}^{p}\mathbb{E}(z_{i1j}^{2}z_{i1k}^{2})-\frac{1}{pT_{i}}\sum_{j=1}^{p}\sum_{k=1}^{p}\lambda_{ijk}^{2} \\
  &=& \frac{p}{T_{i}}\left\{\frac{1}{p^{2}}\sum_{j=1}^{p}\sum_{k=1}^{p}\mathbb{E}(z_{i1j}^{2}z_{i1k}^{2})\right\}-\frac{1}{pT_{i}}\sum_{j=1}^{p}\lambda_{ijj}^{2}.
\end{eqnarray*}
As $T_{\min}\rightarrow\infty$, the second term on the right-hand side converges to zero. For $\epsilon>0$, there exists a constant $M>0$ such that when $T_{\min}>M$, $\sum_{j=1}^{p}\lambda_{ijj}^{2}/(pT_{i})<\epsilon$. Thus, $\psi_{i}^{2}\geq(1-\kappa)-\epsilon$ and $\psi^{2}\geq(1-\kappa)-\epsilon$.
\[
  \lambda_{\min}(\bS_{i}^{*})\geq\frac{\hat{\psi}^{2}}{\hat{\delta}^{2}}\mu=\frac{\psi^{2}}{\delta^{2}}\mu+\left(\frac{\hat{\psi}^{2}}{\hat{\delta}^{2}}\mu-\frac{\psi^{2}}{\delta^{2}}\mu\right)\geq \frac{\psi^{2}}{\delta^{2}}\mu-\epsilon\geq\frac{\psi^{2}}{2C_{2}+C_{1}\sqrt{C_{2}}}-\epsilon\geq\frac{(1-\kappa)-\epsilon}{2C_{2}+C_{1}\sqrt{C_{2}}}-\epsilon.
\]
For a choice of $\epsilon$, we have
\[
  \mathbb{P}\left\{\lambda_{\min}(\bS_{i}^{*})\geq\frac{1-\kappa}{2(2C_{2}+C_{1}\sqrt{C_{2}})}\right\}\rightarrow 1.
\]
Therefore, for both $c\leq 1-\kappa$ and $c>1-\kappa$, the smallest eigenvalue of $\bS_{i}^{*}$ is bounded away from zero. Analogous to the proof of the largest eigenvalue, for the case that $p/T_{\max}$ does not have a limit, we can also have the conclusion for the whole sequence. Since both the largest and the smallest eigenvalues are bounded, $\bS_{i}^{*}$ is well-conditioned and invertible.

\subsection{Proof of Lemma~\ref{lemma:beta_equiv} and Theorem~\ref{thm:beta_consist}}

We first proof Lemma~\ref{lemma:beta_equiv}.
\begin{proof}
  \[
    \mathbb{E}(\bgamma^\top\Sigma_{i}^{*}\bgamma)=\frac{\psi^{2}}{\delta^{2}}\mu(\bgamma^\top\bgamma)+\frac{\phi^{2}}{\delta^{2}}\mathbb{E}(\bgamma^\top\bS_{i}\bgamma)=\frac{\psi^{2}}{\delta^{2}}\mu(\bgamma^\top\bgamma)+\frac{\phi^{2}}{\delta^{2}}\exp(\bx_{i}^\top\bbeta)=\exp(\bx_{i}^\top\bbeta^{*}).
  \]
  \[
    \frac{\sum_{i}\exp(\bx_{i}^\top\bbeta^{*})/n}{\sum_{i}\exp(\bx_{i}^\top\bbeta)/n}=\frac{\psi^{2}}{\delta^{2}}\frac{\mu(\bgamma^\top\bgamma)}{\sum_{i}\exp(\bx_{i}^\top\bbeta)/n}+\frac{\phi^{2}}{\delta^{2}}=\frac{\psi^{2}}{\delta^{2}}+\frac{\phi^{2}}{\delta^{2}}=1.
  \]
  \[
    \Rightarrow \quad \frac{1}{n}\sum_{i=1}^{n}\exp(\bx_{i}^\top\bbeta^{*})=\frac{1}{n}\sum_{i=1}^{n}\exp(\bx_{i}^\top\bbeta).
  \]
  Therefore,
    \[
      \bbeta^{*}=\bbeta.
    \]
\end{proof}

Next, we prove that the proposed estimator $\bbeta$ is a consistent estimator (Theorem~\ref{thm:beta_consist}).
\begin{proof}
  Using the consistency of pseudo-likelihood estimator~\citep{gong1981pseudo} and the conclusion in Lemma~\ref{lemma:beta_equiv}, $\hat{\bbeta}$ is a consistent estimator of $\bbeta$.
\end{proof}


\section{Additional Simulation Results}
\label{appendix:sec:sim}

\subsection{\texorpdfstring{$\bgamma$}{} unknown}
\label{appendix:sub:sim_gamma_unknown}

Here, we present the performance of estimating the fourth dimension (D4) when $\bgamma$ is unknown (Figure~\ref{appendix:fig:sim_gammaunknown_asmp}). From the figures, as $n$ and $T$ increase, the estimate of the covariance matrices, the projection and the model coefficient converge to the truth.
\begin{figure}[!ph]
  \begin{center}
    \subfloat[Bias of $\hat{\beta}_{1}$]{\includegraphics[width=0.3\textwidth]{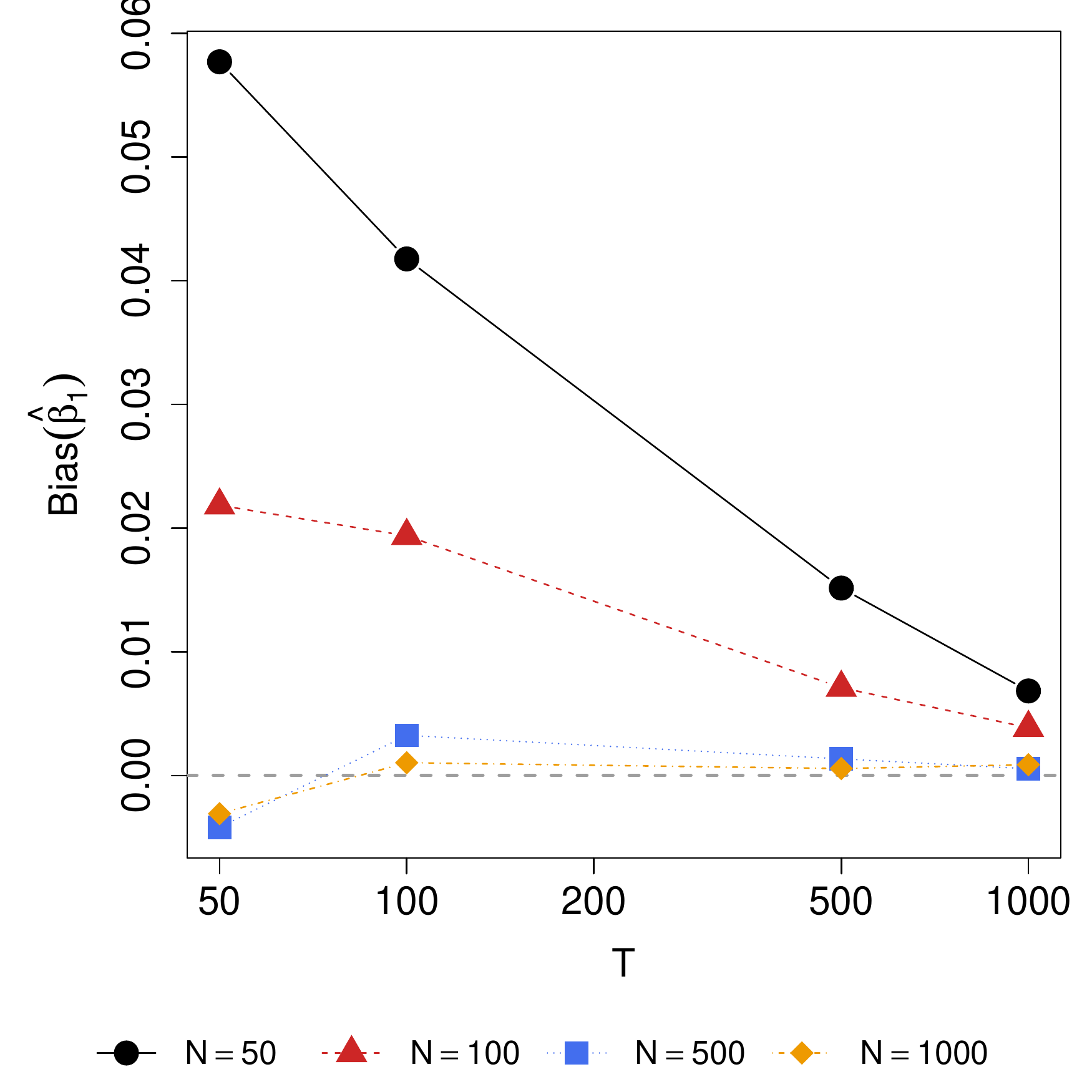}}
    \enskip{}
    \subfloat[MSE of $\hat{\beta}_{1}$]{\includegraphics[width=0.3\textwidth]{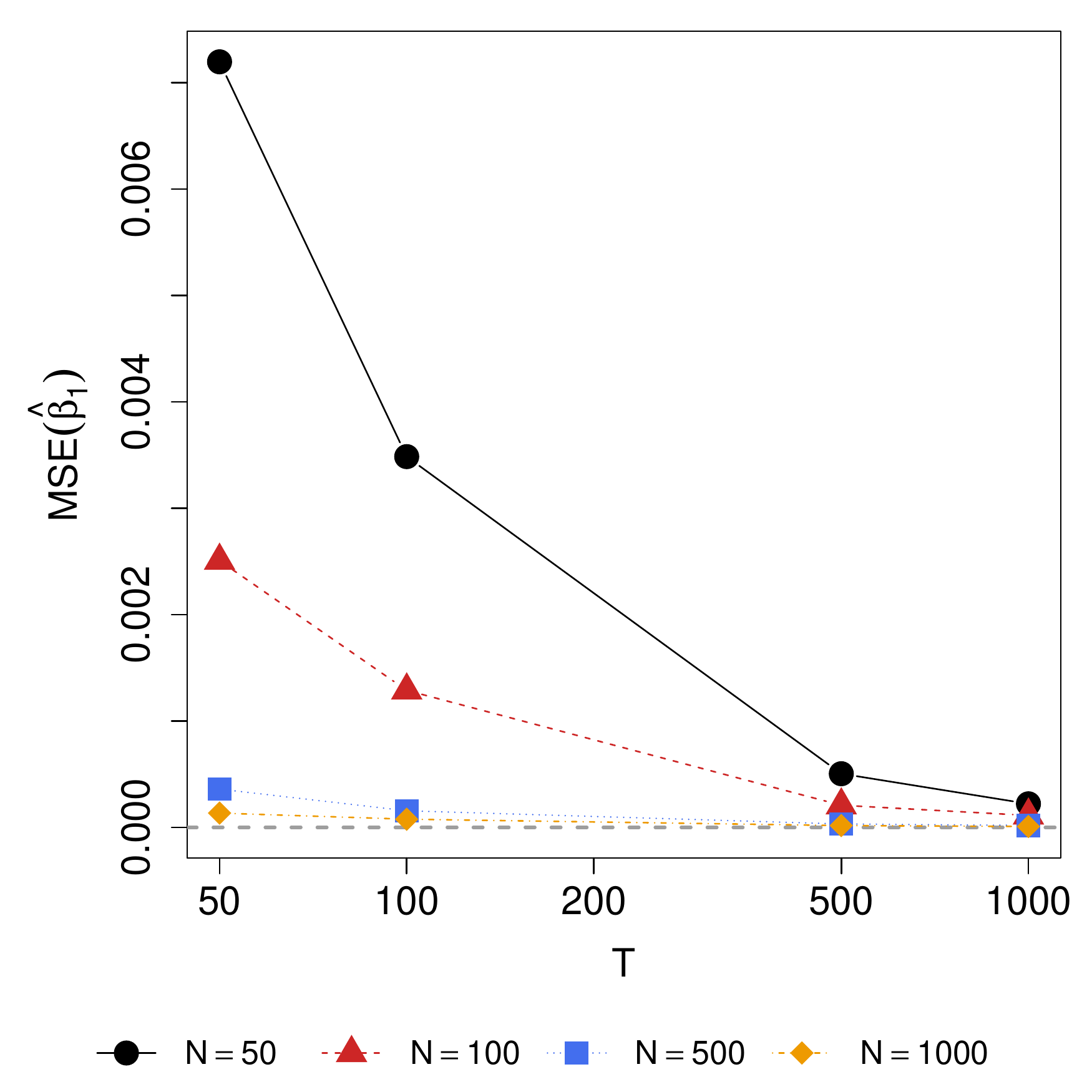}}
    \enskip{}
    \subfloat[Coverage probability of $\hat{\beta}_{1}$]{\includegraphics[width=0.3\textwidth]{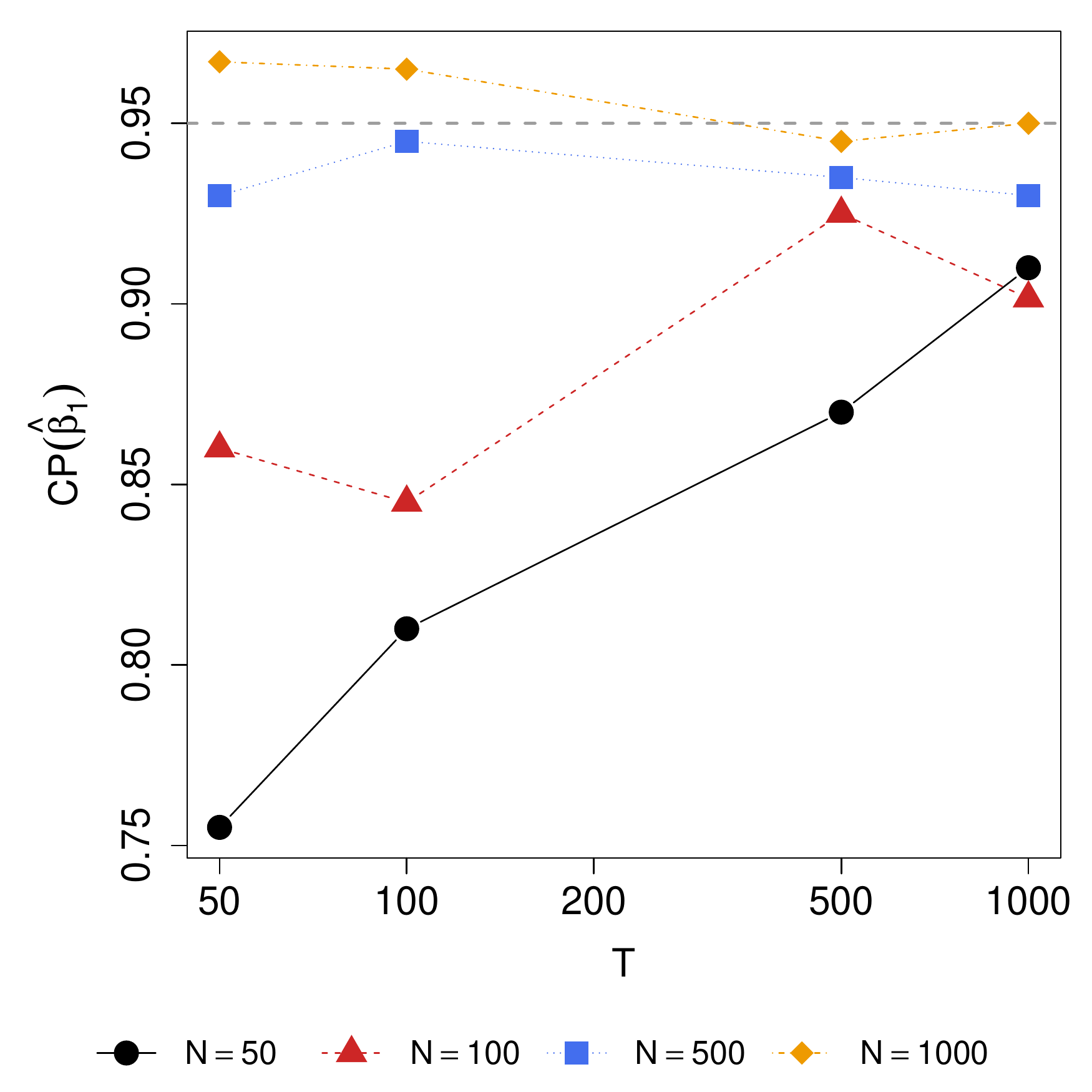}}

    \subfloat[MSE of $\hat{\lambda}_{ij}$]{\includegraphics[width=0.3\textwidth]{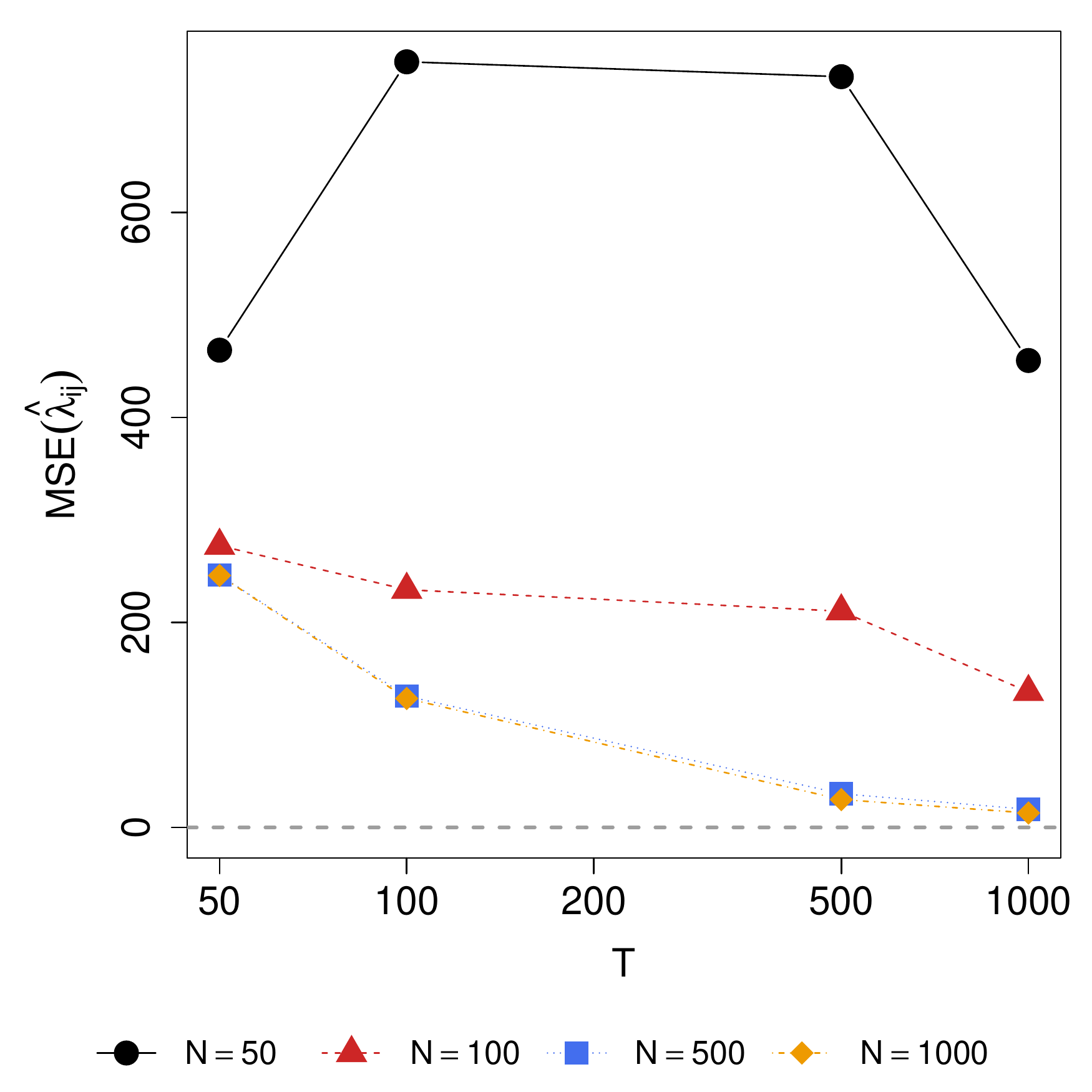}}
    \enskip{}
    \subfloat[Similarity metric $|\langle\hat{\bgamma},\bpi_{4}\rangle|$]{\includegraphics[width=0.3\textwidth]{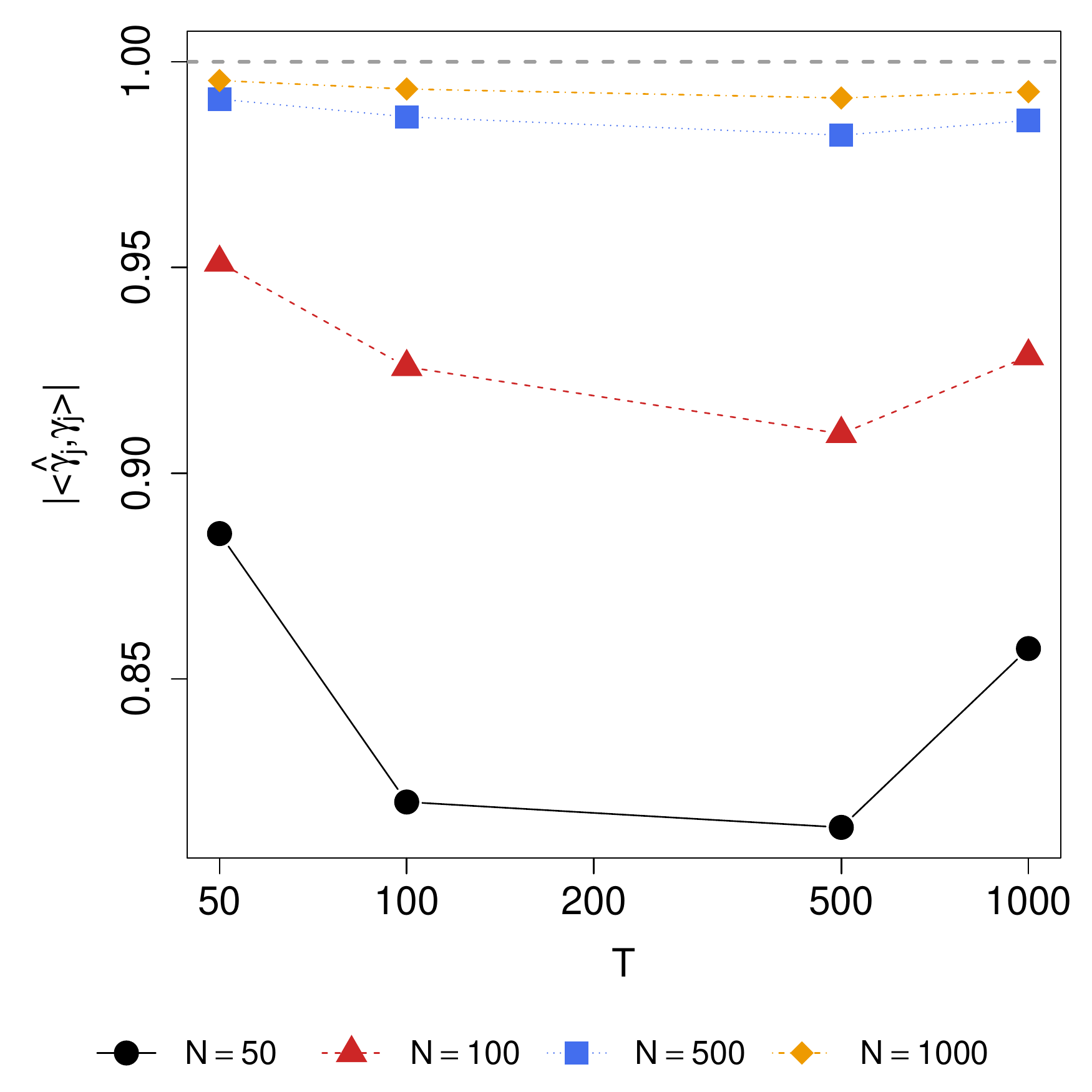}}
  \end{center}
  \caption{\label{appendix:fig:sim_gammaunknown_asmp}Estimation performance of PS-CAP in estimating the fourth dimension (D4) when $\bgamma$ is unknown. For $\hat{\beta}_{1}$, (a) bias, (b) mean squared error (MSE) and (c) coverage probability (CP) are presented, where CP is obtained from 500 bootstrap samples. For the eigenvalues $\hat{\lambda}_{ij}$, (d) MSE is presented. For $\hat{\bgamma}$, (e) similarity to $\bpi_{4}$ is presented. Data dimension $p=100$. Sample sizes vary from $n=50,100,500,100$ and $T_{i}=T=50,100,500,1000$.}
\end{figure}

\bibliographystyle{apalike}
\bibliography{Bibliography}

\begin{thebibliography}{}

\bibitem[Anderson, 1973]{anderson1973asymptotically}
Anderson, T. (1973).
\newblock Asymptotically efficient estimation of covariance matrices with
  linear structure.
\newblock {\em The Annals of Statistics}, 1(1):135--141.

\bibitem[Anderson, 1963]{anderson1963asymptotic}
Anderson, T.~W. (1963).
\newblock Asymptotic theory for principal component analysis.
\newblock {\em The Annals of Mathematical Statistics}, 34(1):122--148.

\bibitem[Badhwar et~al., 2017]{badhwar2017resting}
Badhwar, A., Tam, A., Dansereau, C., Orban, P., Hoffstaedter, F., and Bellec,
  P. (2017).
\newblock Resting-state network dysfunction in {Alzheimer's} disease: a
  systematic review and meta-analysis.
\newblock {\em Alzheimer's \& Dementia: Diagnosis, Assessment \& Disease
  Monitoring}, 8:73--85.

\bibitem[Bai and Yin, 1993]{bai1993limit}
Bai, Z. and Yin, Y. (1993).
\newblock Limit of the smallest eigenvalue of a large dimensional sample
  covariance matrix.
\newblock {\em The annals of Probability}, pages 1275--1294.

\bibitem[Boik, 2002]{boik2002spectral}
Boik, R.~J. (2002).
\newblock Spectral models for covariance matrices.
\newblock {\em Biometrika}, 89(1):159--182.

\bibitem[Cai et~al., 2016]{cai2016structured}
Cai, T.~T., Ren, Z., and Zhou, H.~H. (2016).
\newblock Estimating structured high-dimensional covariance and precision
  matrices: Optimal rates and adaptive estimation.
\newblock {\em Electronic Journal of Statistics}, 10(1):1--59.

\bibitem[Chen et~al., 2011]{chen2011robust}
Chen, Y., Wiesel, A., and Hero, A.~O. (2011).
\newblock Robust shrinkage estimation of high-dimensional covariance matrices.
\newblock {\em IEEE Transactions on Signal Processing}, 59(9):4097--4107.

\bibitem[Chiu et~al., 1996]{chiu1996matrix}
Chiu, T.~Y., Leonard, T., and Tsui, K.-W. (1996).
\newblock The matrix-logarithmic covariance model.
\newblock {\em Journal of the American Statistical Association},
  91(433):198--210.

\bibitem[Corder et~al., 1993]{corder1993gene}
Corder, E.~H., Saunders, A.~M., Strittmatter, W.~J., Schmechel, D.~E., Gaskell,
  P.~C., Small, G., Roses, A.~D., Haines, J., and Pericak-Vance, M.~A. (1993).
\newblock Gene dose of apolipoprotein {E} type 4 allele and the risk of
  {Alzheimer's} disease in late onset families.
\newblock {\em Science}, 261(5123):921--923.

\bibitem[Daniels and Kass, 2001]{daniels2001shrinkage}
Daniels, M.~J. and Kass, R.~E. (2001).
\newblock Shrinkage estimators for covariance matrices.
\newblock {\em Biometrics}, 57(4):1173--1184.

\bibitem[De~Marco and Venneri, 2017]{de2017apoe}
De~Marco, M. and Venneri, A. (2017).
\newblock {ApoE}-dependent differences in functional connectivity support
  memory performance in early-stage {Alzheimer’s} disease (p4. 094).
\newblock {\em Neurology}, 88(16 Supplement).

\bibitem[Flury, 1984]{flury1984common}
Flury, B.~N. (1984).
\newblock Common principal components in k groups.
\newblock {\em Journal of the American Statistical Association},
  79(388):892--898.

\bibitem[Fox and Dunson, 2015]{fox2015bayesian}
Fox, E.~B. and Dunson, D.~B. (2015).
\newblock Bayesian nonparametric covariance regression.
\newblock {\em Journal of Machine Learning Research}, 16:2501--2542.

\bibitem[Franks and Hoff, 2019]{franks2019shared}
Franks, A.~M. and Hoff, P. (2019).
\newblock Shared subspace models for multi-group covariance estimation.
\newblock {\em Journal of Machine Learning Research}, 20(171):1--37.

\bibitem[Gong and Samaniego, 1981]{gong1981pseudo}
Gong, G. and Samaniego, F.~J. (1981).
\newblock Pseudo maximum likelihood estimation: theory and applications.
\newblock {\em The Annals of Statistics}, pages 861--869.

\bibitem[Gour et~al., 2014]{gour2014functional}
Gour, N., Felician, O., Didic, M., Koric, L., Gueriot, C., Chanoine, V.,
  Confort-Gouny, S., Guye, M., Ceccaldi, M., and Ranjeva, J.~P. (2014).
\newblock Functional connectivity changes differ in early and late-onset
  {Alzheimer's} disease.
\newblock {\em Human Brain Mapping}, 35(7):2978--2994.

\bibitem[Gour et~al., 2011]{gour2011basal}
Gour, N., Ranjeva, J.-P., Ceccaldi, M., Confort-Gouny, S., Barbeau, E.,
  Soulier, E., Guye, M., Didic, M., and Felician, O. (2011).
\newblock Basal functional connectivity within the anterior temporal network is
  associated with performance on declarative memory tasks.
\newblock {\em Neuroimage}, 58(2):687--697.

\bibitem[Grosenick et~al., 2013]{grosenick2013interpretable}
Grosenick, L., Klingenberg, B., Katovich, K., Knutson, B., and Taylor, J.~E.
  (2013).
\newblock Interpretable whole-brain prediction analysis with {GraphNet}.
\newblock {\em NeuroImage}, 72:304--321.

\bibitem[Hoff, 2009]{hoff2009hierarchical}
Hoff, P.~D. (2009).
\newblock A hierarchical eigenmodel for pooled covariance estimation.
\newblock {\em Journal of the Royal Statistical Society: Series B (Statistical
  Methodology)}, 71(5):971--992.

\bibitem[Hoff and Niu, 2012]{hoff2012covariance}
Hoff, P.~D. and Niu, X. (2012).
\newblock A covariance regression model.
\newblock {\em Statistica Sinica}, 22(2):729--753.

\bibitem[Johnstone and Lu, 2009]{johnstone2009consistency}
Johnstone, I.~M. and Lu, A.~Y. (2009).
\newblock On consistency and sparsity for principal components analysis in high
  dimensions.
\newblock {\em Journal of the American Statistical Association}, 104(486).

\bibitem[Koch et~al., 2012]{koch2012diagnostic}
Koch, W., Teipel, S., Mueller, S., Benninghoff, J., Wagner, M., Bokde, A.~L.,
  Hampel, H., Coates, U., Reiser, M., and Meindl, T. (2012).
\newblock Diagnostic power of default mode network resting state {fMRI} in the
  detection of {Alzheimer's} disease.
\newblock {\em Neurobiology of Aging}, 33(3):466--478.

\bibitem[Ledoit and Wolf, 2004]{ledoit2004well}
Ledoit, O. and Wolf, M. (2004).
\newblock A well-conditioned estimator for large-dimensional covariance
  matrices.
\newblock {\em Journal of Multivariate Analysis}, 88(2):365--411.

\bibitem[Ledoit and Wolf, 2012]{ledoit2012nonlinear}
Ledoit, O. and Wolf, M. (2012).
\newblock Nonlinear shrinkage estimation of large-dimensional covariance
  matrices.
\newblock {\em The Annals of Statistics}, 40(2):1024--1060.

\bibitem[Pascal et~al., 2014]{pascal2014generalized}
Pascal, F., Chitour, Y., and Quek, Y. (2014).
\newblock Generalized robust shrinkage estimator and its application to {STAP}
  detection problem.
\newblock {\em IEEE Transactions on Signal Processing}, 62(21):5640--5651.

\bibitem[Pourahmadi et~al., 2007]{pourahmadi2007simultaneous}
Pourahmadi, M., Daniels, M.~J., and Park, T. (2007).
\newblock Simultaneous modelling of the {Cholesky} decomposition of several
  covariance matrices.
\newblock {\em Journal of Multivariate Analysis}, 98(3):568--587.

\bibitem[Safieh et~al., 2019]{safieh2019apoe4}
Safieh, M., Korczyn, A.~D., and Michaelson, D.~M. (2019).
\newblock {ApoE4}: an emerging therapeutic target for {Alzheimer’s} disease.
\newblock {\em BMC Medicine}, 17(1):1--17.

\bibitem[Seiler and Holmes, 2017]{seiler2017multivariate}
Seiler, C. and Holmes, S. (2017).
\newblock Multivariate heteroscedasticity models for functional brain
  connectivity.
\newblock {\em Frontiers in Neuroscience}, 11:696.

\bibitem[Smith et~al., 2004]{smith2004advances}
Smith, S.~M., Jenkinson, M., Woolrich, M.~W., Beckmann, C.~F., Behrens, T.~E.,
  Johansen-Berg, H., Bannister, P.~R., De~Luca, M., Drobnjak, I., Flitney,
  D.~E., et~al. (2004).
\newblock Advances in functional and structural {MR} image analysis and
  implementation as {FSL}.
\newblock {\em NeuroImage}, 23:S208--S219.

\bibitem[Tibshirani et~al., 2005]{tibshirani2005sparsity}
Tibshirani, R., Saunders, M., Rosset, S., Zhu, J., and Knight, K. (2005).
\newblock Sparsity and smoothness via the fused lasso.
\newblock {\em Journal of the Royal Statistical Society: Series B (Statistical
  Methodology)}, 67(1):91--108.

\bibitem[Tyler, 1987]{tyler1987distribution}
Tyler, D.~E. (1987).
\newblock A distribution-free {M}-estimator of multivariate scatter.
\newblock {\em The Annals of Statistics}, pages 234--251.

\bibitem[Yin et~al., 1988]{yin1988limit}
Yin, Y.-Q., Bai, Z.-D., and Krishnaiah, P.~R. (1988).
\newblock On the limit of the largest eigenvalue of the large dimensional
  sample covariance matrix.
\newblock {\em Probability theory and related fields}, 78(4):509--521.

\bibitem[Zhao et~al., 2020]{zhao2020sparse}
Zhao, Y., Lindquist, M.~A., and Caffo, B.~S. (2020).
\newblock Sparse principal component based high-dimensional mediation analysis.
\newblock {\em Computational Statistics \& Data Analysis}, 142:106835.

\bibitem[Zhao et~al., 2019]{zhao2019covariate}
Zhao, Y., Wang, B., Mostofsky, S., Caffo, B., and Luo, X. (2019).
\newblock Covariate assisted principal regression for covariance matrix
  outcomes.
\newblock {\em Biostatistics}.

\bibitem[Zou et~al., 2006]{zou2006sparse}
Zou, H., Hastie, T., and Tibshirani, R. (2006).
\newblock Sparse principal component analysis.
\newblock {\em Journal of computational and graphical statistics},
  15(2):265--286.

\bibitem[Zou et~al., 2017]{zou2017covariance}
Zou, T., Lan, W., Wang, H., and Tsai, C.-L. (2017).
\newblock Covariance regression analysis.
\newblock {\em Journal of the American Statistical Association},
  112(517):266--281.

\end{thebibliography}

\end{document}